\documentclass[sigconf, nonacm, pdfa]{acmart}
\usepackage{hyperref} 
\usepackage{booktabs} 
\usepackage{color} 
\usepackage{balance} 
\usepackage{url}
\usepackage{tabularx}
\usepackage{siunitx} 
\usepackage{epstopdf} 
\usepackage{epsfig,tabularx,subfigure,multirow, graphicx}
\usepackage{enumitem}
\usepackage{caption}
\usepackage{bbding}
\usepackage{amsthm,amsmath}  

\newcolumntype{L}[1]{>{\raggedright\arraybackslash}p{#1}}
\newcolumntype{C}[1]{>{\centering\arraybackslash}p{#1}}
\newcolumntype{R}[1]{>{\raggedleft\arraybackslash}p{#1}}

\usepackage[linesnumbered,ruled,vlined]{algorithm2e}
\SetKwRepeat{Do}{do}{while}
\SetCommentSty{mycommfont}

\long\def\comment#1{}

\newcommand{\nop}[1]{}

\newtheorem{theorem}{\bf Theorem}[section]
\newtheorem{lemma}{\bf Lemma}[section]

\newtheorem{example}{\bf Example}

\theoremstyle{remark}

\theoremstyle{definition}
\newtheorem{definition}{\bf Definition}

\usepackage[a-2b]{pdfx}
\newcommand\vldbdoi{10.14778/3725688.3725715}
\newcommand\vldbpages{1905 - 1918}
\newcommand\vldbvolume{18}
\newcommand\vldbissue{6}
\newcommand\vldbyear{2025}
\newcommand\vldbauthors{\authors}
\newcommand\vldbtitle{\shorttitle} 
\newcommand\vldbavailabilityurl{https://github.com/dulei715/DynamicWEventCode}
\newcommand\vldbpagestyle{empty}

\let\oldequation\equation
\let\oldendequation\endequation
\renewenvironment{equation}{
	\begingroup
	\scriptsize 
	\oldequation
}{
	\oldendequation
	\endgroup
}

\newcommand{\entity}[1]{\mathcal{#1}}
\newcommand{\algvar}[1]{\mathcal{#1}}
\newcommand{\constvar}[1]{\mathbb{#1}}

\newcommand{\vectorfont}[1]{\boldsymbol{#1}}

\newcommand{\problemDefineSimpleName}{PWEPP-IDS}
\newcommand{\problemDefineTotalName}{Personalized $w$-Event Private Publishing for Infinite Data Streams}

\newcommand{\privacyDefineSimpleName}{($\vectorfont{w}$,$\vectorfont{\epsilon}$)-EPDP}
\newcommand{\privacyDefineTotalName}{$\vectorfont{w}$-Event $\vectorfont{\epsilon}$-Personalized Differential Privacy}

\newcommand{\solutionA}{PWSM}
\newcommand{\solutionATotalName}{Personalized Window Size Mechanism}

\newcommand{\solutionMethodA}{PBD}
\newcommand{\solutionMethodATotalName}{Personalized Budget Distribution}
\newcommand{\solutionMethodB}{PBA}
\newcommand{\solutionMethodBTotalName}{Personalized Budget Absorption}

\newcommand{\solutionCMPPLDPU}{PLBU}
\newcommand{\solutionCMPPLDPUTotalName}{Personalized LDP Budget Uniform}

\newcommand{\solutionCMPA}{BD}
\newcommand{\solutionCMPATotalName}{Budget Distribution}
\newcommand{\solutionCMPB}{BA}
\newcommand{\solutionCMPBTotalName}{Budget Absorption}
\newcommand{\solutionCMPLDPU}{LBU}
\newcommand{\solutionCMPLDPUTotalName}{LDP Budget Uniform}

\newcommand{\checkInDatasetName}{Foursquare}
\newcommand{\trajectoryDatasetName}{Taxi}
\newcommand{\tlnsDatasetName}{TLNS}
\newcommand{\sinDatasetName}{Sin}
\newcommand{\logDatasetName}{Log}

\newcounter{ruleCounter}
\begin{document}

\title{Infinite Stream Estimation under Personalized $w$-Event Privacy}

\author{Leilei Du}
\affiliation{%
	\institution{Hunan University}
	\city{Changsha}
	\country{China}
}
\email{leileidu@hnu.edu.cn}

\author{Peng Cheng}
\orcid{0000-0002-9797-6944}
\affiliation{%
	\institution{Tongji University}
	\city{Shanghai}
	\country{China}
}
\email{cspcheng@tongji.edu.cn}

\author{Lei Chen}
\affiliation{%
	\institution{HKUST (GZ) \& HKUST}
	\city{Guangzhou \& HK SAR}
	\country{China}
}
\email{leichen@cse.ust.hk}

\author{Heng Tao Shen}
\affiliation{%
	\institution{Tongji University \& UESTC}
	\city{Shanghai \& Chengdu}
	\country{China}
}
\email{shenhengtao@hotmail.com}

\author{Xuemin Lin}
\affiliation{%
	\institution{Shanghai Jiaotong University}
	\city{Shanghai}
	\country{China}
}
\email{xuemin.lin@gmail.com}

\author{Wei Xi}
\affiliation{%
	\institution{Xi'an Jiaotong University}
	\city{Xi'an}
	\country{China}
}
\email{xiwei@xjtu.edu.cn}

\sloppy

\begin{abstract}
	Streaming data collection is indispensable for stream data analysis, such as event monitoring. However, publishing these data directly leads to privacy leaks. $w$-event privacy is a valuable  tool to protect individual privacy within a given time window while maintaining high accuracy in data collection.
	Most existing $w$-event privacy studies on infinite data stream only focus on homogeneous privacy requirements for all users. In this paper, we propose personalized $w$-event privacy protection that allows different users to have different privacy requirements in private data stream estimation. Specifically, we design a mechanism that allows users to maintain constant privacy requirements at each time slot, namely \solutionATotalName{} (\solutionA{}). Then, we propose two solutions to accurately estimate stream data statistics while achieving \privacyDefineTotalName{} (\privacyDefineSimpleName{}), namely \solutionMethodATotalName{} (\solutionMethodA{}) and \solutionMethodBTotalName{} (\solutionMethodB{}). 
	\solutionMethodA{} always provides at least the same privacy budget for the next time step as the amount consumed in the previous release.
	\solutionMethodB{} fully absorbs the privacy budget from the previous $k$ time slots, while also borrowing from the privacy budget of the next $k$ time slots, to increase the privacy budget for the current time slot.
	We prove that both \solutionMethodA{} and \solutionMethodB{} outperform the state-of-the-art private stream estimation methods while satisfying the privacy requirements of all users.
	We demonstrate the efficiency and effectiveness of our \solutionMethodA{} and \solutionMethodB{} on both real and synthetic datasets, compared with the recent uniformity $w$-event approaches, \solutionCMPATotalName{} (\solutionCMPA{}) and  \solutionCMPBTotalName{} (\solutionCMPB). Our \solutionMethodA{} achieves $68\%$ less error than \solutionCMPA{} on average on real datasets. 
	Besides, our \solutionMethodB{} achieves $24.9\%$ less error than \solutionCMPB{} on average on synthetic datasets.
\end{abstract}

\maketitle

\pagestyle{\vldbpagestyle}
\begingroup\small\noindent\raggedright\textbf{PVLDB Reference Format:}\\
\vldbauthors. \vldbtitle. PVLDB, \vldbvolume(\vldbissue): \vldbpages, \vldbyear.\\
\href{https://doi.org/\vldbdoi}{doi:\vldbdoi}
\endgroup
\begingroup
\renewcommand\thefootnote{}\footnote{\noindent
	This work is licensed under the Creative Commons BY-NC-ND 4.0 International License. Visit \url{https://creativecommons.org/licenses/by-nc-nd/4.0/} to view a copy of this license. For any use beyond those covered by this license, obtain permission by emailing \href{mailto:info@vldb.org}{info@vldb.org}. Copyright is held by the owner/author(s). Publication rights licensed to the VLDB Endowment. \\
	\raggedright Proceedings of the VLDB Endowment, Vol. \vldbvolume, No. \vldbissue\ %
	ISSN 2150-8097. \\
	\href{https://doi.org/\vldbdoi}{doi:\vldbdoi} \\
}\addtocounter{footnote}{-1}\endgroup

\ifdefempty{\vldbavailabilityurl}{}{
	\vspace{.3cm}
	\begingroup\small\noindent\raggedright\textbf{PVLDB Artifact Availability:}\\
	The source code, data, and/or other artifacts have been made available at \url{\vldbavailabilityurl}.
	\endgroup
}

\section{Introduction}
With the popularity of smart devices and high-quality wireless networks, people can easily access the internet and utilize online services. They continuously report data to platforms and receive services like log stream analysis~\cite{DBLP:conf/www/XiePMJM23}, event monitoring~\cite{DBLP:journals/tpds/GuoJZZ12}, and video querying~\cite{DBLP:conf/cvpr/MoonHPPH23}. To provide better services, these platforms collect data and conduct real-time analysis over aggregated data streams.

However, collecting stream data directly poses severe privacy risks, causing users to refuse communication with platforms. For instance, an AIDS patient may decline to participate in an investigation due to privacy concerns~\cite{feijoo2023exploring}. To resolve this conflict, differential privacy (DP) is proposed  to protect individual privacy while ensuring accurate data estimation~\cite{DBLP:conf/icalp/Dwork06}.

Recently, $w$-event privacy based on DP has emerged for private stream data collection and analysis \cite{DBLP:conf/infocom/WangZLWQR16, DBLP:conf/infocom/WangLPRLC20, DBLP:conf/sigmod/RenSYYZX22}. It effectively protects the privacy of $w$ consecutive related events while offering accurate stream statistics. However, different users may have different privacy requirements.
For instance, entertainers may be reluctant to reveal too much about their locations (i.e., large $w$-event size), while street artists may be willing to expose their locations (i.e., small $w$-event size) for more attention. Thus, if we fix the window size $w$ for all users, it is hard to make everyone satisfied.

We illustrate an example of online car-hailing shown in Figure~\ref{fig:introduction_example}.

\begin{figure}[t!]
	\centering
	\setlength{\abovecaptionskip}{0.2cm}
	\includegraphics[width=0.45\textwidth]{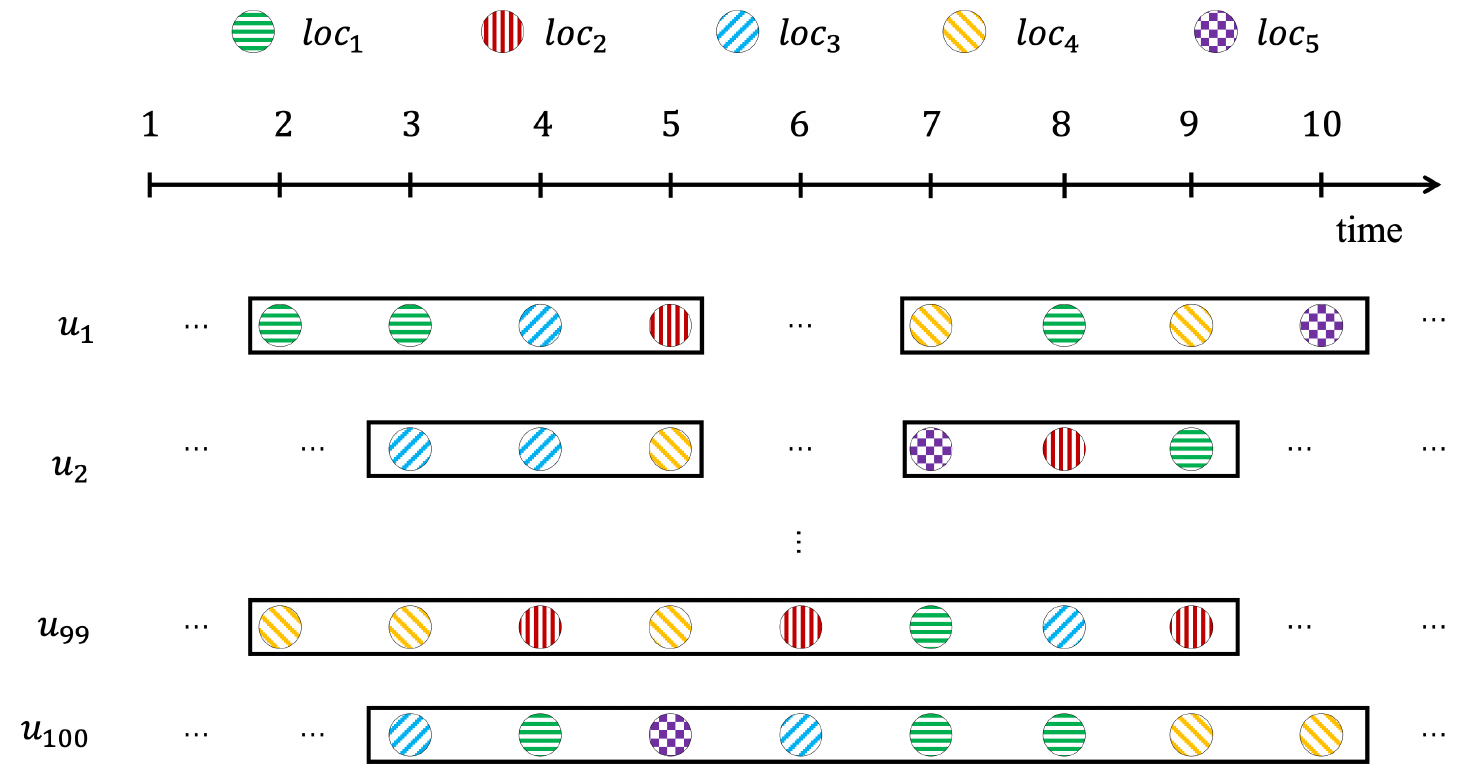}
	\caption{Different event window sizes for different time slots. }\label{fig:introduction_example}
\end{figure}

\begin{example}
	Consider a scenario with $100$ drivers $U=\{u_1,..., u_{100}\}$ who share their locations from $\{loc_1,...,loc_5\}$ at each time slot. 
	Each driver $u_i$ is protected by $w_i$-event privacy, meaning their location data is safeguarded through $\epsilon$-DP across at least $w_i$ consecutive time slots, where $\epsilon$ represents their required privacy protect stength.
	For example, $u_1$ requires location protection across any $4$ consecutive time slots, while $u_{99}$ and $u_{100}$ need protection across any $8$ consecutive time slots. 
	For the drivers $u_i \in U \backslash \{u_{99}, u_{100}\}$, the window size does not exceed $4$. 
	
	Under traditional $w$-event privacy, satisfying all drivers' privacy needs requires setting the event window size to the maximum value (i.e., $w=8$) and making full use of the privacy budget to achieve high utility while maintaining $8$-event privacy. 
	Let $AE_{avg}$ denote the average square error at each time slot, defined as the variance when adding Laplace noise (i.e., {\scriptsize$AE_{avg}=2b^2=2\times\left(\frac{1}{\epsilon/w}\right)^2$}).
	With a total privacy budget $\epsilon$ of $1$ and using the \textit{Uniform method}~\cite{DBLP:journals/pvldb/KellarisPXP14}, the average square error at each time slot under $8$-event privacy is $AE_{avg} = 2\times(\frac{w}{\epsilon})^2 = 128$.
	However, the first $98$ drivers do not actually need $8$-event privacy. 
	By setting the window size to $w=4$ and using the \textit{threshold method}~\cite{DBLP:conf/icde/JorgensenYC15} (or the \textit{sample method}~\cite{DBLP:conf/icde/JorgensenYC15}), we can achieve $AE_{avg} \approx 2\times(\frac{w}{\epsilon})^2 = 32$, which is significantly lower than the error from traditional $8$-event privacy.
\end{example}

In this paper, we define the \problemDefineTotalName{} (\problemDefineSimpleName{}) problem to model personalized privacy requirements in stream data publication. 
To solve \problemDefineSimpleName{}, 
there are two challenges: 1) effectively unifying the privacy budget across all users into a single value to maximize publication utility; 2) effectively distribute each user's personalized privacy budget to their personalized window size to maximize publication utility.

To improve publication utility, we address \problemDefineSimpleName{} using the centralized DP model~\cite{DBLP:conf/icalp/Dwork06}, which requires a single centralized privacy budget for publishing statistics at each time slot. 
However, users have different personalized privacy budgets (not the same budgets), traditionally, satisfying everyone's privacy requirements means selecting the minimum budget among all users, which results in the lowest utility. \textit{How to use a privacy budget higher than the minimum one to achieve higher utility while satisfying the privacy requirement of user with the minimum privacy?} It seems unachievable at a glance. We solve this challenge through elaborately applying the Sampling Mechanism~\cite{DBLP:conf/icde/JorgensenYC15}. Our method theoretically guarantees that even though the selected unified privacy budget is higher than the minimum privacy budget, no privacy leakages for any users exist.

Intuitively, time slots with higher rates of change contain more valuable information and are thus more important. 
To maximize utility, we need to allocate large privacy budgets to publications at these important time slots while approximating others (allocating none). \textit{How to identify important time slots and allocate privacy budgets to achieve maximum publication utility?}
To address this challenge, we design two methods: \solutionMethodATotalName{} (\solutionMethodA{}) and \solutionMethodBTotalName{} (\solutionMethodB{}). 
Both methods identify important time slots by measuring the dissimilarity between current and historical statistics. 
The key difference lies in their assumption: 
\solutionMethodA{} assumes stream data has a stable and high rate of change. Thus, it allocates budgets to each publication in an exponentially decreasing fashion per window. 
\solutionMethodB{}, however, assumes stream data has an unstable and low rate of change. Thus, it skipps or approximates many less important publications. Besides, it maximizes the accuracy of important publications by utilizing unused budget from skipped publications while nullifying future time slot budgets. 
We prove that both \solutionMethodA{} and \solutionMethodB{} satisfy \privacyDefineTotalName{} (\privacyDefineSimpleName{}) and provide their average error upper bounds.
We summarize our contributions as follows.
\begin{itemize}[leftmargin=*]
	\item We formally define \privacyDefineTotalName{} 
	for \problemDefineSimpleName{} in Section~\ref{pro_def}.
	
	\item We propose  \solutionATotalName{} (\solutionA{}) and two methods, namely \solutionMethodATotalName{} (\solutionMethodA{}) and \solutionMethodBTotalName{} (\solutionMethodB{}), to support personalized $\vectorfont{w}$-event privacy with theoretical analyses in Section~\ref{basic_method}.
	
	\item We test our methods on both real and synthetic data sets to demonstrate their efficiency and effectiveness in Section~\ref{experiment}.
\end{itemize}

\section{Related Work}
\subsection{Data Stream Estimation under Differential Privacy}
Based on the privacy model, there are two types of data stream estimation methods: centralized differential privacy~\cite{DBLP:conf/icalp/Dwork06} (CDP) based methods and local differential privacy~\cite{DBLP:conf/stoc/BassilyS15} (LDP) based methods.

\noindent\textbf{Data Stream Estimation under CDP.}
Dwork et al.  first address the problem of Differential Privacy (DP) on data streams~\cite{DBLP:conf/stoc/DworkNPR10}.
They define two types of DP levels: \emph{event-level differential privacy} (event-DP) and \emph{user-level differential privacy} (user-DP).

In event-DP, each single event is hidden in statistic queries. 
Dwork et al. focus on the finite event scenarios and propose a binary tree method to achieve high statistical utility while maintaining event-DP~\cite{DBLP:conf/stoc/DworkNPR10}.
Chan et al. extend it to infinite cases, and produce partial summations for binary counting~\cite{DBLP:journals/tissec/ChanSS11}.
Dwork et al. introduce a cascade buffer counter that updates adaptively based on stream density~\cite{DBLP:conf/soda/Dwork10}.
Bolot et al. propose \textit{decayed privacy} which reduces the privacy costs for past data~\cite{DBLP:conf/icdt/BolotFMNT13}.
Chen et al. develop PeGaSus, a perturb-group-smooth framework for multiple queries under event-DP~\cite{DBLP:conf/ccs/ChenMHM17}.
However, event-DP assumes all element in a stream are independent, making it unsuitable for correlated data stream publishing.

In user-DP, all events for each user are hidden in statistic queries.
Fan et al. propose the FAST algorithm with a sampling-and-filtering framework, counting finite stream data under user-DP~\cite{DBLP:journals/tkde/FanX14}.
Cummings et al. address heterogeneous user data, estimating population-level means while achieving user-DP~\cite{DBLP:conf/nips/CummingsFMT22}. However, they only consider finite data. Offering user-DP for infinite data requires infinite perturbation, leading to poor long-term utility~\cite{DBLP:journals/pvldb/KellarisPXP14}.
Cao et al.~\cite{DBLP:journals/tkde/CaoYXX19} assume that temporal data adhere to the Markov constraint. They decompose privacy leakage into Backward Privacy Leakage (BPL) and Forward Privacy Leakage (FPL), and on this basis propose a new privacy definition: Differential Privacy under Temporal Correlations ($\alpha$-DP$_{\entity{T}}$). However, this definition is limited by the Markov assumption.

To bridge the gap between event-DP and user-DP, Kellaris et al. propose $w$-event DP for infinite streams~\cite{DBLP:journals/pvldb/KellarisPXP14}.
It ensures $\epsilon$-DP for any group of events within a time window of size $w$. They introduce two methods, \textit{Budget Distribution} and \textit{Budget Absorption}, to optimize privacy budget use and estimate statistics effectively. However, neither method handles stream data with significant changes. 
Wang et al. apply the $w$-event concept to the FAST method, proposing a multi-dimensional stream release mechanism called \textit{ResueDP}, which achieves accurate estimation for both rapid and slow data stream changes~\cite{DBLP:conf/infocom/WangZLWQR16}. 
A limitation of all these methods is their reliance on a trusted server to ensure privacy.

\noindent\textbf{Data Stream Estimation under LDP.}
To overcome the dependence on a trusted server, LDP~\cite{DBLP:conf/stoc/BassilyS15} has recently been proposed and adopted by many major companies such as Microsoft, Apple and Google. 
Erlingsson et al. introduce RAPPOR to estimate finite streams under LDP~\cite{DBLP:conf/ccs/ErlingssonPK14}. They design a two-layer randomized response mechanism (i.e., permanent randomized response and instantaneous randomized response) to protect each individual's data. However, RAPPOR is limited to uncorrelated stream data. To address the problem of correlated time series data, Erlingsson et al. develop a new privacy model that introduces \textit{shuffling} to amplify the LDP privacy level~\cite{DBLP:conf/soda/ErlingssonFMRTT19}. However, this model only suits finite stream data.
Joseph et al. propose THRESH for evolving data under LDP~\cite{DBLP:conf/nips/JosephRUW18}, which consumes privacy budget at global update time slots selected by users' LDP voting.
However, it is not applicable to infinite streams as it assumes a fixed number of global updates.
Wang et al. extend event-level privacy from CDP to LDP and design the efficient ToPL method under event LDP~\cite{DBLP:conf/ccs/0001C0SC0LJ21}.
Nevertheless, event-level LDP focuses solely on event-level privacy, lacking privacy protection for correlated data in streams.
Bao et al. propose an $(\epsilon,\delta)$-LDP method (called CGM) for finite streaming data collection using the analytic Guassian mechanism, but requires periodic privacy budget renewal~\cite{DBLP:journals/pvldb/BaoYXD21}.
Ren et al. introduce LDP-IDS for infinite streaming data collection and analysis under $w$-event LDP~\cite{DBLP:conf/sigmod/RenSYYZX22}.
They propose two budget allocation methods and two population allocation methods, bridging the gap between event LDP and user LDP while improving estimation accuracy.
However, all these methods cannot be adopted to support personalized event window sizes.

\subsection{Non-Uniformity Differential Privacy}
Recently, some studies address the non-uniform privacy requirements among items (table columns) or records (table rows)~\cite{DBLP:conf/uss/Murakami019}. 

Alaggan et al. first examine scenarios where each database instance comprises a single user's profile~\cite{DBLP:journals/jpc/AlagganGK16}. They focus on varying privacy requirements for different items and formally define Heterogeneous Differential Privacy (HDP).

Jorgensen et al. investigate the privacy preservation for individual rows, introducing Personalized Differential Privacy (PDP)~\cite{DBLP:conf/icde/JorgensenYC15}. They design two mechanisms leveraging non-uniform privacy requirements to achieve better utility than standard uniform DP.
Kotsogiannis et al. recognize that different data have different sensitivity, then define One-side Differential Privacy (OSPD) and propose algorithms that truthfully release non-sensitive record samples to enhance accuracy in DP-solutions~\cite{DBLP:conf/icde/KotsogiannisDHM20}.

Andr\'{e}s et al. introduce a novel non-uniform privacy concept called Geo-Indistinguishability (Geo-I), where the privacy level for any point increases as the distance to this point decreases~\cite{DBLP:conf/ccs/AndresBCP13}.
Wang et al.~\cite{DBLP:journals/tmc/WangHLWWYQ19} and Du et al.~\cite{DBLP:conf/icde/DuCZX00F23} explore  PDP in spatial crowdsourcing, and develop highly effective private task assignment methods to satisfy diverse workers' privacy and utility requirements.
Liu et al. investigate HDP in federated learning~\cite{DBLP:journals/pvldb/LiuLXLM21}. 
They assume different clients hold DP budget and divide them into private and public parts, then propose two methods to project the ``public'' clients' models into ``private'' clients' models to improve the joint model's utility.
However, all above studies are not suitable for stream data. 

In this paper, we propose \solutionATotalName{} (\solutionA{}) with two implementation methods: \solutionMethodATotalName{} (\solutionMethodA{}) and \solutionMethodBTotalName{} (\solutionMethodB{}). 
Our approach extends traditional $w$-event privacy mechanisms by introducing $\vectorfont{\epsilon}$-personalized differential privacy methods to support personalized privacy requirements.
This enhancement enables our mechanism and methods to handle both infinite correlated data streams and personalized privacy requirements, building upon the foundations of traditional $w$-event privacy mechanisms.

\begin{table}[t!]\vspace{-2ex}
	\caption{Summary for related work.}
	\label{relatedwork_table}\vspace{1ex}
	\centering
		\resizebox{8.5cm}{!}{
			\begin{tabular}{|ccc|c|c|}
				\hline
				\multicolumn{2}{|c|}{\textbf{Model Types}}                                                                          & \textbf{Methods}                                                                                                                           & \textbf{\begin{tabular}[c]{@{}c@{}}Is infinite \\  and correlated\end{tabular}} & \textbf{\begin{tabular}[c]{@{}c@{}}Is personalized \\ privacy\end{tabular}} \\ \hline
				\multicolumn{1}{|c|}{}                                 & \multicolumn{1}{c|}{}                                      & Finite B-tree~\cite{DBLP:conf/stoc/DworkNPR10}                                                                       & \XSolidBrush                                                                                      & \XSolidBrush                      \\ \cline{3-5} 
				\multicolumn{1}{|c|}{}                                 & \multicolumn{1}{c|}{}                                      & Infinite B-tree~\cite{DBLP:journals/tissec/ChanSS11}                                                                 & \XSolidBrush                                                                                      & \XSolidBrush                      \\ \cline{3-5} 
				\multicolumn{1}{|c|}{}                                 & \multicolumn{1}{c|}{}                                      & \begin{tabular}[c]{@{}c@{}}Adaptive-density \\ Counter~\cite{DBLP:conf/soda/Dwork10}\end{tabular}                    & \XSolidBrush                                                                                      & \XSolidBrush                      \\ \cline{3-5} 
				\multicolumn{1}{|c|}{}                                 & \multicolumn{1}{c|}{}                                      & Decayed Privacy~\cite{DBLP:conf/icdt/BolotFMNT13}                                                                    & \XSolidBrush                                                                                      & \XSolidBrush                      \\ \cline{3-5} 
				\multicolumn{1}{|c|}{}                                 & \multicolumn{1}{c|}{\multirow{-5}{*}{event-level privacy}} & PeGaSus~\cite{DBLP:conf/ccs/ChenMHM17}                                                                               & \XSolidBrush                                                                                      & \XSolidBrush                      \\ \cline{2-5} 
				\multicolumn{1}{|c|}{}                                 & \multicolumn{1}{c|}{}                                      & FAST~\cite{DBLP:journals/tkde/FanX14}                                                                                & \CheckmarkBold                                                                                    & \XSolidBrush                      \\ \cline{3-5} 
				\multicolumn{1}{|c|}{}                                 & \multicolumn{1}{c|}{\multirow{-1}{*}{user-level privacy}}  & \begin{tabular}[c]{@{}c@{}}Private heterogeneous \\ mean estimation~\cite{DBLP:conf/nips/CummingsFMT22}\end{tabular} & \CheckmarkBold                                                                                    & \XSolidBrush                      \\ \cline{3-5} 
				\multicolumn{1}{|c|}{}                                 & \multicolumn{1}{c|}{}                                      & $\alpha$-DP$_{\entity{T}}$~\cite{DBLP:journals/tkde/CaoYXX19}                                                                               & \CheckmarkBold                                                                                    & \XSolidBrush                      \\ \cline{2-5} 
				\multicolumn{1}{|c|}{}                                 & \multicolumn{1}{c|}{}                                      & BD \& BA~\cite{DBLP:journals/pvldb/KellarisPXP14}                                                                    & \CheckmarkBold                                                                                    & \XSolidBrush                      \\ \cline{3-5} 
				\multicolumn{1}{|c|}{\multirow{-9}{*}{Centralized DP}} & \multicolumn{1}{c|}{\multirow{-2}{*}{$w$-event privacy}}     & ResuseDP~\cite{DBLP:conf/infocom/WangZLWQR16}                                                                        & \CheckmarkBold                                                                                    & \XSolidBrush                      \\ \hline
				\multicolumn{1}{|c|}{}                                 & \multicolumn{1}{c|}{}                                      & RAPPOR~\cite{DBLP:conf/ccs/ErlingssonPK14}                                                                           & \XSolidBrush                                                                                      & \XSolidBrush                      \\ \cline{3-5} 
				\multicolumn{1}{|c|}{}                                 & \multicolumn{1}{c|}{\multirow{-2}{*}{event-level privacy}} & ToPL~\cite{DBLP:conf/ccs/0001C0SC0LJ21}                                                                              & \XSolidBrush                                                                                      & \XSolidBrush                      \\ \cline{2-5} 
				\multicolumn{1}{|c|}{}                                 & \multicolumn{1}{c|}{}                                      & Shuffling LDP~\cite{DBLP:conf/soda/ErlingssonFMRTT19}                                                                & \CheckmarkBold                                                                                    & \XSolidBrush                      \\ \cline{3-5} 
				\multicolumn{1}{|c|}{}                                 & \multicolumn{1}{c|}{}                                      & THRESH~\cite{DBLP:conf/nips/JosephRUW18}                                                                             & \CheckmarkBold                                                                                    & \XSolidBrush                      \\ \cline{3-5} 
				\multicolumn{1}{|c|}{}                                 & \multicolumn{1}{c|}{\multirow{-3}{*}{user-level privacy}}  & CGM~\cite{DBLP:journals/pvldb/BaoYXD21}                                                                              & \CheckmarkBold                                                                                    & \XSolidBrush                      \\ \cline{2-5} 
				\multicolumn{1}{|c|}{\multirow{-6}{*}{Local DP}}       & \multicolumn{1}{c|}{$w$-event privacy}                       & LDP-IDS~\cite{DBLP:conf/sigmod/RenSYYZX22}                                                                           & \CheckmarkBold                                                                                    & \XSolidBrush                      \\ \hline
				\multicolumn{2}{|c|}{Item heterogeneous}                                                                            & HDP~\cite{DBLP:journals/jpc/AlagganGK16}                                                                             & \XSolidBrush                                                                                      & \XSolidBrush                      \\ \hline
				\multicolumn{2}{|c|}{}                                                                                              & PDP~\cite{DBLP:conf/icde/JorgensenYC15}                                                                              & \XSolidBrush                                                                                      & \CheckmarkBold                    \\ \cline{3-5} 
				\multicolumn{2}{|c|}{}                                                                                              & OSDP~\cite{DBLP:conf/icde/KotsogiannisDHM20}                                                                         & \XSolidBrush                                                                                      & \CheckmarkBold                    \\ \cline{3-5} 
				\multicolumn{2}{|c|}{}                                                                                              & Geo-I~\cite{DBLP:conf/ccs/AndresBCP13}                                                                               & \XSolidBrush                                                                                      & \CheckmarkBold                    \\ \cline{3-5} 
				\multicolumn{2}{|c|}{}                                                                                              & PWSM, VPDM~\cite{DBLP:journals/tmc/WangHLWWYQ19}                                                                     & \XSolidBrush                                                                                      & \CheckmarkBold                    \\ \cline{3-5} 
				\multicolumn{2}{|c|}{}                                                                                              & PUCE, PGT~\cite{DBLP:conf/icde/DuCZX00F23}                                                                           & \XSolidBrush                                                                                      & \CheckmarkBold                    \\ \cline{3-5} 
				\multicolumn{2}{|c|}{\multirow{-6}{*}{Record heterogenous}}                                                         & PFA, PFA+~\cite{DBLP:journals/pvldb/LiuLXLM21}                                                                       & \XSolidBrush                                                                                      & \CheckmarkBold                    \\ \hline
				\multicolumn{3}{|c|}{ \textbf{Our mechanisms}}                                                                                                                                                                                                       & \CheckmarkBold                                                                                    & \CheckmarkBold                    \\ \hline
			\end{tabular}
		}
\end{table}

\section{Problem Settings}\label{pro_def}
In this section, we first introduce key concepts, including data streams. Next, we present the new definition of \privacyDefineTotalName{}. Finally, we provide the problem definition: \problemDefineTotalName{} (\problemDefineSimpleName{}). Table~\ref{tbl:notations} outlines the notations used throughout this paper.

\subsection{Data Stream}
\begin{definition}(Data Stream~\cite{DBLP:journals/pvldb/KellarisPXP14}).
	Let $D_t\in \algvar{D}$ be a database with $d$ columns and $n$ rows (each row representing a user) at $t$-th time slot.
	The infinite database sequence $S=[D_1,D_2, \ldots]$ is called a data stream, where $S[t]$ is the $t$-th element in $S$ (i.e., $S[t]=D_t$).
\end{definition}

For any data stream $S$, its substream between time slot $t_l$ and $t_r$ (where $t_l < t_r$) is noted as $S_{t_l,t_r}=[D_{t_l},D_{t_l+1}, \ldots, D_{t_r}]$.
For $t_l=1$, we denote $S_t=[D_1,D_2, \ldots, D_t]$ and call it the \textit{stream prefix} of $S$.

\begin{definition}(Data Stream Count Publishing~\cite{DBLP:journals/pvldb/KellarisPXP14}).\label{def:count_release}
	Let $Q: \algvar{D}\to \constvar{R}^d$ be a count query. 
	Then, $Q(S[t])=Q(D_t)=\vectorfont{c}_t$ is the count data to be published at time slot $t$, where $\vectorfont{c}_t(j)$ represents the count of the $j$-th column of $D_t$.
	The infinite count data series $[\vectorfont{c}_1, \vectorfont{c}_2,\ldots]$ is called a data stream count publishing.
\end{definition}

\subsection{\privacyDefineTotalName{}}
\begin{definition}($w$-Neighboring Stream Prefixes~\cite{DBLP:journals/tissec/ChanSS11,DBLP:journals/pvldb/KellarisPXP14}).
	Let $w$ be a positive integer, two stream prefixes $S_t$, $S'_t$ are $w$-neighboring (i.e., $S_t \sim_w S'_t$), if 
	\begin{enumerate} [leftmargin=*]
		\item for each $S_t[k]$, $S'_t[k]$ such that $k\leq t$ and $S_t[k]\neq S'_t[k]$, it holds that $S_t[k]$ and $S'_t[k]$ are neighboring~\cite{DBLP:journals/pvldb/KellarisPXP14} in centralized DP, and
		\item for each $S_t[k_1]$, $S_t[k_2]$, $S'_t[k_1]$, $S'_t[k_2]$ with $k_1<k_2$, $S_t[k_1]\neq S'_t[k_1]$ and $S_t[k_2]\neq S'_t[k_2]$, it holds that $k_2-k_1+1\leq w$.
	\end{enumerate}
\end{definition}

\begin{definition}(\privacyDefineTotalName{}, \privacyDefineSimpleName{})\label{Def_EPDP}.
	Let $\entity{M}$ be a mechanism that takes a stream prefix of arbitrary size as input.
	Let $\entity{O}$ be the set of all possible outputs of $\entity{M}$.
	Given a universe of users $U=\{u_1,u_2,\ldots ,u_{|U|}\}$, 
	then $\entity{M}$ is $(\vectorfont{w},\vectorfont{\epsilon})$-EPDP if $\forall O\subseteq\entity{O}$, $\forall w_i\in\vectorfont{w}$ and $\forall S_t, S'_t$ satisfying $S_t \sim_{w_i} S'_t$, it holds that
	{\scriptsize$$\Pr[M(S_t)\in O]\leq e^{\epsilon_i} \Pr[M(S'_t)\in O],$$}
	where $u_i\in U$ requires $w_i$-event privacy and $\epsilon_i$ denotes $u_i$'s privacy budget requirement within $w_i$ continuous events.
\end{definition}
We denote the pair $( w_i, \epsilon_i)$ as $u_i$'s \textit{privacy requirement}. Specifically, when $w_i=1$, it collapses as $\vectorfont{\epsilon}$-Personalized Differential Privacy ($\vectorfont{\epsilon}$-PDP)~\cite{DBLP:conf/icde/JorgensenYC15}. Besides, when $(w_i, \epsilon_i)$ becomes constant (i.e., $(w,\epsilon)$), it collapses as $w$-Event Privacy\cite{DBLP:journals/pvldb/KellarisPXP14,DBLP:conf/sigmod/RenSYYZX22}.

\subsection{Definition of \problemDefineSimpleName{}}
Given a data stream $S$, the server aims to obtain the data stream count publishing, denoted as $\vectorfont{c}=[\vectorfont{c}_1, \vectorfont{c}_2, \ldots]$.
However, to protect user privacy, the server only receives the obfuscated version of the data stream, $S'$, and subsequently publishes the estimated data stream count (i.e., estimation count), denoted as $\vectorfont{r}=[\vectorfont{r}_1, \vectorfont{r}_2, \ldots]$. 
We now define the problem as follows.

\begin{table}[t!]
	\caption{Notations.}
	\label{tbl:notations}
	\centering
	\scalebox{1}{
		\begin{tabular}{c|c}
			\hline
			\textbf{Notations}        & \textbf{Description}                                       \\ \hline
			$\entity{D}$             & the database domain                                        \\ 
			$D_t$             & a database at time slot $t$                                        \\ 
			$S$             & a data stream                                        \\ 
			$u_i$			& the $i$-th user			\\ 
			$\vectorfont{x}_{i,t}$       & $u_i$'s data at time slot $t$                                        \\ 
			$\vectorfont{c}_t$             & a real statistical histogram at time slot $t$                                     \\ 
			$\vectorfont{r}_i$             & an estimation statistic histogram at time slot $t$                                      \\ 
			$w_{i}$             & $u_i$'s privacy window size                                      \\ 
			$\epsilon_{i}$             & $u_i$'s privacy budget                                  \\ \hline
		\end{tabular}
	}
\end{table}

\begin{definition} (\problemDefineSimpleName{}).\label{Problem_Def} 
	Given a user set $U=\{u_1, u_2,..., u_n\}$,  each $u_i$ holds a privacy requirement pair $(w_{i}, \epsilon_{i})$ and a series data $\vectorfont{x}_{i,t}$ for $t\in\constvar{N}^+$. All the $\vectorfont{x}_{i,t}$ for $u_i\in U$ at time slot $t$ form $D_t$. All the $D_t$ form an infinite data stream $S=[D_1, D_2, \ldots]$. \problemDefineSimpleName{} is to publish an obfuscated histogram $\vectorfont{r}=[\vectorfont{r}_1, \vectorfont{r}_2, \ldots]$ of $S$ at each time slot $t$ achieving $(\vectorfont{w},\vectorfont{\epsilon})$-EPDP with the error between $\vectorfont{r}$ and $\vectorfont{c}$ minimized, namely $\forall T\in\constvar{N}^+$:
	\begin{equation}\notag
		\begin{aligned}
			\min_{\epsilon_\theta}\;\; & \sum\limits_{t\in[T]} \|\vectorfont{r}_t-\vectorfont{c}_t\|_2^2 &\\
			s.t.\;\; & \sum_{k=\min{(t-w_{i}+1,1)}}^{t}\epsilon_{i,k}\leq \epsilon_{i}, \;\;\;\; \forall u_i\in U & 
		\end{aligned}
	\end{equation}
	\noindent where $\epsilon_{i,k}$ indicates the privacy budget at time slot $k$.
\end{definition}

\section{\solutionATotalName{}}\label{basic_method}
In this section, we analyze the errors in reporting obfuscated data stream counts and introduce Optimal Budget Selection (OBS) method  to minimize these errors. 
We then propose \solutionATotalName{} (\solutionA{}) to address \problemDefineSimpleName{}. 
The core idea of \solutionA{} is to select the optimal privacy budget $\epsilon_{opt}(t)$ and report obfuscated count results that satisfy $\epsilon_{opt}(t)$-DP at each time slot $t$.

\subsection{Reporting Errors}
For each time slot, we use the Sampling Mechanism (SM)~\cite{DBLP:conf/icde/JorgensenYC15} to satisfy all users' privacy requirements (i.e., achieving $\vectorfont{\epsilon}$-PDP). 
The SM consists of two steps: \textit{sample} ($SM_{s}$) and \textit{disturb} ($SM_{d}$). 
In $SM_s$, the server first sets a privacy budget threshold $\epsilon_\theta$, then constructs a sampling subset $D_S$ by appending items $x_i$ with $\epsilon_i \geq \epsilon_\theta$ to $D_S$, while sampling other items $x_j$ with $\epsilon_j < \epsilon_\theta$ at a probability of $p_j = \frac{e^{\epsilon_j} - 1}{e^{\epsilon_\theta} - 1}$. In $SM_d$, the server employs a DP mechanism (e.g., the Laplace Mechanism) to report an obfuscated result that achieves $\epsilon_\theta$-DP.

SM introduces two types of errors: \textit{sampling error} and \textit{noise error}.
At each time slot $t$, given a privacy budget threshold $\epsilon_\theta$, the data reporting error is $err(\epsilon_\theta)=err_{s}(\epsilon_\theta)+err_{dp}(\epsilon_\theta)$. 
Here, $err_{s}(\epsilon_\theta)$ represents the \textit{sampling error} from sampling users with privacy budgets below $\epsilon_\theta$, while $err_{dp}(\epsilon_\theta)$ represents the \textit{noise error} from adding noise to achieve $\epsilon_\theta$-DP. 
Next, we introduce these sampling and noise errors in detail.

\begin{definition}(Sampling Error~\cite{DBLP:conf/icde/JorgensenYC15}).
	Given a privacy budget threshold $\epsilon_\theta$ and $m$ kinds of privacy budget requirements $\tilde{\epsilon}_1, \tilde{\epsilon}_2,\ldots , \tilde{\epsilon}_m$ from $n$ users with $\tilde{\epsilon}_i<\tilde{\epsilon}_j$ for $i<j$ and $i, j\in[m]$ where $\tilde{\epsilon}_i$ is declared by $n_i$ users~$(\sum\limits_{i=1}^{m}n_i=n)$,
	the sampling error $err_{s}(\epsilon_\theta)$ is defined as
	\begin{equation}\label{eq:sampling_error}
		\begin{aligned}
			err_{s}(\epsilon_\theta) &= Var(count(\vectorfont{r}_t)) + bias(\vectorfont{r}_t)^2 = \sum\limits_{\tilde{\epsilon}_i<\epsilon_\theta}n_i p_i(1-p_i) + \left(\sum\limits_{\tilde{\epsilon}_i<\epsilon_\theta}n_i(1-p_i)\right)^2,
		\end{aligned}
	\end{equation}
	where $p_i=\frac{e^{\tilde{\epsilon}_i}-1}{e^{\epsilon_\theta}-1}$.
\end{definition}
\begin{definition}(Noise Error).
	The noise error $err_{dp}(\epsilon_{\theta})$ is defined as the error of the Laplace mechanism, namely, 
	\begin{equation}\label{eq:noise_error}
		\begin{aligned}
			err_{dp}(\epsilon_{\theta}) &= \frac{2}{\epsilon_\theta^2}.
		\end{aligned}
	\end{equation}
\end{definition}
Various metrics exist to measure the errors of Laplace mechanisms for noise error, including variance \cite{DBLP:conf/icde/JorgensenYC15,DBLP:conf/sigmod/RenSYYZX22}, scale \cite{DBLP:journals/pvldb/KellarisPXP14,DBLP:journals/fttcs/DworkR14}, and $(\alpha,\beta)$-usefulness \cite{DBLP:journals/fttcs/DworkR14,DBLP:conf/stoc/BlumLR08}. In this work, we employ variance as our metric~\cite{DBLP:conf/icde/JorgensenYC15}.

Based on Equations~\eqref{eq:sampling_error} and~\eqref{eq:noise_error}, we can observe that $err_{s}$ depends  on $n_i$, $\tilde{\epsilon}_i$ and $\epsilon_\theta$, and is independent of $\vectorfont{r}_t$. 
Similarly, $err_{dp}$ depends on $\epsilon_\theta$, and is independent of $\vectorfont{r}_t$.

\subsection{Optimal Budget Selection} 
Given the privacy budget requirements $(\epsilon_{1,t}, \epsilon_{2,t}, \ldots, \epsilon_{n,t})$ of $n$ users, we can determine the frequency of each privacy budget requirement and select the optimal $\epsilon_\theta$ that minimizes the data reporting error $err$. This process is detailed in Algorithm~\ref{alg:OPT_B_C}.

Taking $n$ privacy budgets as input, the Optimal Budget Selection (OBS) algorithm counts the different privacy budgets.
Assume there are $\tilde{n}$ distinct privacy budgets, with $n_k$ users requiring $\tilde{\epsilon}_k$ for $k\in[\tilde{n}]$.
Let $\tilde{\epsilon}$ be the set of different privacy budget and $N$ be their corresponding frequencies (Lines \ref{e_set}-\ref{n_set}).
Then, OBS finds the minimum reporting error $err_{min}$ (lines \ref{start_search}-\ref{end_search}).
Specifically, it iterates over all $\tilde{\epsilon}_k\in\tilde{\vectorfont{\epsilon}}$ and selects the value $\tilde{\epsilon}_k$ with the smallest reporing error $err=err_{s}(\tilde{\epsilon}_k)+err_{dp}(\tilde{\epsilon}_k)$. 
\begin{algorithm}[t!]\small
	\caption{Optimal Budget Selection (OBS)}
	\label{alg:OPT_B_C}
	\DontPrintSemicolon
	\KwIn{personalized privacy budget set $\vectorfont{\epsilon}=(\epsilon_1,\epsilon_2,\ldots , \epsilon_n)$}
	\KwOut{$\epsilon_{opt}$, $err_{min}$}
	Set $\tilde{\vectorfont{\epsilon}}=(\tilde{\epsilon}_1,\tilde{\epsilon}_2,\ldots ,\tilde{\epsilon}_{\tilde{n}})$ as the set of different $\epsilon\in\vectorfont{\epsilon}$;\\ \label{e_set}
	Set $N=(n_1,n_2,\ldots ,n_{\tilde{n}})$ as the corresponding frequency of $\tilde{\epsilon}_k\in\tilde{\vectorfont{\epsilon}}$;\\ \label{n_set}
	Initialize $err_{min}$ as the upper bound of error value;\\
	\For{$\tilde{\epsilon}_k\in\tilde{\vectorfont{\epsilon}}$}{\label{start_search}
		$err=err_{s}(\tilde{\epsilon}_k)+err_{dp}(\tilde{\epsilon}_k)$;\\
		\If{$err<err_{min}$}{
			$err_{min}=err$;\\
			$\epsilon_{opt}$ as $\tilde{\epsilon}_k$;\\ 
		}
	}\label{end_search}
	\Return{$\epsilon_{opt}$, $err_{min}$}\;
\end{algorithm}

\begin{example}[Running Example of the OBS Algorithm]
	Suppose we have $10$ privacy budgets as input: $\vectorfont{\epsilon} = ($0.1, 0.4, 0.4, 0.1, 0.4, 0.4, 0.8, 0.8, 0.8, 0.4$)$.
	OBS first determines $\tilde{\vectorfont{\epsilon}}=($0.1, 0.4, 0.8$)$, $\tilde{n}=|\tilde{\vectorfont{\epsilon}}|=3$, and $N=($2, 5, 3$)$.
	Based on these statistics, OBS iterates through the $3$ privacy budgets in $\tilde{\vectorfont{\epsilon}}$ and calculates the errors:
	{\small$err_1=0+\frac{2}{0.1^2}=200$}, {\small$err_2=2\times\frac{e^{0.1}-1}{e^{0.4}-1}\times(1-\frac{e^{0.1}-1}{e^{0.4}-1})+(2\times(1-\frac{e^{0.1}-1}{e^{0.4}-1}))^2+\frac{2}{0.4^2}=15.31$} and {\small$err_3=2\times\frac{e^{0.1}-1}{e^{0.8}-1}\times(1-\frac{e^{0.1}-1}{e^{0.8}-1})+5\times\frac{e^{0.4}-1}{e^{0.8}-1}\times(1-\frac{e^{0.4}-1}{e^{0.8}-1})+(2\times(1-\frac{e^{0.1}-1}{e^{0.8}-1})+5\times(1-\frac{e^{0.4}-1}{e^{0.8}-1}))^2+\frac{2}{0.8^2}=89.74$}.
	Finally, OBS returns $0.4$ with the minimal error $15.31$.
\end{example}

\subsection{\solutionATotalName{}}
Budget division~\cite{DBLP:journals/pvldb/KellarisPXP14,DBLP:conf/sigmod/RenSYYZX22} is a traditional framework for publishing private stream data under $w$-event privacy.
It comprises two basic methods, namely \textit{Uniform} and \textit{Sampling} and two adaptive methods, namely \textit{\solutionCMPATotalName{}} (\solutionCMPA{}) and \textit{\solutionCMPBTotalName{}} (\solutionCMPB{}).
The adaptive methods leverage the stream's variation tendency, resulting in more accurate obfuscated estimations.

In this subsection, we extend the adaptive budget division framework to a personalized context and introduce our \solutionATotalName{} (\solutionA{}). 
Based on \solutionA{}, we propose two methods: \solutionMethodATotalName{} (\solutionMethodA{}) and \solutionMethodBTotalName{} (\solutionMethodB{}). 

In real applications, users must specify their privacy budgets and window sizes. System administrators first define a discretized privacy budget range (e.g., \{$0.1, 0.5, 0.9$\}) and a window size range (e.g., \{$40, 80, 120$\}). Then, they map ascending privacy budget values to descending privacy budget levels (e.g., High, Medium, Low) and ascending window size values to ascending window size levels (e.g., Small, Medium, Large). Users can then select both a privacy budget level and a window size level based on their needs and past experience. After users submit these selections, the server converts them into the corresponding values.

As shown in Algorithm~\ref{alg:PWSM}, the \solutionA{} algorithm takes three inputs: the historical estimation $His$, personalized privacy budget $\vectorfont{\epsilon}$, and personalized window size set $\vectorfont{w}$. Both $\vectorfont{\epsilon}$ and $\vectorfont{w}$ are fixed values collected from all users during system initialization.
\solutionA{} first calculates all users' privacy budget resources $\vectorfont{\epsilon}_t$ at the current time slot $t$ to satisfy \privacyDefineSimpleName{} (line~\ref{currentPB}).
It then divides $\vectorfont{\epsilon}_t$ into two parts: $\vectorfont{\epsilon}_t^{(1)}$ and $\vectorfont{\epsilon}_t^{(2)}$ (line~\ref{dividePB}).
Using $\vectorfont{\epsilon}_t^{(1)}$, \solutionA{} calculates the dissimilarity $dis$ between the current count value and the last published one by invoking the SM method~\cite{DBLP:conf/icde/JorgensenYC15}  (line~\ref{calculate:dis}).
Next, it sets the change threshold as the reporting error $err$ calculated with $\vectorfont{\epsilon}_t^{(2)}$ (line~\ref{calculate:err}).
Finally, \solutionA{} adaptively decides whether to publish a new obfuscated estimation or skip (i.e., use the last published one to approximate) by comparing $dis$ to $\sqrt{err}$ (lines~\ref{begin:judge_report}-\ref{end:judge_report}). 

\begin{algorithm}[t!]\small
	\caption{\solutionA{}}
	\label{alg:PWSM}
	\DontPrintSemicolon
	\KwIn{historical estimation $His$,  privacy requirement ($\vectorfont{w}$, $\vectorfont{\epsilon}$)}
	\KwOut{$\vectorfont{r}$}
	Get the current privacy budgets $\vectorfont{\epsilon}_t$ of all users as $\vectorfont{\epsilon}$ and $\vectorfont{w}$;\\ \label{currentPB}
	Divide $\vectorfont{\epsilon}_t$ into two parts $\vectorfont{\epsilon}^{(1)}_t$ and $\vectorfont{\epsilon}^{(2)}_t$ satisfying $\vectorfont{\epsilon}_t=\vectorfont{\epsilon}^{(1)}_t+\vectorfont{\epsilon}^{(2)}_t$;\\\label{dividePB}
	Calculate dissimilarity $dis$ between current estimation and the last estimation by $SM(\vectorfont{\epsilon}^{(1)}_t)$;\label{calculate:dis}\\
	Calculate the reporting error $err$ of current estimation by $OBS(\vectorfont{\epsilon}^{(2)}_t)$;\\ \label{calculate:err}
	\If{$dis>\sqrt{err}$}{\label{begin:judge_report}
		Calculate current estimation $\vectorfont{r}$ by $SM(\vectorfont{\epsilon}^{(2)}_t)$;\\
	} \Else {
		Set current estimation $\vectorfont{r}$ as the last reporting value;\\
	}
	\Return{$\vectorfont{r}$};\label{end:judge_report}
\end{algorithm}

To determine whether to publish a new obfuscated estimation or skip, we need to introduce a judgment measure called the \textit{personalized private dissimilarity measure}.

\textbf{Personalized Private Dissimilarity Measure.} The personalized dissimilarity measure $dis^*$ is defined as the absolute error between the true statistic $\tilde{\vectorfont{c}}_t$ under $SM_s$ (i.e., the \textit{sample} step of SM) at current time slot $t$ and the last publishing $\vectorfont{r}_l$, namely,
\begin{equation}\notag
	dis^* = \frac{1}{d}\sum\limits_{k=1}^{d}|\tilde{\vectorfont{c}}_t[k] - \vectorfont{r}_l[k]|.
\end{equation}
Our goal is to privately obtain the personalized dissimilarity $dis^*$ using the optimal privacy budget $\epsilon_{opt}$ calculated through $OBS$ algorithm. 
The personalized private dissimilarity measure $dis$ is defined as:
\begin{equation}\notag
	dis = dis^*+Lap\left(\frac{1}{d\cdot\epsilon_{opt}}\right),
\end{equation}
where $Lap$ represents the Laplace mechanism.

\subsection{\solutionMethodATotalName{} and \solutionMethodBTotalName{}}
We first introduce some notations to further clarify \solutionA{} in Algorithm~\ref{alg:PWSM}, and then propose two solutions to implement \solutionA{} in different scenarios.

\begin{figure}[t!]
	\centering
	\includegraphics[width=0.40\textwidth]{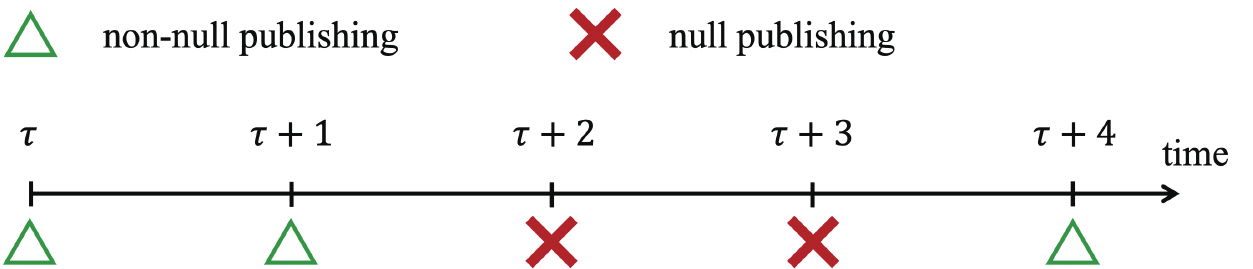}
	\caption{A non-null publishing example.}\label{fig:non-null-publishing}
\end{figure}

\noindent\textbf{Basic notations.} 
For a sequence of publications $(\vectorfont{r}_1,\vectorfont{r}_2,...,\vectorfont{r}_t)$ of length $t$, we define a \textit{null publishing} as an approximation value and \textit{non-null publishing} as a new value. For any time slot $2\leq\tau\leq t$, we refer to $\vectorfont{r}_{\tau-1}$ as the last reporting value (or last publishing) of time slot $\tau$. In the sequence $(\vectorfont{r}_1,\vectorfont{r}_2,...,\vectorfont{r}_{\tau})$, we define the most recent non-null publishing $\vectorfont{r}_l$ where $l<\tau$ as the last non-null publishing.
For example in Figure~~\ref{fig:non-null-publishing}, the publications at time slots $\tau, \tau+1, \tau+4$ are non-null publishing, while those at $\tau+2$ and $\tau+3$ are null publishing.

\begin{figure}[t!]
	\centering
	\includegraphics[width=0.40\textwidth]{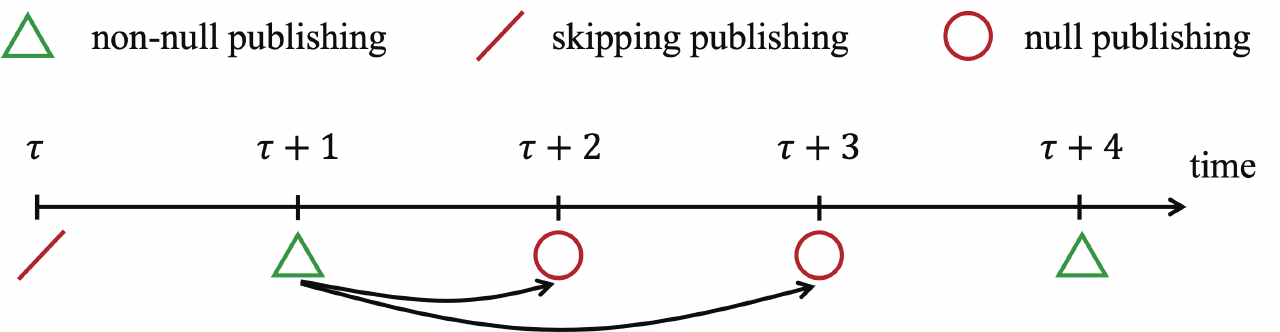}
	\caption{A nullified time slot example.}\label{fig:nullified-publishing}
\end{figure}
Given a privacy budget $\epsilon$ with window size $w$, 
the budget share $\bar{\epsilon}=\epsilon/w$ represents the smallest unit of privacy budget. 
The target is to maintain the total budget below $\epsilon$ within any $w$ window while maximizing utility. 
Assume publishing new obfuscated data costs $x$ budget shares ($x>1$), the following $x-1$ time slots use approximated values from their last reporting values. 
We refer to these $x-1$ time slots as nullified time slots.
For example, in Figure~\ref{fig:nullified-publishing}, with $\epsilon=4$ and $w=4$, $\bar{\epsilon}=1$. 
When time slot $\tau+1$ uses $3$ shares, the time slots $\tau+2$ and $\tau+3$ become nullified.

\begin{algorithm}[t!]\small
	\caption{Dissimilarity Calculation (DC)}
	\label{alg:DC}
	\DontPrintSemicolon
	\KwIn{$D_t$, current personalized privacy budget list  $\vectorfont{\epsilon}_t$, historical data publication $(\vectorfont{r}_1,\vectorfont{r}_2,\ldots, \vectorfont{r}_{t-1})$}
	\KwOut{$\vectorfont{r}_t$}
	$\epsilon_{opt}$= OBS($\vectorfont{\epsilon}_{t}$) ;\\ \label{mt1_start2}
	$\tilde{D}_t=SM_s(D_t,\vectorfont{\epsilon}_{t},\epsilon_{opt})$;\\ \label{pdp_start}
	$\tilde{\vectorfont{c}}_t=Q(\tilde{D}_t)$;\\
	Get the last non-null publishing $\vectorfont{r}_l$ from $(\vectorfont{r}_1,\vectorfont{r}_2,\ldots , \vectorfont{r}_{t-1})$;\\
	\Return{$dis=\frac{1}{d}\sum_{j=1}^{d}|\tilde{\vectorfont{c}}_t[j]-\vectorfont{r}_l[j]|+Lap(1/(d\cdot\epsilon_{opt}))$};\\ \label{pdp_end}
\end{algorithm}

\begin{algorithm}[t!]\small
	\caption{Personalized Budget Distribution}
	\label{alg:PBD}
	\DontPrintSemicolon
	\KwIn{$D_t$, privacy requirement ($\vectorfont{w}$, $\vectorfont{\epsilon}$), historical data publication $(\vectorfont{r}_1,\vectorfont{r}_2,\ldots, \vectorfont{r}_{t-1})$}
	\KwOut{$\vectorfont{r}_t$}
	Get the current window average budget $\bar{\epsilon}_i=\epsilon_i/w_i$ for each $i\in[n]$;\\
	$\vectorfont{\epsilon}_{t}^{(1)}=(\bar{\epsilon}_1/2,\bar{\epsilon}_2/2,\ldots ,\bar{\epsilon}_n/2)$;\\ \label{mt1_start}
	Get dissimilarity $dis$ by DC($D_t, \vectorfont{\epsilon}_{t}^{(1)}, \vectorfont{r}_1,...,\vectorfont{r}_{t-1}$) in Algorithm~\ref{alg:DC};\\ \label{mt1_end}
	$\epsilon_{rm,i}=\epsilon_i/2-\sum_{k=t-w_i+1}^{t-1}\epsilon_{i,k}^{(2)}$;\\ \label{pbd_mt2_start}
	$\vectorfont{\epsilon}_{t}^{(2)}=(\epsilon_{rm,1}/2,\epsilon_{rm,2}/2,\ldots ,\epsilon_{rm,n}/2)$;\\
	$\epsilon_{opt}^{(2)}$, $err_{opt}^{(2)}$ = OBS($\vectorfont{\epsilon}_{t}^{(2)}$);\\
	\label{pbd_mt2_mid} 
	\If{$dis>\sqrt{err_{opt}^{(2)}}$}{
		$\tilde{D}_t^{(2)}=SM_s(D_t,\vectorfont{\epsilon}_{t}^{(2)},\epsilon_{opt}^{(2)})$;\\ \label{new_pub_start}
		$\tilde{\vectorfont{c}}_t^{(2)}=Q(\tilde{D}_t^{(2)})$;\\
		\Return{$\vectorfont{r}_t=SM_d(\tilde{\vectorfont{c}}_t^{(2)},\epsilon_{opt}^{(2)})$}; \label{new_pub_end}
	} \Else {
		\Return{$\vectorfont{r}_t=\vectorfont{r}_{t-1}$};
	}\label{mt2_end}
\end{algorithm}
\noindent\textbf{\solutionMethodATotalName{} (\solutionMethodA{}).} As shown in Algorithm~\ref{alg:PBD}, \solutionMethodA{}  inputs the current user data $D_t$, all users' privacy requirements, and historical data publication.
The privacy budget $\epsilon_i$ of each user $u_i$ is divided into two parts: 1) calculate the dissimilarity between the current data distribution and the last published obfuscated data distribution (denoted as Part$_{DC}$) (Lines~\ref{mt1_start}-\ref{mt1_end}); 2) calculate the new obfuscated publication at the current time slot (denoted as Part$_{NOP}$) (Lines~\ref{pbd_mt2_start}-\ref{pbd_mt2_mid} and Lines~\ref{new_pub_start}-\ref{new_pub_end}).

In Part$_{DC}$, we allocate half of the average privacy budget per time slot for dissimilarity calculation (i.e., $\frac{\epsilon_i}{2w_i}$ for $u_i$). 
The process then calls the Dissimilarity Calculation (Algorithm~\ref{alg:DC}) to determine the dissimilarity. Within Algorithm~\ref{alg:DC}, the OBS algorithm selects the optimal budget threshold $\epsilon_{opt}$. Finally, it uses the SM~\cite{DBLP:conf/icde/JorgensenYC15} to compute the dissimilarity $dis$ (lines \ref{pdp_start}-\ref{pdp_end}).

In Part$_{NOP}$, we first calculate the remaining privacy budget $\epsilon_{rm,i}$ for each $u_i$.
We then set the publication privacy budget for each $u_i$ to half of $\epsilon_{rm,i}$.
Similar to dissimilarity calculation, we use the OBS algorithm to determine the optimal privacy budget $\epsilon_{opt}^{(2)}$ and its corresponding error $err_{opt}^{(2)}$.
At this point, we have obtained two measurements: the dissimilarity $dis$ and the square root of error $\sqrt{err_{opt}^{(2)}}$.
We compare these two measurements to determine whether to publish a new obfuscated statistic result or approximate the current result with the last publication.
If the $dis$ is greater than $\sqrt{err_{opt}^{(2)}}$, it indicates that the difference between the current data and the last published data exceeds the error of noise, then we republish a new obfuscated statistic result.
Otherwise, we take the last published result instead.

We illustrate the process of Personalized Budget Distribution with an example as follows:
\begin{example}\label{example_for_PBD}
	Suppose there are $n$ users distributed across $5$ locations, forming a complete graph.
	Figure~\ref{data_example} illustrates the privacy budget requirements, window size requirements and locations for the first three users across time slots $1$ to $5$.
	\begin{figure}[t!]
		\centering
		\includegraphics[width=0.33\textwidth]{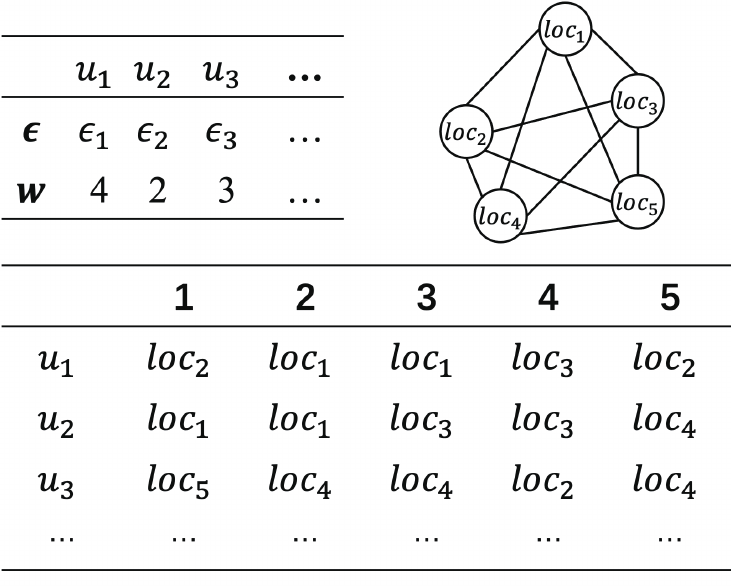}
		\caption{An Information example for \solutionMethodA{}.}\label{data_example}
	\end{figure}
	Figure~\ref{PBD_example} demonstrates the estimation process of \solutionMethodA{}.
	The total privacy budget for each user $u_i$ is evenly split into two parts, each containing $\epsilon_i/2$. 
	The first part is allocated for dissimilarity calculation, while the second is for publication noise calculation.
	For instance, $\epsilon_1$ is divided into $\vectorfont{\epsilon}_1^{(1)}(u_1)=\epsilon_1/2$ and $\vectorfont{\epsilon}_1^{(2)}(u_1)=\epsilon_1/2$.
	We compute the privacy budget usage $\epsilon_{i,t}^{(1)}$ for dissimilarity and $\epsilon_{i,t}^{(2)}$ for noise statistic publication for each user at each time slot. 
	These values are recorded in an $n\times 2$  matrix at each time slot in Figure~\ref{PBD_example}.
	Using $u_1$ as an example, $\epsilon_{1,t}^{(1)} = \vectorfont{\epsilon}_1^{(1)}(u_1)/w_1=\epsilon_1/8$. 
	At time slot $1$, $\epsilon_{1,1}^{(2)}=\vectorfont{\epsilon}_1^{(2)}(u_1)/2=\epsilon_1/4$.
	The algorithm calculates the dissimilarity $dis$ at time slot $1$ using all $\epsilon_{i,1}^{(1)}$, and the error $err_{opt}^{(2)}$ using all $\epsilon_{i,1}^{(2)}$.
	Assume $dis>\sqrt{err_{opt}^{(2)}}$, then a new obfuscated statistic $\vectorfont{r}_1$ is published at time slot $1$.
	At time slot $2$, assume $dis\leq\sqrt{err_{opt}^{(2)}}$, then $\epsilon_{i,2}^{(2)}$ is not used to publish a new obfuscated statistic result, and its usage is set to zeros for all users.
	At time slot $3$,  $\epsilon_{1,3}^{(2)}=(\epsilon_{1}/2-\epsilon_{1,1}^{(2)})/2=\epsilon_{1}/8$.
	The vector below each matrix in Figure~\ref{PBD_example} represents the total privacy budget used at the current time slot for each user. For example, at time slot $1$, the total privacy budget usage for $u_1$ is $\epsilon_{1,1}^{(1)}+\epsilon_{1,1}^{(2)}=3\epsilon_1/8$.
	\begin{figure}[t!]
		\centering
		\includegraphics[width=0.45\textwidth]{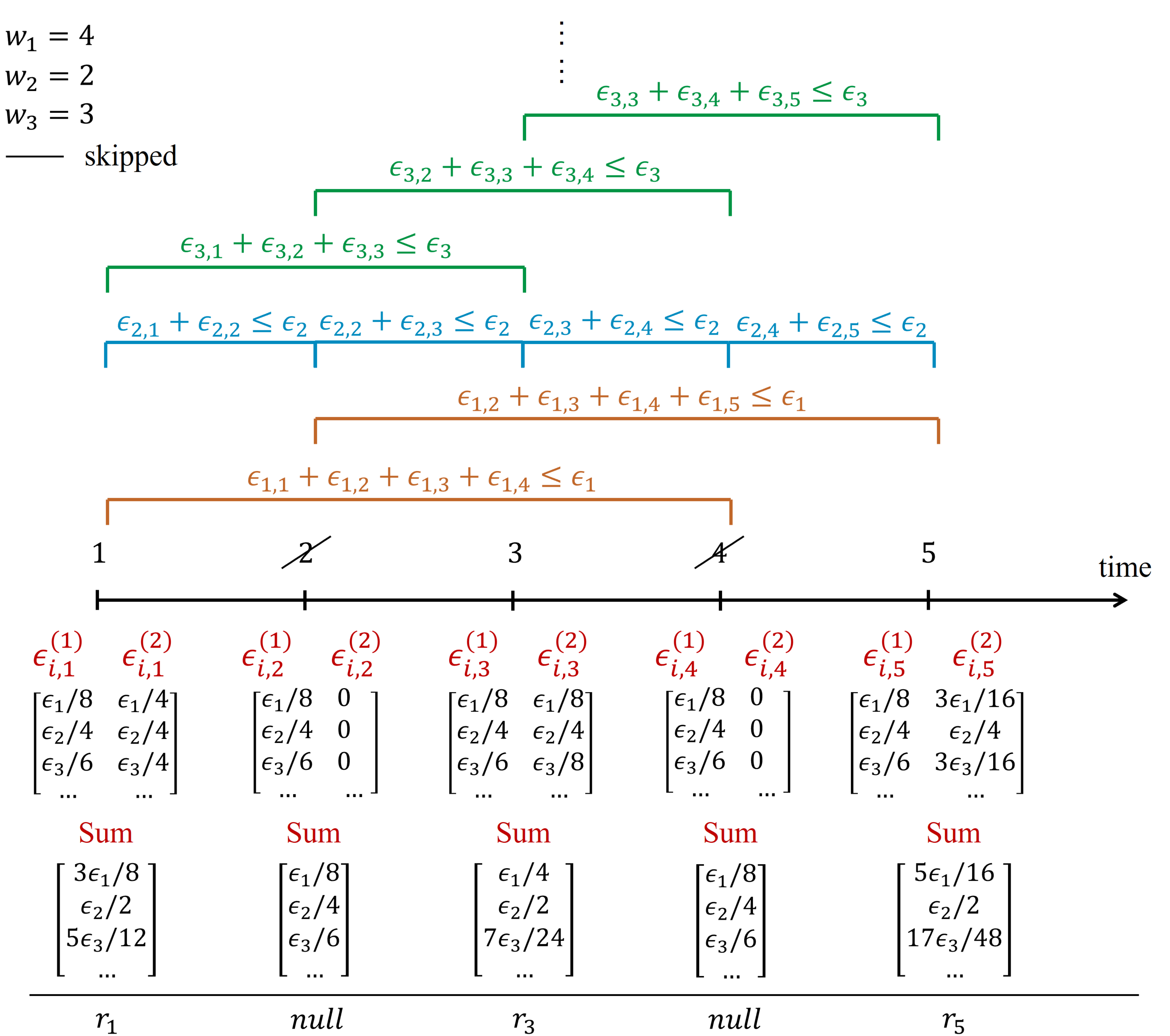}
		\caption{A process example for \solutionMethodA{}.}\label{PBD_example}
	\end{figure}
\end{example}

\noindent\textbf{\solutionMethodBTotalName{} (\solutionMethodB{}).} Algorithm~\ref{alg:PBA} outlines the process of \solutionMethodB{}.
The dissimilarity calculation (Part$_{DC}$) in \solutionMethodB{} is identical to that of \solutionMethodA{}. However, \solutionMethodB{} and \solutionMethodA{} differ significantly in their strategies on allocating the publication privacy budget.

For Part$_{NOP}$ in \solutionMethodB{}, we allocate an average privacy budget of $\frac{\epsilon_{i}}{2w_i}$ (one share) for each $u_i$ at each time slot $t$. 
A publication at time slot $t$ can use more than one share by borrowing from its successor time slots. The variable $t_{i,N}$ in Line~\ref{nullified} represents the number of successor time slots occupied by the last publication.
We calculate the maximal $\tilde{t}_{N}$ of all $t_{i,N}$ and determine whether the current time has been occupied ($t-l\leq\tilde{t}_N$). If so, we approximate the publication using the last published result. Otherwise, we calculate the remaining budget shares from the precursor time slots (i.e., $t_{A,i}$ in Line~\ref{absorbed}) and set the current publication budget as the total absorbed shares (Line~\ref{publicationBudget}).
The subsequent steps follow the same process as outlined in Algorithm~\ref{alg:PBD}.
\begin{algorithm}[t!]\small
	\caption{Personalized Budget Absorption}
	\label{alg:PBA}
	\DontPrintSemicolon
	\KwIn{$D_t$, EPDP privacy requirement ($\vectorfont{w}$, $\vectorfont{\epsilon}$), historical data publication $(\vectorfont{r}_1,\vectorfont{r}_2,\ldots, \vectorfont{r}_{t-1})$}
	\KwOut{$\vectorfont{r}_t$}
	Get the current window average budget $\bar{\epsilon}_i=\epsilon_i/w_i$ for each $i\in[n]$;\\
	$\vectorfont{\epsilon}_{t}^{(1)}=(\bar{\epsilon}_1/2,\bar{\epsilon}_2/2,\ldots ,\bar{\epsilon}_n/2)$;\\ \label{mt22_start}
	Get dissimilarity $dis$ by DC($D_t, \vectorfont{\epsilon}_{t}^{(1)}, \vectorfont{r}_1,...,\vectorfont{r}_{t-1}$) in Algorithm~\ref{alg:DC};\\
	Initialize nullified time slots $t_{i,N}$ as $0$;\\
	Set  $t_{i,N}=\frac{\epsilon_{i,l}^{(2)}}{\epsilon_i/(2w_i)}-1$ for $i\in[n]$ if $l$ exists where $l$ is the last non-null publishing time slot;\\ \label{nullified}
	Set nullified time slot bound $\tilde{t}_N=\max_{i\in[n]}t_{i,N}$;\\
	\If{$t-l\leq\tilde{t}_{N}$}{
		\Return $\vectorfont{r}_t=\vectorfont{r}_{t-1}$;
	} \Else{
		Set absorbed time slots $t_{A,i}=\max{(t-l-t_{i,N},0)}$ for $i\in[n]$;\\	\label{absorbed}
		Set publication budget $\epsilon_{i,t}^{(2)}=\frac{\epsilon_i}{2w_i}\cdot\min{(t_{A,i}, w_i)}$ for $i\in[n]$;\\	\label{publicationBudget}
		$\vectorfont{\epsilon}_{t}^{(2)}=\left(\epsilon_{1,t}^{(2)},\epsilon_{2,t}^{(2)},\ldots ,\epsilon_{n,t}^{(2)}\right)$;\\
		$\epsilon_{opt}^{(2)}$, $err_{opt}^{(2)}$ = OBS($\vectorfont{\epsilon}_{t}^{(2)}$);\\ \label{mt222_start}
		\If{$dis>\sqrt{err_{opt}^{(2)}}$}{
			$\tilde{D}_t^{(2)}=SM_s(D_t,\vectorfont{\epsilon}_{t}^{(2)},\epsilon_{opt}^{(2)})$;\\ 
			$\tilde{\vectorfont{c}}_t^{(2)}=Q(\tilde{D}_t^{(2)})$;\\
			\Return{$\vectorfont{r}_t=SM_d(\tilde{\vectorfont{c}}_t^{(2)},\epsilon_{opt}^{(2)})$};
		} \Else {
			\Return{$\vectorfont{r}_t=\vectorfont{r}_{t-1}$};
		}\label{mt222_end}
	}	
\end{algorithm}

\begin{example}\label{example_for_PBA}
	We continue use the demonstration case shown in Figure~\ref{data_example}. Figure~\ref{PBA_example} illustrates the estimation process of \solutionMethodB{}.
	The dissimilarity calculation process in \solutionMethodB{} is identical to that in Example~\ref{example_for_PBD}.
	For Part$_{NOP}$, at time slot $1$, with no budget to absorb, all users utilize one share (i.e., $\epsilon_i/(2w_i)$) to publish a new obfuscated statistic result.
	Assume time slot $2$ is skipped (i.e., $dis\leq\sqrt{err_{opt}^{(2)}}$). 
	At time slot $3$, $t_{1,N}=t_{2,N}=t_{3,N}=0$. Thus, the nullified bound $\tilde{t}_{N}$ is $0$.
	Since $t-l=3-1=2>\tilde{t}_{N}$, a new obfuscated statistic result is reported.
	The publication budget set is calculated as $\vectorfont{\epsilon}_3^{(2)}=\left(\epsilon_1/4,\epsilon_2/2,\epsilon_3/3,\ldots \right)$.
	At time slot $4$, $t_{1,N}=t_{2,N}=t_{3,N}=1$.
	As $t-l=4-3=1\leq\tilde{t}_{N}$, no output is produced.
	At time slot $5$, all $t_{i,N}$ remain $1$, and $t-l=5-3=2>\tilde{t}_{N}$. 
	The absorbed time slots $t_{A,i}$ all equal $1$.
	The resulting publication budget set is $\vectorfont{\epsilon}_5^{(2)}=\left(\epsilon_1/8,\epsilon_2/4,\epsilon_3/6,\ldots \right)$.
\end{example}

\begin{figure}[t!]
	\centering\vspace{-2ex}
	\includegraphics[width=0.45\textwidth]{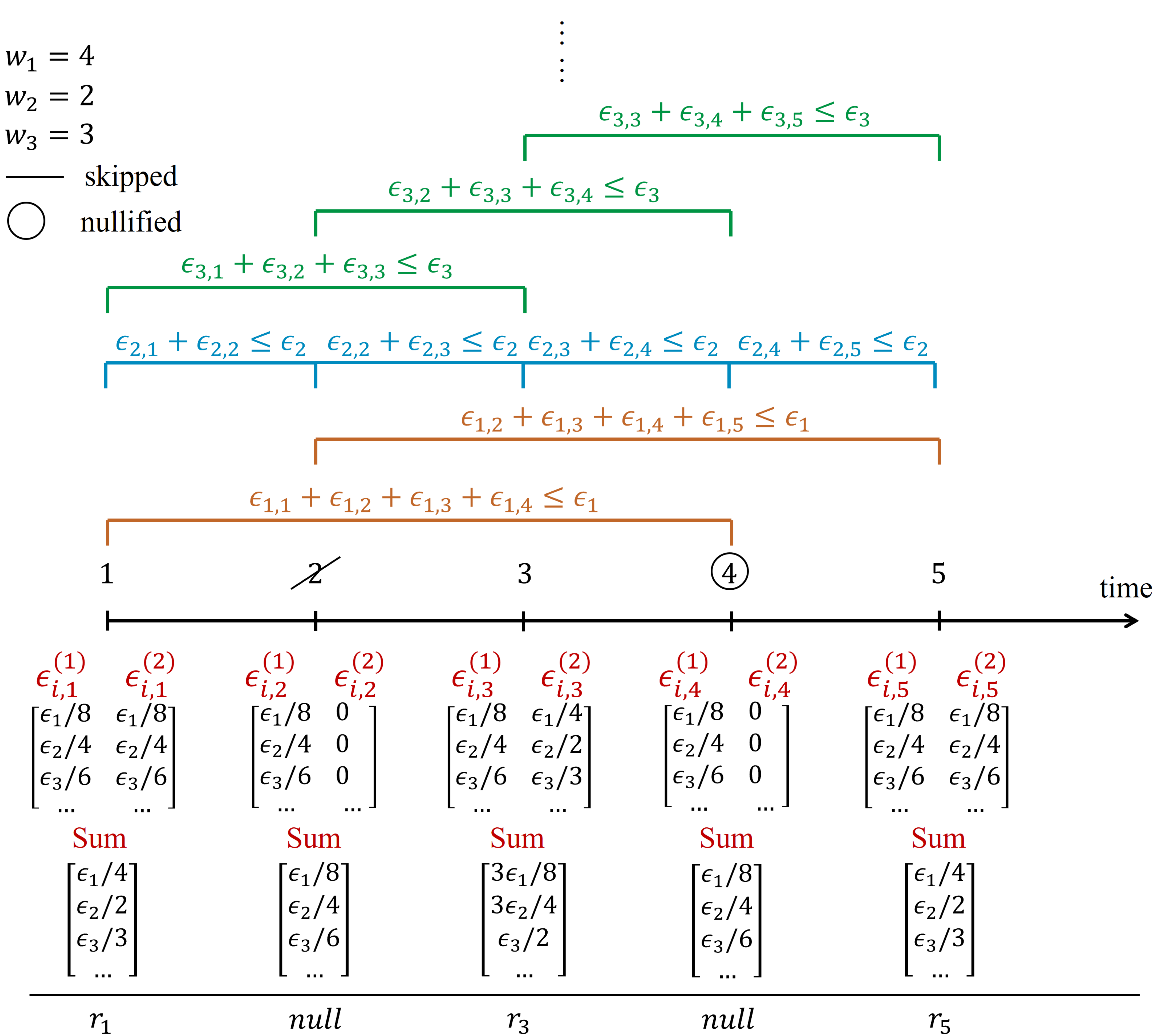}
	\caption{A process example for \solutionMethodB{}.}\label{PBA_example}
\end{figure}

\subsection{Analyses}

\noindent\textbf{Time Cost Analysis.}
Let $m$ be the number of distinct privacy requirements $(w_i,\epsilon_i)$, where $m\leq n$.
The time complexity of OBS is $O(m)$ for both \solutionMethodA{} and \solutionMethodB{}.
The Sample Mechanism and Query operations each have a time complexity of $O(n)$. Thus, the time complexities of \solutionMethodA{} and \solutionMethodB{} both are  $O(n)$.

\noindent\textbf{Privacy Analysis.} The privacy analysis for \solutionMethodA{} and \solutionMethodB{}:
\begin{theorem}\label{Thm:solutionAB_privacy_analysis}
	\solutionMethodA{} and \solutionMethodB{} satisfy $(\vectorfont{w},\vectorfont{\epsilon})$-EPDP.
\end{theorem}

\begin{proof}
	
	(1) \solutionMethodA{} satisfies $(\vectorfont{w},\vectorfont{\epsilon})$-EPDP.
	
	In the process of Part$_{DC}$, for each user $u_i$, the dissimilarity budget at each time slot
	is $\epsilon_i/(2w_i)$. Then for each time slot $t$, we have $\sum_{k=\max(t-w_i+1,1)}^{t} \epsilon_{i,k}^{(1)} = \epsilon_i/2.$
	%
	
	In Part$_{NOP}$, for each user $u_i$ at time slot $t$, only half of the publication
	budget is used when publication occurs: $\epsilon_{i,t}^{(2)}=(\epsilon_i/2-\sum_{k=\max(t-w_i+1,1)}^{t-1}\epsilon_{i,k}^{(2)})/2$. 
	For any time slot $t\in[1,w_i]$, the summation publication budgets used for $u_i$ is at most {\small$\sum_{k=1}^{w_i}\epsilon_i/(2\cdot 2^{k})\leq (\epsilon_i/2)\cdot (1-\frac{1}{2^{w_i}})\leq \epsilon_i/2$}. 
	Suppose {\small$\sum\limits_{k=\max(t-w_i+1,1)}^{t}\epsilon_{i,k}^{(2)}\leq \epsilon_i/2$} for $t=w_i+s$
	(i.e., $\sum\limits_{k=\max(s+1,1)}^{w_i+s} \epsilon_{i,k}^{(2)}$ $\leq \epsilon_i/2$).
	Then for $t=w_i+s+1$, we have:
	\begin{equation}\label{divide_equation}
		\sum_{k=\max(s+2,1)}^{w_i+s+1} \epsilon_{i,k}^{(2)} = \sum_{k=\max(s+2,1)}^{w_i+s} \epsilon_{i,k}^{(2)} + \epsilon_{i,w_i+s+1}^{(2)}. 
	\end{equation}
	Since $\epsilon_{i,w_i+s+1}^{(2)}$ is at most half of the remaining publication budget at time slot $w_i+s$:
	\begin{equation}\label{distribution_equation}
		\epsilon_{i,w_i+s+1}^{(2)} \leq (\epsilon_i/2-\sum_{k=\max(s+2,1)}^{w_i+s} \epsilon_{i,k}^{(2)})/2.
	\end{equation}
	According to Equations~\eqref{divide_equation} and~\eqref{distribution_equation}, we have:
	\begin{equation}\notag
		\begin{aligned}
			\sum_{k=\max(s+2,1)}^{w_i+s+1} \epsilon_{i,k}^{(2)} &\leq \sum_{\substack{k=\max(\\s+2,1)}}^{w_i+s} \epsilon_{i,k}^{(2)} + (\epsilon_i/2-\sum_{\substack{k=\max(\\s+2,1)}}^{w_i+s} \epsilon_{i,k}^{(2)})/2 \\
			&=\epsilon_i/4 + (\sum_{k=\max(s+2,1)}^{w_i+s} \epsilon_{i,k}^{(2)})/2\\
			&\leq \epsilon_i/4 + \epsilon_i/4\\
			&=\epsilon_i/2.
		\end{aligned}
	\end{equation}
	Therefore, for any $t\geq 1$, we have:
	\begin{equation}\notag
		\sum_{k=\max{(t-w_i+1,1)}}^{t} \epsilon_{i,k}^{(2)} \leq \epsilon_i/2.
	\end{equation}
	According to the Composition Theorems~\cite{DBLP:journals/fttcs/DworkR14}, we have:
	\begin{equation}\notag
		\begin{aligned}
			\sum_{k=\max(t-w_i+1,1)}^{t} \epsilon_{i,k} &= \sum_{\substack{k=\max(t\\-w_i+1,1)}}^{t} \epsilon_{i,k}^{(1)} + \sum_{\substack{k=\max(t\\-w_i+1,1)}}^{t} \epsilon_{i,k}^{(2)} \leq \epsilon_i.
		\end{aligned}
	\end{equation}
	
	For any user $u_i$ and any two $w_i$-neighboring stream prefixes $S_t$ and $S_t'$ (i.e., $S_t\sim_{w_i}S_t'$), let $t_s$ be the earliest time slot where $S_t[t_s]\neq S_t'[t_s]$ and $t_e$ be the latest time slot where $S_t[t_e]\neq S_t'[t_e]$. Then we have $t_e-t_s+1\leq w_i$. Denoting the output of our \solutionMethodA{} as $\solutionMethodA{}(S_t[t])=o_{t}\in \entity{O}$, for any $O\subseteq \entity{O}$, we have:
	\begin{equation}\notag
		\begin{aligned}
			\frac{\Pr[\solutionMethodA{}(S_t)]\in O}{\Pr[\solutionMethodA{}(S_t')]\in O} &\leq \Pi_{k=t_s}^{t_e}\frac{\Pr[\solutionMethodA{}(S_t[k])=o_{k}]}{\Pr[\solutionMethodA{}(S_t'[k])=o_{k}]}\\
			&\leq e^{\sum_{k=t_s}^{t_e}\epsilon_{i,k}} \leq e^{\sum_{k=\max{(t_e-w_i+1,1)}}^{t_e}\epsilon_{i,k}} 
			\leq e^{\epsilon_i}.
		\end{aligned}
	\end{equation}
	Therefore, \solutionMethodA{} satisfies $(\vectorfont{w},\vectorfont{\epsilon})$-EPDP where $\vectorfont{w}=(w_1,w_2,\ldots ,w_n)$ and $\vectorfont{\epsilon}=((u_1,\epsilon_1),(u_2,\epsilon_2),\ldots ,(u_n,\epsilon_n))$.
	
	(2) \solutionMethodB{} satisfies $(\vectorfont{w},\vectorfont{\epsilon})$-EPDP.
	
	The Part$_{DC}$ in \solutionMethodB{} is identical to that that in \solutionMethodA{}. Consequently, for each time slot $t$, we have:
	\begin{equation}\label{absorb_part1}
		\sum_{k=\max(t-w_i+1,1)}^{t} \epsilon_{i,k}^{(1)} = \epsilon_i/2.
	\end{equation}
	
	In Part$_{NOP}$, for any user $u_i$ and any window of size $w_i$, there are $s_i$ publication time slots in the window.
	We denote these publication time slots as $(k_1, k_2,\ldots ,k_{s_i})$.
	For any publication time slot $k_j$ ($j\in[s_i]$), the quantity of its absorbing unused budgets is denoted as $\alpha_{i,k_j}$. 
	Figure~\ref{PBD_proof} illustrates an example where $s_i=3$ and $w_i=9$. 
	\begin{figure}[t!]\vspace{-2ex}
		\centering
		\includegraphics[width=0.45\textwidth]{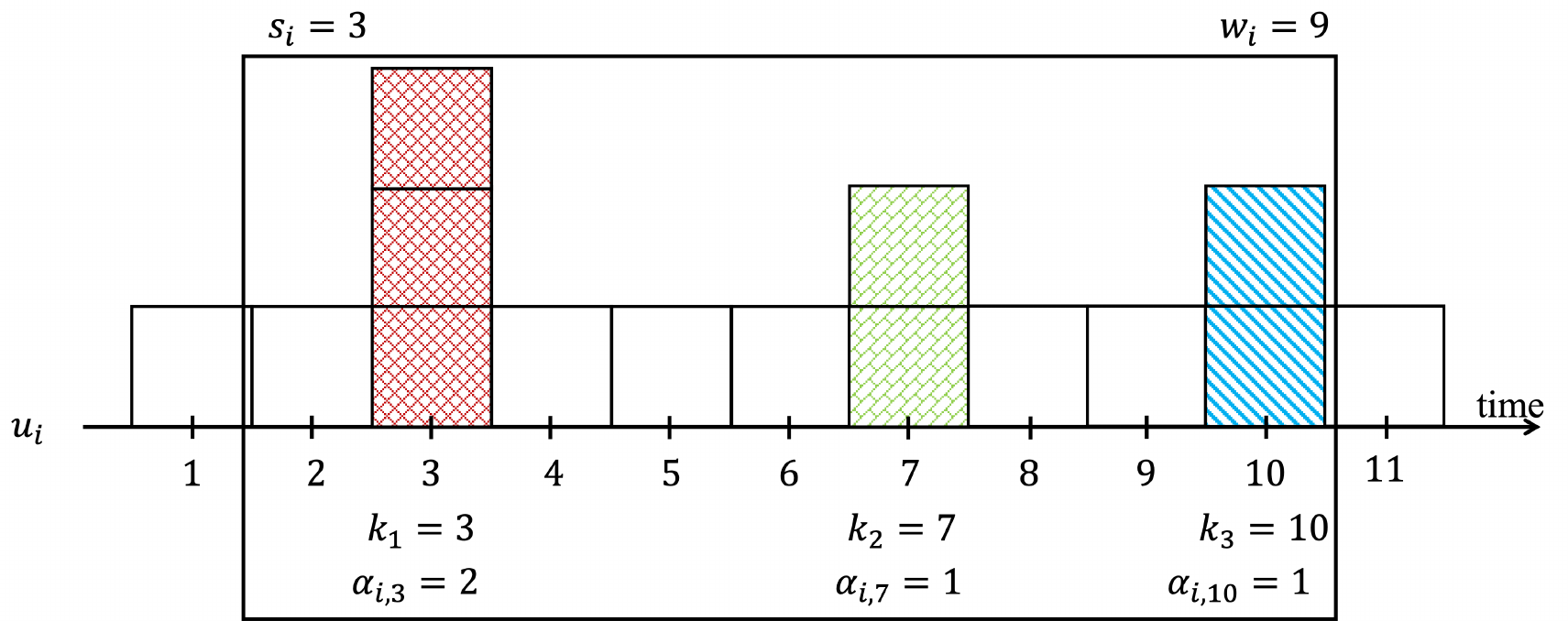}
		\caption{An example for parameters in \solutionMethodB{}.}\label{PBD_proof}
	\end{figure}
	
	Based on Algorithm~\ref{alg:PBA}, we have:
	\begin{equation}\notag
		w_i\geq\sum_{j=1}^{s_i}(1+2\alpha_{i,k_j})-\alpha_{i,k_1}-\alpha_{i,k_{s_i}}.
	\end{equation}
	Then, for the total publication budgets used in any window, we have
	\begin{equation}\label{absorb_part2}
		\begin{aligned}
			\sum_{k=\max{(t-w_i+1,1)}}^{t} \epsilon_{i,k}^{(2)} &\leq \frac{\epsilon_i}{2w_i}\cdot \sum_{j=1}^{s_i}(1+\alpha_{i,k_j})\\
			&\leq \frac{\epsilon_i\cdot\sum_{j=1}^{s_i}(1+\alpha_{i,k_j})}{2\sum_{j=1}^{s_i}(1+2\alpha_{i,k_j})-2\alpha_{i,k_1}-2\alpha_{i,k_{s_i}}}\\
			&= \frac{\epsilon_i\cdot\sum_{j=1}^{s_i}(1+\alpha_{i,k_j})}{2\sum_{j=1}^{s_i}(1+\alpha_{i,k_j})+2\sum_{j=2}^{s_i-1}\alpha_{i,k_j}}\\
			&\leq \epsilon_i/2.
		\end{aligned}
	\end{equation}
	Based on Equations~\eqref{absorb_part1} and~\eqref{absorb_part2}, and applying the Composition Theorems~\cite{DBLP:journals/fttcs/DworkR14}, we obtain:
	\begin{equation}\notag
		\begin{aligned}
			\sum_{k=\max{(t-w_i+1,1)}}^{t} \epsilon_{i,k} &= \sum_{\substack{k=\max(t\\-w_i+1,1)}}^{t} \epsilon_{i,k}^{(1)} + \sum_{\substack{k=\max(t\\-w_i+1,1)}}^{t} \epsilon_{i,k}^{(2)} \leq \epsilon_i.
		\end{aligned}
	\end{equation}
	The subsequent proof process follows the same steps as in \solutionMethodA{}.
	Ultimately, we demonstrate that \solutionMethodB{} also satisfies $(\vectorfont{w},\vectorfont{\epsilon})$-EPDP.
\end{proof}

\textbf{Utility Analysis.}
For each user $u_i$ in \solutionMethodA{} and \solutionMethodB{}, we define $w_L$ as the smallest window size among all users.
For each $u_i$, given $(w_i, \epsilon_i)$, let $\epsilon_{L}=\min_{i\in[n]}\frac{\epsilon_i}{w_i}$ and $\epsilon_{R}=\max_{i\in[n]}\frac{\epsilon_i}{w_i}$ represent the minimum and maximum values of $\frac{\epsilon_i}{w_i}$, respectively.
Let $n_A$ be the number of occurrences of $\epsilon_R$ across all users.

We make the following assumptions: At most $\tilde{s}\leq w_L$ publications occur at time slots $q_1$, $q_2$,\ldots, $q_{\tilde{s}}$ in the window of size $w_L$, with no budget absorption from past time slots outside the window.
Additionally, for each user, each publication approximates the same number of skipped or nullified publications.

We first present a crucial lemma, followed by two theorems that bound the average errors of \solutionMethodA{} and \solutionMethodB{}, respectively.
\begin{lemma}\label{lemma:sm_utility}
	Given $m$ distinct privacy budget-quantity pairs $P=\{(\epsilon_{j},n_j)|j\in[m],\sum_{j\in[m]}n_j=n\}$ where pair $(\epsilon_{j},n_j)$ indicates that $\epsilon_{j}$ appears $n_j$ times in the user privacy requirement, and a query with sensitivity $I$, the error upper bound $\widetilde{err}_O(P)$ of the SM process with privacy budget chosen from OBS is: {\scriptsize$$\min{\left({\frac{2I^2}{\min_j\epsilon_j^2}},(n-n_{M})(n-n_{M}+\frac{1}{4})+{\frac{2I^2}{\max_j\epsilon_j^2}}\right)},$$}
	where $n_{M}=n_k$ with $k=\arg\max_{j\in[m]}{\epsilon_j}$.
\end{lemma}
\begin{proof}
	Let $M_L$ be the SM with privacy budget chosen as $\min_j{\epsilon_j}$.
	According to the SM process, all budget types will be selected. 
	In this case, the sampling error $err_s$ is $0$ and the noise error $err_{dp}$ is $2\cdot(\frac{I}{\min_j{\epsilon_j}})^2=\frac{2I^2}{\min_j{\epsilon_j^2}}$.
	Thus, the total error of $M_L$ is $err_{M_L}=\frac{2I^2}{\min_j{\epsilon_j^2}}$.
	Let $M_R$ be the SM with privacy budget chosen as $\max_j{\epsilon_j}$. In this case, $(m-1)$ types of privacy budget are chosen with probability $p_k=\frac{e^{\epsilon_k}-1}{e^{\max_j{\epsilon_j}}-1}$ less than $1$ ($k\in[m]$). For the sampling error, we have:
	\begin{equation}\notag
		\begin{aligned}
			err_s &= \sum_{\epsilon_k<\max_j{\epsilon_j}} n_k p_k(1-p_k) + \left(\sum_{\epsilon_k<\max_j{\epsilon_j}} n_k(1-p_k)\right)^2\\
			&< \sum_{\epsilon_k<\max_j{\epsilon_j}} n_k \left(\frac{p_k+1-p_k}{2}\right)^2 + \left(\sum_{\epsilon_k<\max_j{\epsilon_j}} n_k\right)^2\\
			&=(n-n_{M})(n-n_{M}+\frac{1}{4}).
		\end{aligned}
	\end{equation}
	The noise error $err_{dp}$ in this case is $2\cdot(\frac{I}{\max_j{\epsilon_j}})^2=\frac{2I^2}{\max_j{\epsilon_j^2}}$.
	Thus, the total error of $M_R$ is $err_{M_R}=(n-n_M)(n-n_M+\frac{1}{4})+\frac{2I^2}{\max_j{\epsilon_j^2}}$.
	According to the OBS process, we have $\widetilde{err}_{O}(P)\leq err_{M_L}$ and $\widetilde{err}_{O}(P)\leq err_{M_R}$.
	Therefore, 
	\begin{equation}\notag
		\begin{aligned}
			\widetilde{err}_{O}(P) &\leq\min{(err_{M_L},err_{M_R})}=\min{\left(\frac{2I^2}{\min_j{\epsilon_j^2}},(n-n_M)(n-n_M+\frac{1}{4})+\frac{2I^2}{\max_j{\epsilon_j^2}}\right)}.
		\end{aligned}
	\end{equation}
\end{proof}

For \solutionMethodA{}, we present Theorem~\ref{Thm:solutionMethodA_utility_analysis} as follows.
\begin{theorem}\label{Thm:solutionMethodA_utility_analysis}
	The average error per time slot in \solutionMethodA{} is at most {\scriptsize$\min{\left(\frac{8}{d^2\epsilon_L},Z+\frac{8}{d^2\epsilon_R}\right)}+\min{\left(\frac{32\cdot(4^{\tilde{s}}-1)}{3\tilde{s}\epsilon_L},Z+\frac{32\cdot(4^{\tilde{s}}-1)}{3\tilde{s}\epsilon_R}\right)}$} where {\scriptsize$Z=(n-n_A)(n-n_A+\frac{1}{4})$}, if at most $\tilde{s}$ publications occur in any window with size $w_L$.
\end{theorem}

\begin{proof}
	Given a privacy budget-quantity pair set $P$, let $EOPT(P)$ be the optimal privacy budget chosen from OBS.
	Given a positive number $\beta$, we define $\beta\cdot P=\{(\beta\cdot\epsilon_j,n_j)|(\epsilon_j,n_j)\in P\}$.
	For each user $u_i$ with privacy requirement pair $(w_i,\epsilon_i)$, we calculate their average budget per window as $\frac{\epsilon_i}{w_i}$. We denote the set of all average budgets as $\overline{\epsilon}=\{\frac{\epsilon_i}{w_i}|i\in[n]\}$.
	We then construct the privacy budget-quantity pair set of each type of average budget as $P_A=\{(\epsilon_j, n_j)|\epsilon_j\in\overline{\epsilon}\}$.
	Let $Z=(n-n_{A})(n-n_{A}+\frac{1}{4})$ be the sampling error upper bound, where $n_{A}$ is the quantity of $\max_{i\in[n]}\frac{\epsilon_i}{w_i}$ in $\overline{\epsilon}$.
	
	When Part$_{DC}$ is not private, the error stems from Part$_{NOP}$.
	In Part$_{NOP}$, errors arise from both publications and approximations.
	According to the Part$_{NOP}$, an approximation error does not exceed the publication error at the most recent publication time slot.
	For the average error $\overline{err}_{NOP}$ of all time slots within the window of size $w_L$, based on the \solutionMethodA{} process, we have: 
	\begin{equation}\label{equation:pbd_m2}
		\begin{aligned}
			\overline{err}_{NOP} &= \frac{1}{w_L}\sum_{k\in[\tilde{s}]}\frac{w_L}{\tilde{s}}\cdot \widetilde{err}_O\left(\frac{1}{2^{k+1}}P_A\right)\\
			&< \frac{1}{\tilde{s}}\sum_{k\in[\tilde{s}]}\min{\left(\frac{2}{(\frac{\epsilon_L}{2^{k+1}})^2},Z+\frac{2}{(\frac{\epsilon_R}{2^{k+1}})^2}\right)}\\
			&< \frac{1}{\tilde{s}}\min\left(\sum_{k\in[\tilde{s}]}\frac{8\cdot4^k}{\epsilon_L^2}, \tilde{s}\cdot Z+\sum_{k\in[\tilde{s}]}\frac{8\cdot 4^k}{\epsilon_R^2}\right)\\
			&= \min\left(\frac{32\cdot(4^{\tilde{s}}-1)}{3\tilde{s}\epsilon_L^2}, Z+\frac{32\cdot (4^{\tilde{s}}-1)}{3\tilde{s}\epsilon_R^2}\right).
		\end{aligned}
	\end{equation}
	
	When Part$_{DC}$ is private, the error from Part$_{DC}$ can lead to two scenarios: (1) falsely skipping a publication or (2) falsely performs a publication.
	Both cases are bounded by the error in Part$_{DC}$.
	In Part$_{DC}$, we execute the SM with OBS. The sensitivity of $dis$ is $1/d$.
	For the average error $\overline{err}_{DC}$ of each time slot in window size $w_L$,
	according to Lemma~\ref{lemma:sm_utility}, we have:
	\begin{equation}\label{equation:pbd_m1}
		\begin{aligned}
			\overline{err}_{DC} &< \min{\left({\frac{2}{d^2\min_{i\in[n]}(\frac{\epsilon_i}{2w_i})^2}},Z+{\frac{2}{d^2\max_{i\in[n]}(\frac{\epsilon_i}{2w_i})^2}}\right)}\\
			&= \min{\left({\frac{8}{d^2\epsilon_L^2}},Z+{\frac{8}{d^2\epsilon_R^2}}\right)}.
		\end{aligned}
	\end{equation}
	
	Based on Equation~\eqref{equation:pbd_m1} and~\eqref{equation:pbd_m2}, we can get the average error upper bound as $\overline{err}_{DC}+\overline{err}_{NOP}$.
	
\end{proof}

\solutionMethodA{} achieves low error when the number of publications $\tilde{s}$ per window is small.
However, the error increases exponentially with $\tilde{s}$.
Additionally, the error in Part$_{DC}$ (the first part of the error upper bound in \solutionMethodA{}) rises as $w_L$ increases, however, it diminishes as $d$ increases.
This is because a large $d$ reduces sensitivity leading to smaller noise error.

For \solutionMethodB{}, assume $\alpha$ skipped publications occur before a publication. 
Let $\epsilon_{\tilde{L}}$ and $\epsilon_{\tilde{R}}$ be the minimum and maximum publication privacy budget among all users at time slots $t=w_{L}$ and  $t=(\alpha+1)$, respectively. 
According to the \solutionMethodB{} process, there will be $\alpha$ nullified publications after the publication.
These nullified publications are filled by the last time slot's publication without comparison.
Consequently, the nullified publication error depends on the data distribution at nullified time slots.
We denote the average error of each nullified publication in \solutionMethodB{} as $\overline{err}_{nlf}$.
For \solutionMethodB{}, we have Theorem~\ref{Thm:solutionMethodB_utility_analysis} as follows.
\begin{theorem}\label{Thm:solutionMethodB_utility_analysis}
	The average error per time slot in \solutionMethodB{} is at most $\min(\frac{8}{d^2\epsilon_L},Z+\frac{8}{d^2\epsilon_R})+\frac{1}{2\alpha+1}(\widetilde{err}_{NOP}^{(s,p)}+\alpha\cdot\overline{err}_{nlf})$ 
	where $\widetilde{err}_{NOP}^{(s,p)}$ is $\min(\frac{2}{\epsilon_L^2}H^2_{\alpha+1},(\alpha+1)Z+\frac{2}{\epsilon_R^2}H^2_{\alpha+1})$ when $\alpha\leq w_{L}$ and $\min(\frac{2}{\epsilon_L^2}H^2_{w_L},w_L Z+\frac{2}{\epsilon_R^2}H^2_{w_L})+ (\alpha-w_L+1)\min(\frac{2}{\epsilon_{\tilde{L}}^2}, Z+\frac{2}{\epsilon_{\tilde{R}}^2})$ when $\alpha>w_{L}$ and
	$Z=(n-n_A)(n-n_A+\frac{1}{4})$ and $H^2_{x}$ is the $x$-th square harmonic number, if there are $\alpha$ skipped publications occur in average before each publication.
\end{theorem}

\begin{proof}
	Similar to \solutionMethodA{}, we first analyze the error of Part$_{NOP}$ in \solutionMethodB{} by assuming Part$_{DC}$ is not private.
	We then add the error of Part$_{DC}$, which is identical to that in \solutionMethodA{}, to obtain the final total error.
	When Part$_{DC}$ is not private, the error stems from Part$_{NOP}$.	
	In Part$_{NOP}$, each publication corresponds to $\alpha$ skipped publications preceding it and $\alpha$ nullified publications succeeding it.
	
	For each user $u_i$'s skipped publication, the publication privacy budget lower bound doubles with each time slot increase until it reaches $\epsilon_i/2$ or a publication occurs.
	For example, in Figure~\ref{PBD_publication_lower bound_proof}, where $\alpha=5$, the publication time slot is $t_6$.
	At time slot $t_1$, each $u_i$'s publication budget lower bound is $\epsilon_i/(2w_i)$.
	Take $u_1$ as an example: it reaches $\epsilon_1/2$ at time slot $t_4$. 
	The publication lower bound for $u_1$ remains at $\epsilon_1/2$ until time slot $t_6$.
	\begin{figure}[t!]
		\centering
		\includegraphics[width=0.45\textwidth]{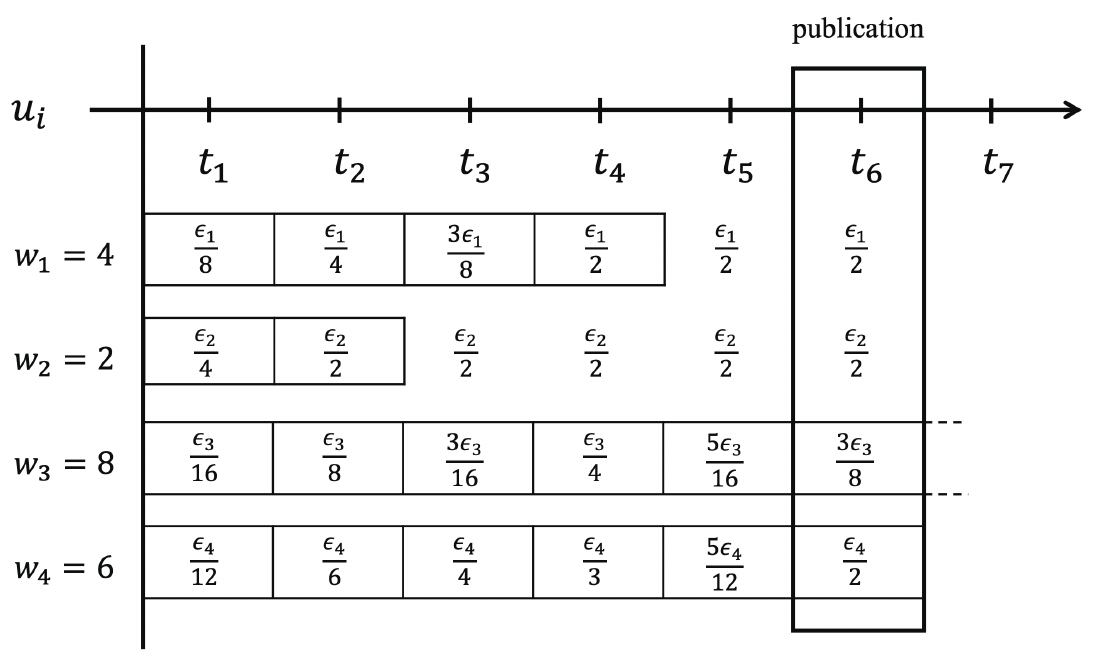}
		\caption{An example of the publication budget lower bound in \solutionMethodB{}.}\label{PBD_publication_lower bound_proof}
	\end{figure}
	Let the publication budget lower bound set for all users at skipped time slots (spanning $\alpha$ time slot) be $\hat{\vectorfont{\epsilon}}=\{\vectorfont{\epsilon}_1,\vectorfont{\epsilon}_2,\ldots ,\vectorfont{\epsilon}_{\alpha}\}$.
	Then, the error upper bound of each skipped publication is the error of publishing new data using $\vectorfont{\epsilon}_{k}$ ($k\in[\alpha]$).
	For example in Figure~\ref{PBD_publication_lower bound_proof}, the error upper bound at $t_3$ is the error of publication a new obfuscated statistic result using $\{\frac{3\epsilon_1}{2},\frac{\epsilon_2}{2},\frac{3\epsilon_3}{16},\frac{\epsilon_4}{4}\}$.
	
	Let $Z=(n-n_{A})(n-n_{A}+\frac{1}{4})$ be the sampling error upper bound, where $n_{A}$ is the number of users with maximum value of $\frac{\epsilon_i}{w_i}$.
	We now consider two cases: $\alpha\leq w_{L}$ and $\alpha > w_{L}$.
	
	\textbf{(1) case 1}: $\alpha\leq w_{L}$.
	
	In this case, the publication budget lower bound doubles with each time slot increase.
	Let $err_{NOP}^{(sk)}(\alpha)$ and $err_{NOP}^{(pb)}$ be the total error upper bounds of the $\alpha$ skipped publications and the publication in Part$_{NOP}$, respectively.
	Let $err_{NOP}^{(s,p)}$ be the error of all skipped publications and the publication in Part$_{NOP}$.
	According to Lemma~\ref{lemma:sm_utility}, we have
	\begin{equation}\label{skipped_error}
		\begin{aligned}
			err_{NOP}^{(sk)}(\alpha) &< \sum_{k\in[\alpha]}\min{\left(\frac{2}{(k\epsilon_L)^2},Z+\frac{2}{(k\epsilon_R)^2}\right)} \\
			&\leq \min{\left(\frac{2}{\epsilon_L^2}H^{2}_{\alpha},\alpha Z+\frac{2}{\epsilon_R^2}H^{2}_{\alpha}\right)} \\
		\end{aligned}
	\end{equation}
	and 
	\begin{equation}\label{skipped_publication_error}
		\begin{aligned}
			err_{NOP}^{(s,p)} &< err_{NOP}^{(sk)}(\alpha) + err_{NOP}^{(pb)} =err_{NOP}^{(sk)}(\alpha+1)\\
			&= \min{\left(\frac{2}{\epsilon_L^2}H^{2}_{\alpha+1},(\alpha+1)Z+\frac{2}{\epsilon_R^2}H^{2}_{\alpha+1}\right)}.
		\end{aligned}
	\end{equation}
	Thus, we derive the average error upper bound $\overline{err}_{NOP}$ of each time slot in Part$_{NOP}$ as 
	\begin{equation}\label{sp_err_1}
		\begin{aligned}
			\overline{err}_{NOP} &< \frac{1}{2\alpha+1}(\widetilde{err}_{NOP}^{(s,p)}+\alpha\cdot\overline{err}_{nlf}),
		\end{aligned}
	\end{equation}
	where $\widetilde{err}_{NOP}^{(s,p)}$ is the final value in Equation~\eqref{skipped_publication_error}.
	
	\textbf{(2) case 2}: $\alpha>w_{L}$.

	In this case, we have
	\begin{equation}\label{skipped_publication_error2}
		\begin{split}
			err_{NOP}^{(s,p)} <& err_{NOP}^{(sk)}(w_L) + \sum_{k=w_L+1}^{\alpha+1}\min{\left(\frac{2}{\epsilon_{\tilde{L}}^2},Z+\frac{2}{\epsilon_{\tilde{R}}^2}\right)} \\
			=& err_{NOP}^{(sk)}(w_L) + (\alpha-w_L+1)\min{\left(\frac{2}{\epsilon_{\tilde{L}}^2}, Z+\frac{2}{\epsilon_{\tilde{R}}^2}\right)}\\
			<& \min{\left(\frac{2}{\epsilon_L^2}H^{2}_{w_L},w_L Z+\frac{2}{\epsilon_R^2}H^{2}_{w_L}\right)} + (\alpha-w_L+1)\min{\left(\frac{2}{\epsilon_{\tilde{L}}^2}, Z+\frac{2}{\epsilon_{\tilde{R}}^2}\right)}.
		\end{split}
	\end{equation}
	Therefore, we obtain the average error upper bound $\overline{err}_{NOP}$ for each time slot in Part$_{NOP}$ as 
	\begin{equation}\label{sp_err_2}
		\begin{aligned}
			\overline{err}_{NOP} &< \frac{1}{2\alpha+1}(\widetilde{err}_{NOP}^{(s,p)}+\alpha\cdot\overline{err}_{nlf})
		\end{aligned}
	\end{equation}
	where $\widetilde{err}_{NOP}^{(s,p)}$ is the value derived in Equation~\eqref{skipped_publication_error2}.
	
	When Part$_{DC}$ is private, its error is identical to that in \solutionMethodA{}:
	\begin{equation}\label{methodB_m_1_error}
		\overline{err}_{DC}<\min{\left({\frac{8}{d^2\epsilon_L^2}},Z+{\frac{8}{d^2\epsilon_R^2}}\right)}.
	\end{equation}
	
	Based on Equation~\eqref{methodB_m_1_error}, \eqref{sp_err_1} and \eqref{sp_err_2}, we can derive the average error upper bound for each time slot in \solutionMethodB{} as:
	\begin{equation}\notag
		\begin{aligned}
			\min{\left({\frac{8}{d^2\epsilon_L^2}},Z+{\frac{8}{d^2\epsilon_R^2}}\right)} + \frac{1}{2\alpha+1}(\widetilde{err}_{NOP}^{(s,p)}+\alpha\cdot\overline{err}_{nlf}),
		\end{aligned}
	\end{equation}
	where $\widetilde{err}_{NOP}^{(s,p)}$ is the final result from Equation~\eqref{skipped_publication_error} when $\alpha\leq w_L$, and from Equation~\eqref{skipped_publication_error2} when $\alpha> w_L$.
\end{proof}

\section{Experiments}\label{experiment}
\begin{figure}[t!]\centering
	\subfigure[][{\scriptsize Taxi}]{
		\scalebox{0.26}[0.26]{\includegraphics{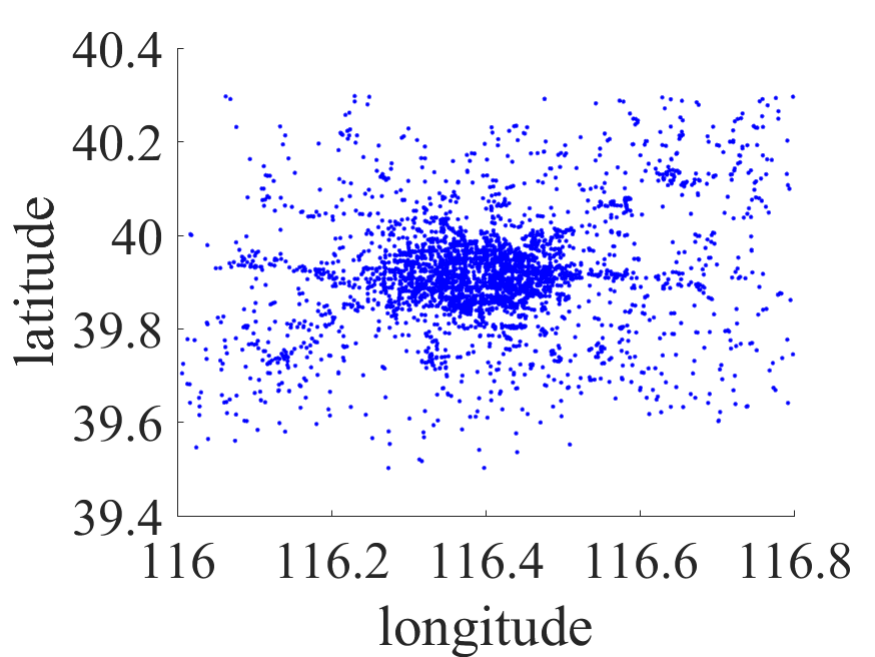}}
		\label{subfig:trajectory}}
	\subfigure[][{\scriptsize Foursquare}]{
		\scalebox{0.26}[0.26]{\includegraphics{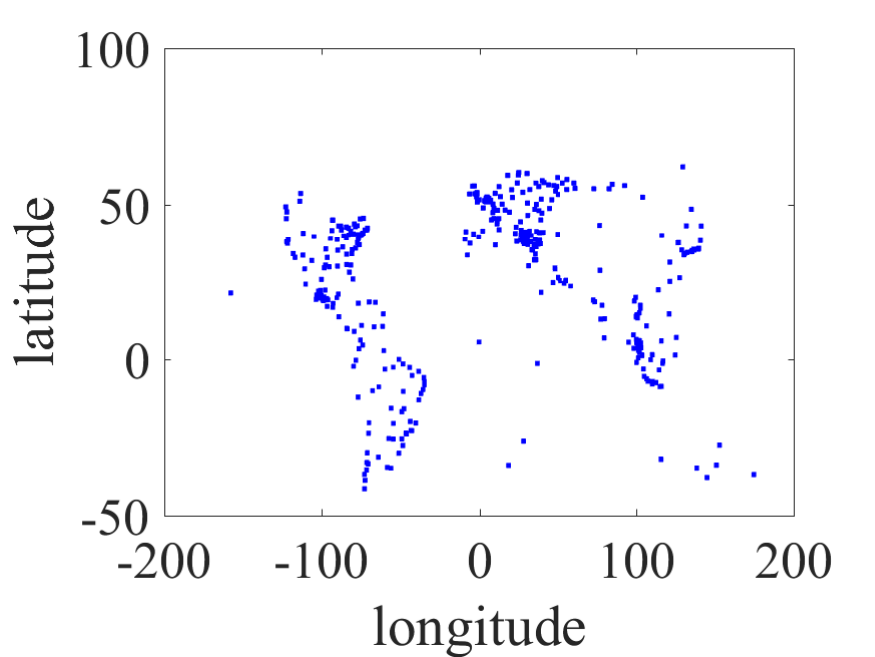}}
		\label{subfig:checkin}}
	\caption{Illustration of Real datasets.}
	\label{fig:real_dataset}
\end{figure}
\subsection{Datasets}

\noindent
\textbf{Real datasets.} We use two real-world datasets, \textit{\trajectoryDatasetName{}}~\cite{DBLP:conf/kdd/YuanZXS11,DBLP:conf/gis/YuanZZXXSH10} and \textit{\checkInDatasetName{}}~\cite{DBLP:journals/tist/YangZQ16,DBLP:journals/jnca/YangZCQ15}, to evaluate the performance of our algorithms. 

\textit{\trajectoryDatasetName{}.} It contains real-time trajectories of $10,357$ taxis' in Beijing from February $2$ to February $8$, 2008.
Each taxi has up to $154,699$ records, where each record comprises~\textit{taxi id}, \textit{data time}, \textit{longitude} and \textit{latitude}. 
For the spatial dimension, we first remove all duplicate records, then extract records with longitude between $116$ and $116.8$ and latitude between $39.5$ and $40.3$, resulting in $14,859,377$ records. We denote this area ($[116,116.8]\times[39.5,40.3]$) as $A_{E}$. 
Figure~\ref{subfig:trajectory} shows $50\%$ of uniformly extracted trajectory points in $A_{E}$.
We further divide $A_{E}$ uniformly into a $10\times 10$ grids, designating these $100$ cells as the location space.
For the time dimension, we sample records every minute and get $8,889$ records.

\textit{\checkInDatasetName{}.} It contains $33,278,683$ Foursquare check-ins from $266,909$ users, during April 2012 to September 2013. 
Each record consists of user id, venue id (place), and time. 
We convert the venue id to the country where the venue is located.
After removing invalid records, we uniformly extract $5\%$ of users' check-ins as shown in Figure \ref{subfig:checkin}.
We set the publication time interval to $100$ minutes, thus divide the chick-ins period into $7,649$ time slots.

\noindent\textbf{Synthetic datasets.} We generate three binary stream datasets using different sequence models. 
We set the length of each binary stream as $T$ and the number of users as $N$.
For each stream, we first generate a probability sequence $(p_1, p_2,..., p_T)$.
At each time slot $t$, each user's real value is set to $1$ with probability $p_t$ and $0$ otherwise.
The probability function we use are as follows:
\begin{itemize}[leftmargin=*]
	\item \tlnsDatasetName{} function. In TLNS, $p_t=p_{t-1}+\algvar{N}(0,Q)$, where $\algvar{N}(0,Q)$ is Gaussian noise with standard variance $\sqrt{Q}=0.0025$. We set $p_0=0.05$ as the initial value.
	If $p_t<0$, we set $p_t=0$; If $p_t>1$, we set $p_t=1$.
	\item \sinDatasetName{} function. In Sin, $p_t=A\sin{(\omega t)}+h$, where $A=0.05$, $\omega=0.01$ and $h=0.075$.
	\item \logDatasetName{} function. In Log, $p_t=A/(1+e^{-bt})$, where $A=0.25$ and $b=0.01$.
\end{itemize}

\subsection{Experiment Setup}
We compare our \solutionMethodA{} and \solutionMethodB{} with two non-personalized methods:  \solutionCMPATotalName{} (\solutionCMPA{}) and \solutionCMPBTotalName{} (\solutionCMPB{})~\cite{DBLP:journals/pvldb/KellarisPXP14}. We also compare against a simple personalized LDP method, \solutionCMPPLDPUTotalName{} (\solutionCMPPLDPU{}), which extends \solutionCMPLDPUTotalName{} (\solutionCMPLDPU{})~\cite{DBLP:conf/sigmod/RenSYYZX22} by replacing the inner CDP mechanism with an LDP mechanism.

Let $\epsilon$ and $w$ be the privacy budget and window size in non-personalized static methods (\solutionCMPA{} and \solutionCMPB{}).
For non-personalized static methods, we set the $\epsilon$ to vary from $0.2$ to $1.0$ and $w$ to vary from $40$ to $200$.
To make our \solutionMethodA{} and \solutionMethodB{} comparable with \solutionCMPA{} and \solutionCMPB{}, we set the lower bound of each user's privacy budget as $\epsilon$ and the upper bound of each user's window size as $w$ in \solutionMethodA{} and \solutionMethodB{} to match the requirement of privacy level.

Given $\tilde{n}$ different privacy budgets $\vectorfont{\tilde{\epsilon}}=\{\epsilon_1,..., \epsilon_{\tilde{n}}\}$, let $N(\epsilon_i)$ be the count of budget value $\epsilon_i$, and  $N(\vectorfont{\tilde{\epsilon}})=\sum_{i=1}^{\tilde{n}}N(\epsilon_i)$ be the total count of all the budgets.
For any $\epsilon_i\in\vectorfont{\tilde{\epsilon}}$, we define the privacy budget ratio of $\epsilon_i$ as $\frac{N(\epsilon_i)}{N(\vectorfont{\tilde{\epsilon}})}$.
Similarly, we define the window size ratio of any $w_i$ in different window sizes $\vectorfont{\tilde{w}}=\{w_1,...,w_{\tilde{n}}\}$ as $\frac{N(w_i)}{N(\tilde{\vectorfont{w}})}$.
We set the privacy domain as $\{0.5, 1.0\}$ and the window size domain as $\{10, 20\}$.
We alter the ratio $o$ of $\epsilon_{i}=0.5$ and  $w_{i}=10$ from $0.1$ to $0.9$.

\begin{table}[t!]\vspace{-2ex}
	\begin{center}
		{\small  
			\caption{\small Experimental settings.} \label{tab:settings}
			\begin{tabular}{l|l}\hline
				{\bf \qquad Parameters \qquad \quad } & {\bf \qquad  \qquad Values \qquad } \\ \hline
				static privacy budget $\epsilon$    &       $0.2, 0.4, \textbf{0.6}, 0.8, 1.0$ \\
				static window size $w$      &       $40, 80, \textbf{120}, 160, 200$\\
				personalized privacy budget $\epsilon_i$ & $\epsilon,\ldots, 0.8, 1.0$ \\
				personalized window size $w_i$ & $40, 80,\ldots, w$\\
				users' quantity ratio 	$o$	&  0.1, 0.3, \textbf{0.5}, 0.7, 0.9\\
				\hline
			\end{tabular}
		}\vspace{-2ex}
	\end{center}
\end{table}

The parameters are shown in Table~\ref{tab:settings}, where the default values are in bold font.   
We run the experiments on an Intel(R) Xeon(R) Silver 4210R CPU @ 2.4GHz with 128 RAM in Java.
Each experiment is run 10 times, and we report the average result.

\begin{figure*}[t!]\centering\vspace{-1ex}
	\subfigure{
		\scalebox{0.33}[0.33]{\includegraphics{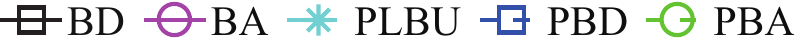}}}\hfill\\
	\addtocounter{subfigure}{-1}\vspace{-2.5ex}
	\subfigure[][{\small \trajectoryDatasetName{}}]{
		\scalebox{0.238}[0.238]{\includegraphics{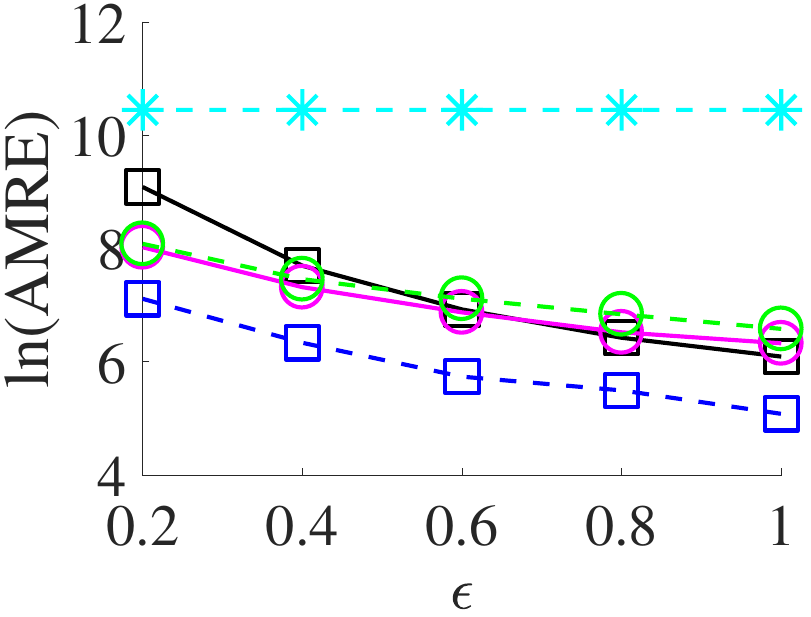}}
		\label{subfig:trajectory_budget_change}}\hfill
	\subfigure[][{\small \checkInDatasetName{}}]{
		\scalebox{0.238}[0.238]{\includegraphics{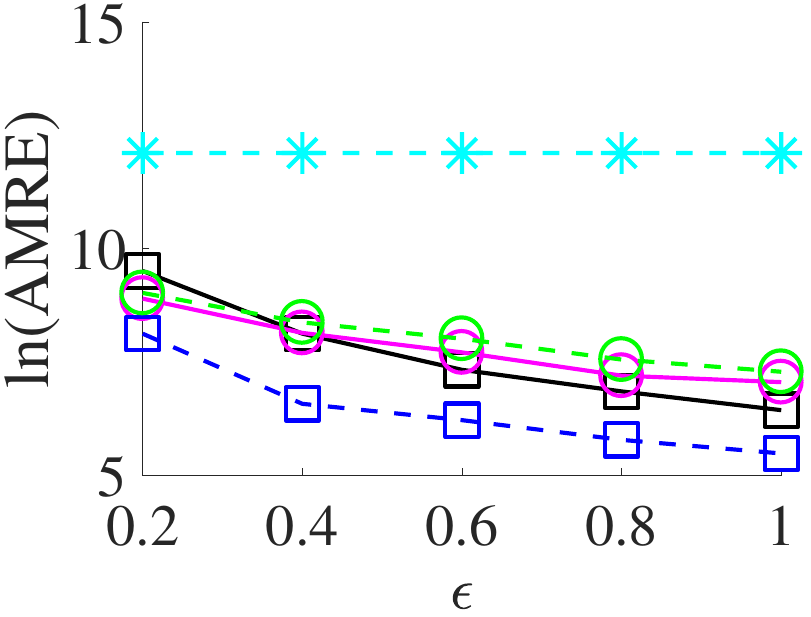}}
		\label{subfig:check_in_budget_change}}\hfill	
	\subfigure[][{\small \tlnsDatasetName{}}]{
		\scalebox{0.238}[0.238]{\includegraphics{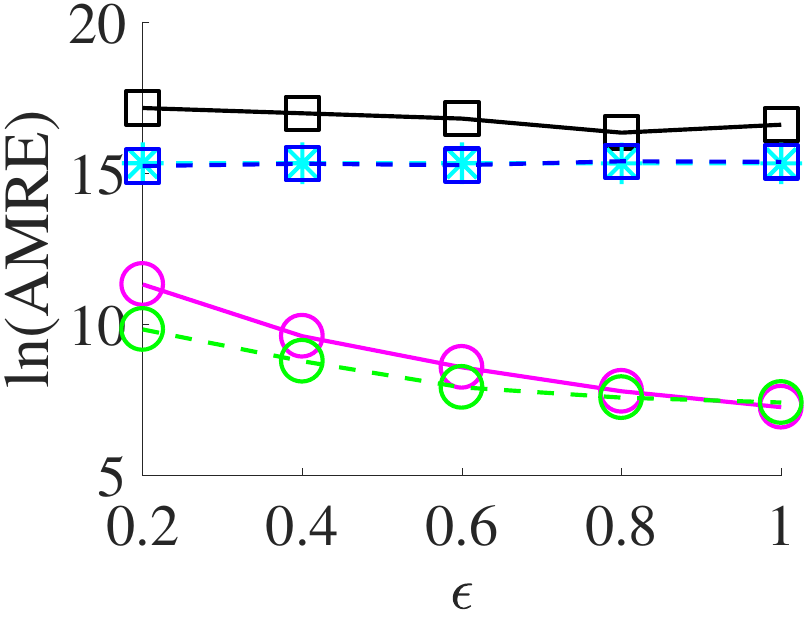}}
		\label{subfig:tlns_budget_change}}\hfill 
	\subfigure[][{\small \sinDatasetName{}}]{
		\scalebox{0.238}[0.238]{\includegraphics{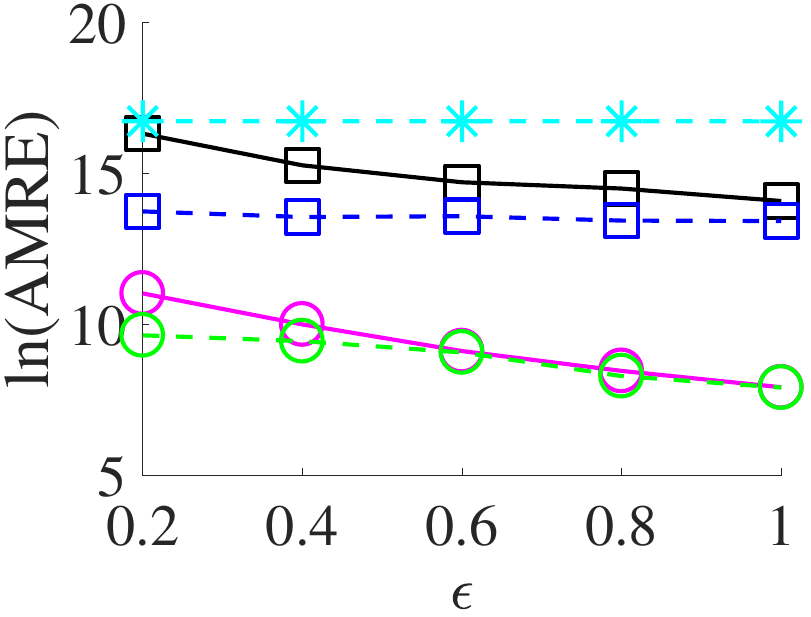}}
		\label{subfig:sin_budget_change}}\hfill 
	\subfigure[][{\small \logDatasetName{}}]{
		\scalebox{0.238}[0.238]{\includegraphics{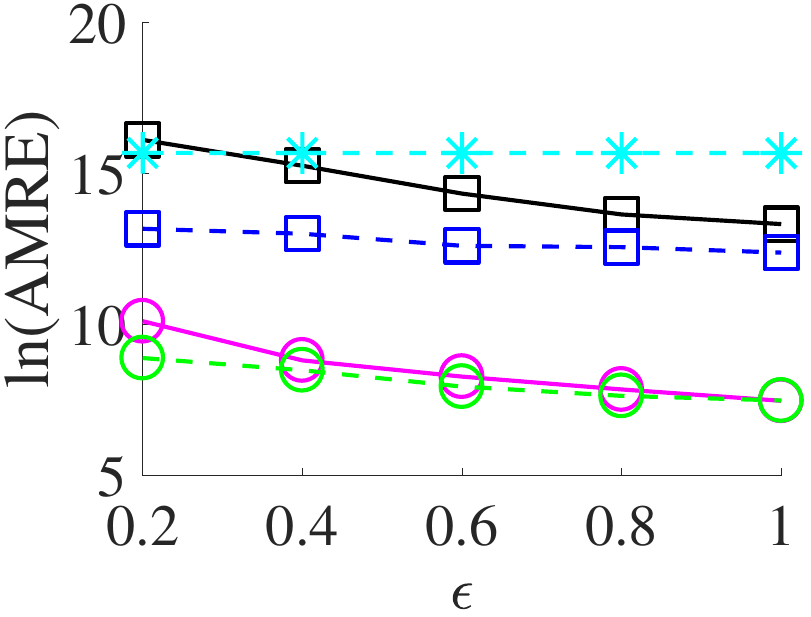}}
		\label{subfig:log_budget_change}}\hfill 
	\caption{\small $AMRE$ with $\epsilon$ varied.}
	\label{fig:alter_e}
\end{figure*}

\begin{figure*}[t!]\centering
	\subfigure{
		\scalebox{0.33}[0.33]{\includegraphics{figures/experiment_result_add/bar3.pdf}}}\hfill\\
	\addtocounter{subfigure}{-1}\vspace{-2.5ex}
	\subfigure[][{\small \trajectoryDatasetName{}}]{
		\scalebox{0.238}[0.238]{\includegraphics{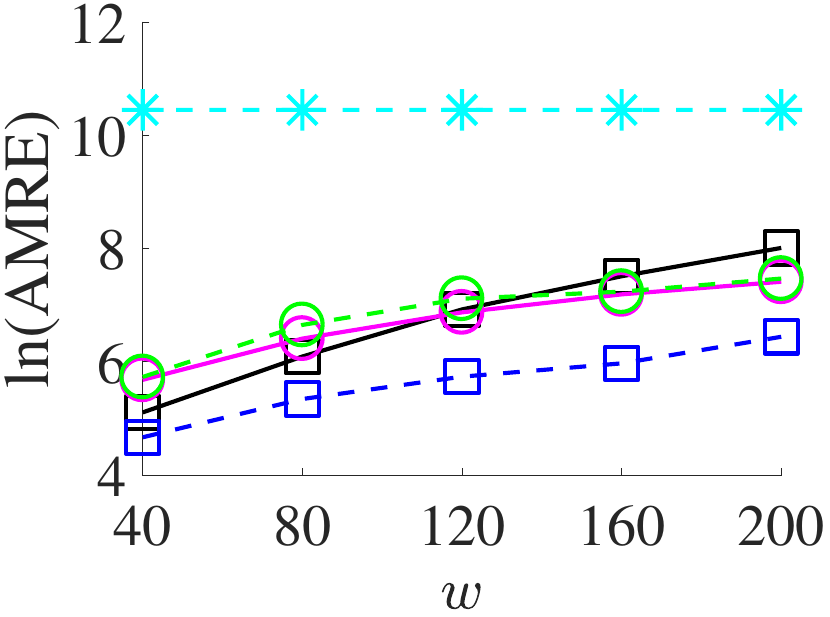}}
		\label{subfig:trajectory_window_size_change}}\hfill
	\subfigure[][{\small \checkInDatasetName{}}]{
		\scalebox{0.238}[0.238]{\includegraphics{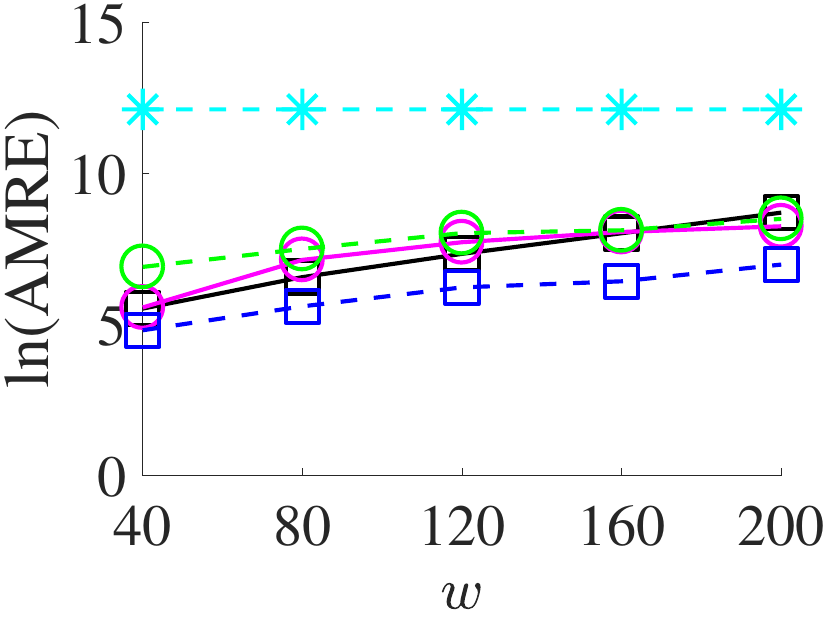}}
		\label{subfig:check_in_window_size_change}}\hfill	
	\subfigure[][{\small \tlnsDatasetName{}}]{
		\scalebox{0.238}[0.238]{\includegraphics{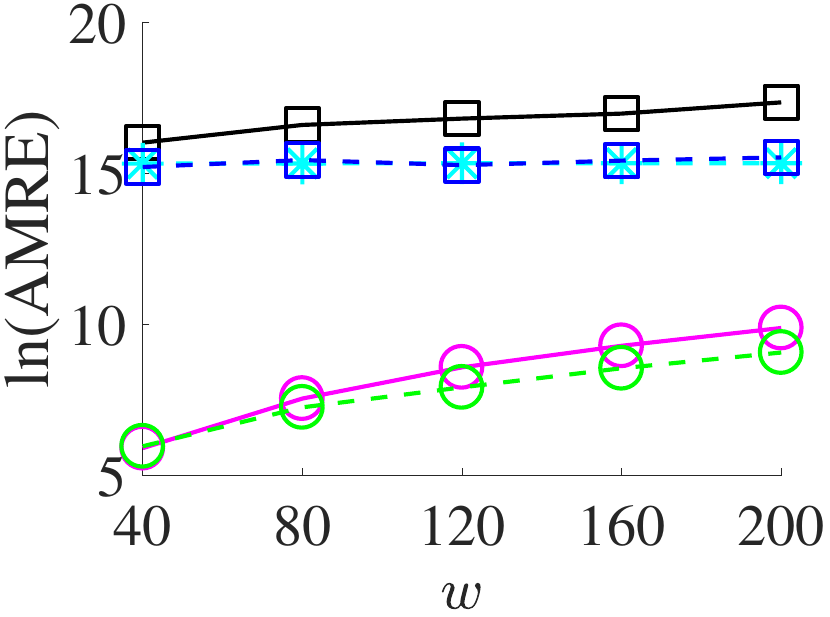}}
		\label{subfig:tlns_window_size_change}}\hfill 
	\subfigure[][{\small \sinDatasetName{}}]{
		\scalebox{0.238}[0.238]{\includegraphics{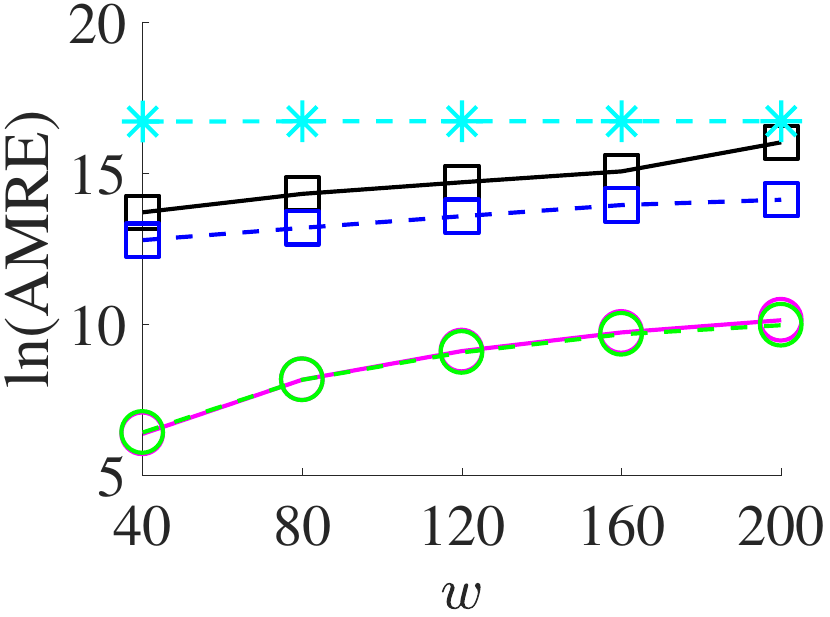}}
		\label{subfig:sin_window_size_change}}\hfill 
	\subfigure[][{\small \logDatasetName{}}]{
		\scalebox{0.238}[0.238]{\includegraphics{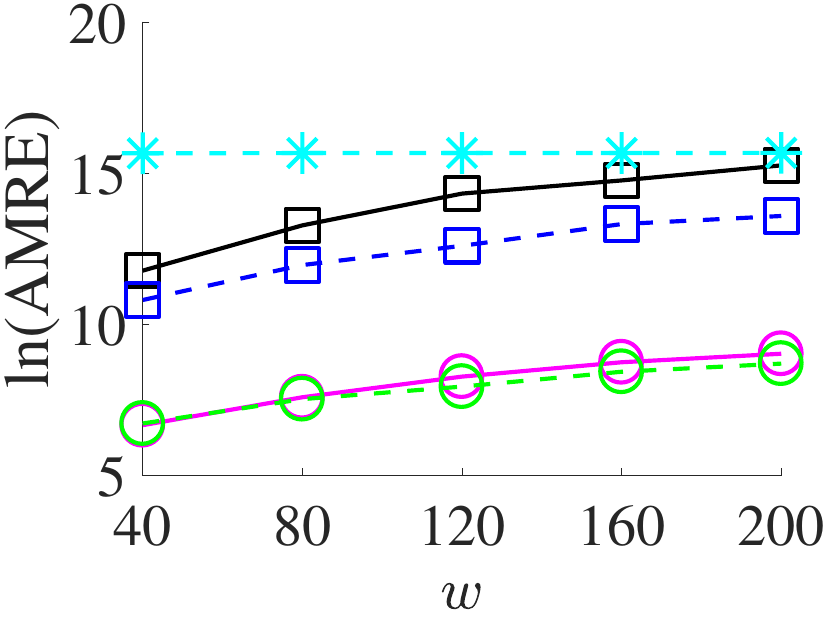}}
		\label{subfig:log_window_size_change}}\hfill 
	\caption{\small $AMRE$ with $w$ varied.}\vspace{-1ex}
	\label{fig:alter_w}
\end{figure*}

\begin{figure*}[t!]\centering
	\subfigure{
		\scalebox{0.33}[0.33]{\includegraphics{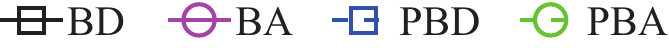}}}\hfill\\
	\addtocounter{subfigure}{-1}\vspace{-2.5ex}
	\subfigure[][{\small \trajectoryDatasetName{}}]{
		\scalebox{0.238}[0.238]{\includegraphics{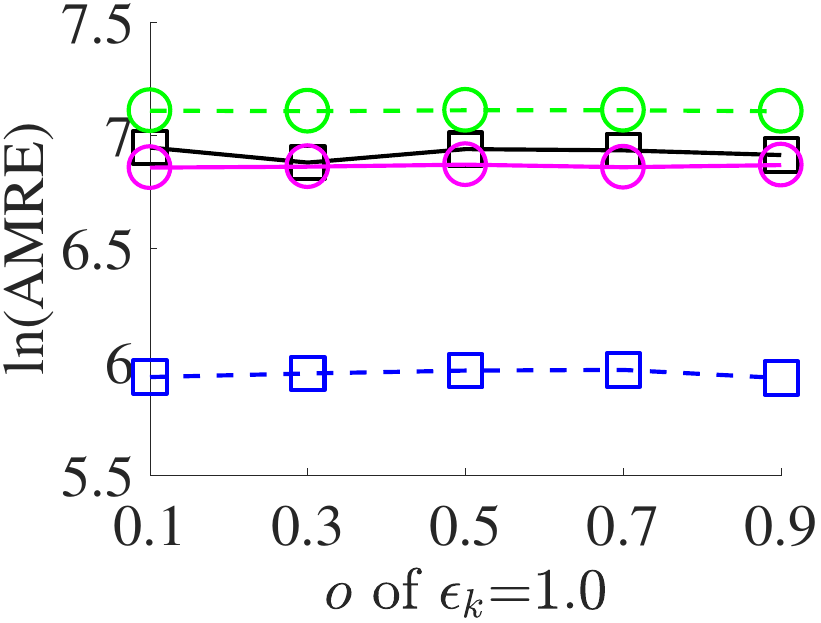}}
		\label{subfig:trajectory_ratio_change_two_budget}}\hfill
	\subfigure[][{\small \checkInDatasetName{}}]{
		\scalebox{0.238}[0.238]{\includegraphics{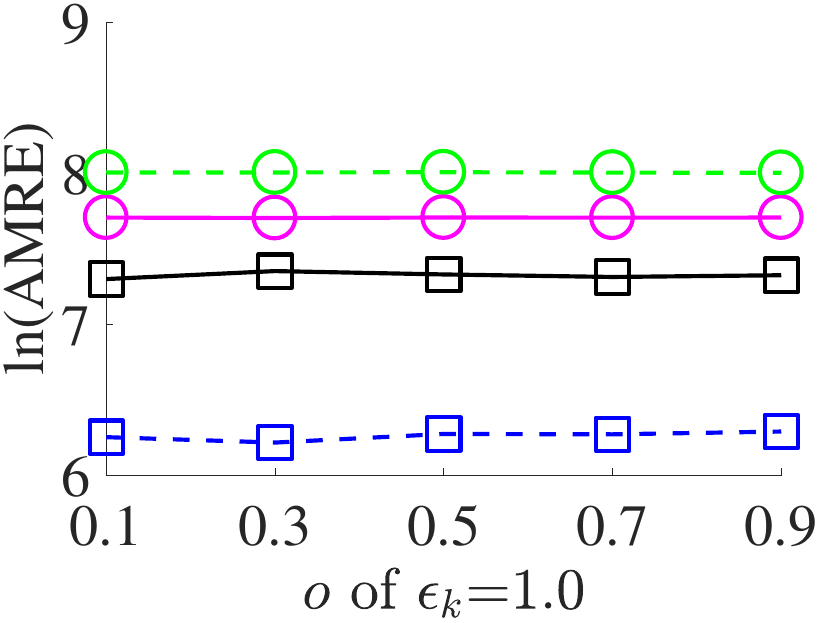}}
		\label{subfig:check_in_ratio_change_two_budget}}\hfill	
	\subfigure[][{\small \tlnsDatasetName{}}]{
		\scalebox{0.238}[0.238]{\includegraphics{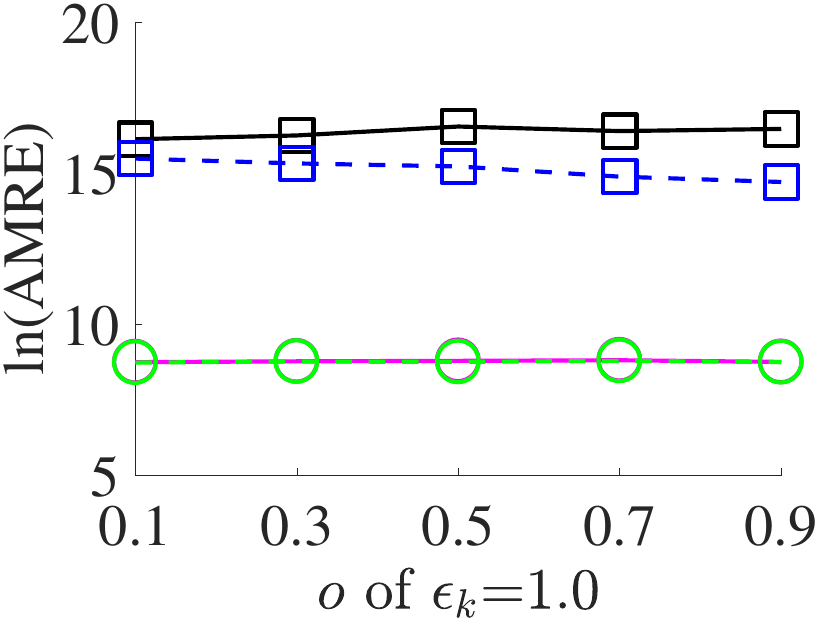}}
		\label{subfig:tlns_ratio_change_two_budget}}\hfill 
	\subfigure[][{\small \sinDatasetName{}}]{
		\scalebox{0.238}[0.238]{\includegraphics{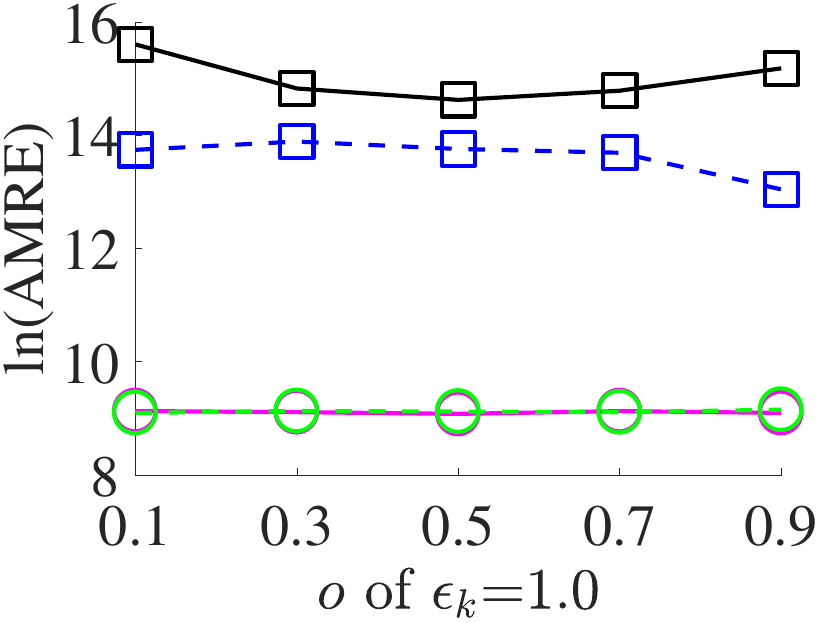}}
		\label{subfig:sin_ratio_change_two_budget}}\hfill 
	\subfigure[][{\small \logDatasetName{}}]{
		\scalebox{0.238}[0.238]{\includegraphics{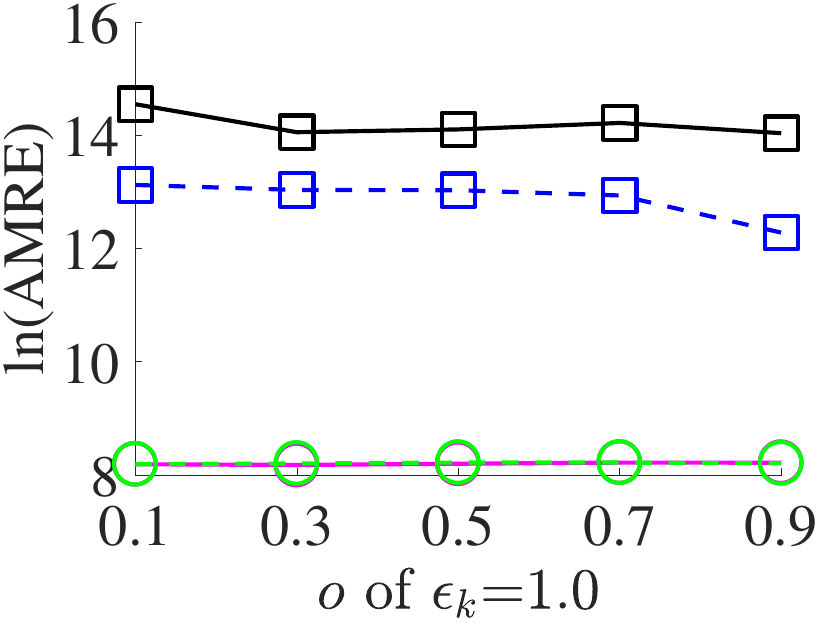}}
		\label{subfig:log_ratio_change_two_budget}}\hfill 
	\caption{\small $AMRE$ with  ratio for  privacy budget varied.}\vspace{-1ex}
	\label{fig:alter_ratio_two_budget}
\end{figure*}

\begin{figure*}[t!]\centering
	\subfigure{
		\scalebox{0.33}[0.33]{\includegraphics{figures/experiment_result/bar_2.pdf}}}\hfill\\
	\addtocounter{subfigure}{-1}\vspace{-2.5ex}
	\subfigure[][{\small \trajectoryDatasetName{}}]{
		\scalebox{0.238}[0.238]{\includegraphics{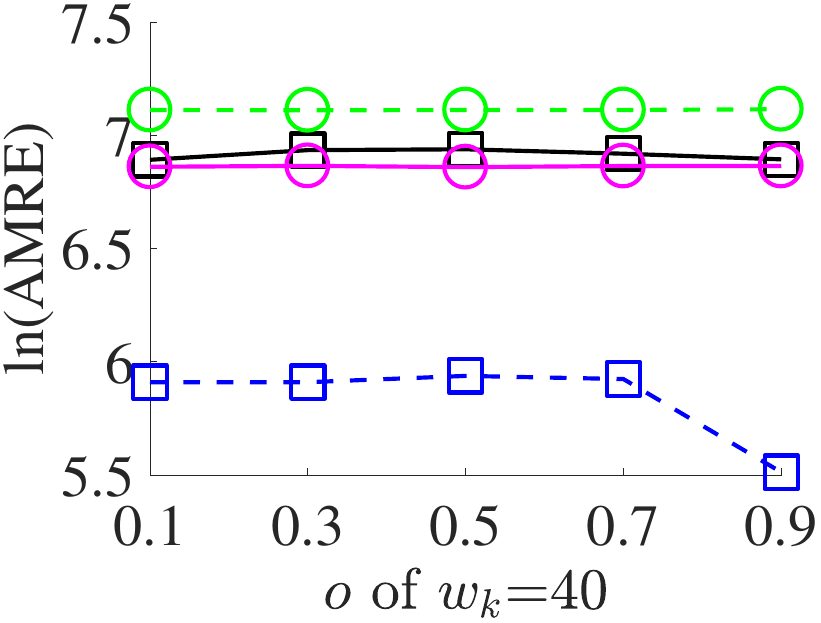}}
		\label{subfig:trajectory_ratio_change_two_window_size}}\hfill
	\subfigure[][{\small \checkInDatasetName{}}]{
		\scalebox{0.238}[0.238]{\includegraphics{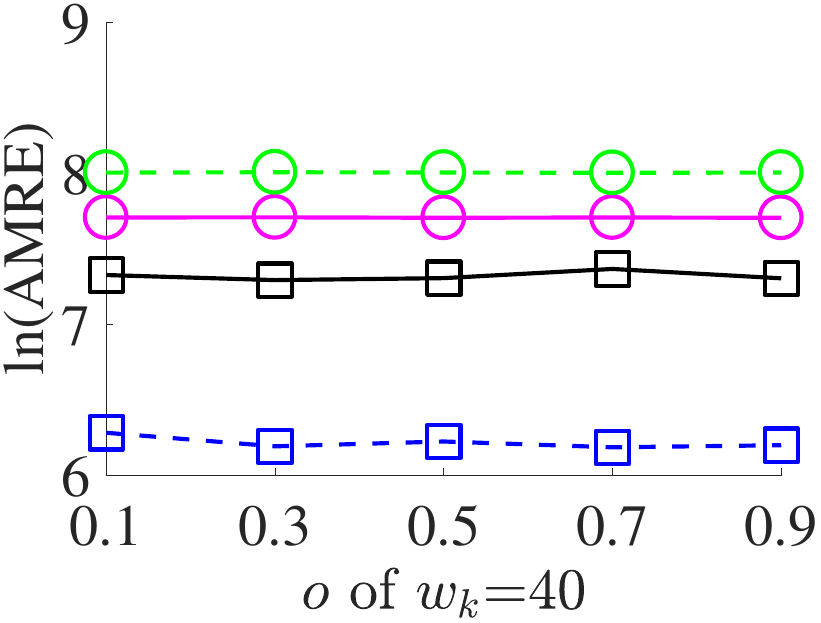}}
		\label{subfig:check_in_ratio_change_two_window_size}}\hfill	
	\subfigure[][{\small \tlnsDatasetName{}}]{
		\scalebox{0.238}[0.238]{\includegraphics{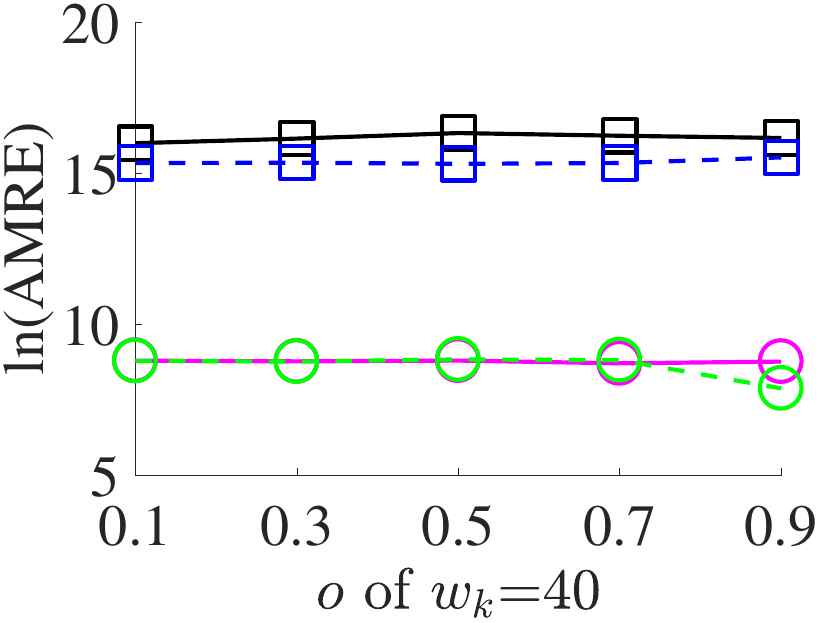}}
		\label{subfig:tlns_ratio_change_two_window_size}}\hfill 
	\subfigure[][{\small \sinDatasetName{}}]{
		\scalebox{0.238}[0.238]{\includegraphics{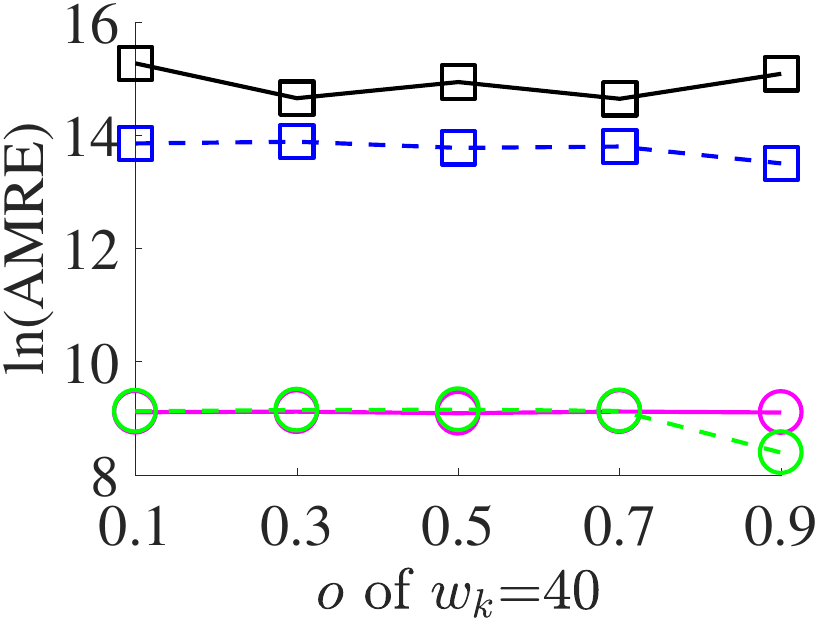}}
		\label{subfig:sin_ratio_change_two_window_size}}\hfill 
	\subfigure[][{\small \logDatasetName{}}]{
		\scalebox{0.238}[0.238]{\includegraphics{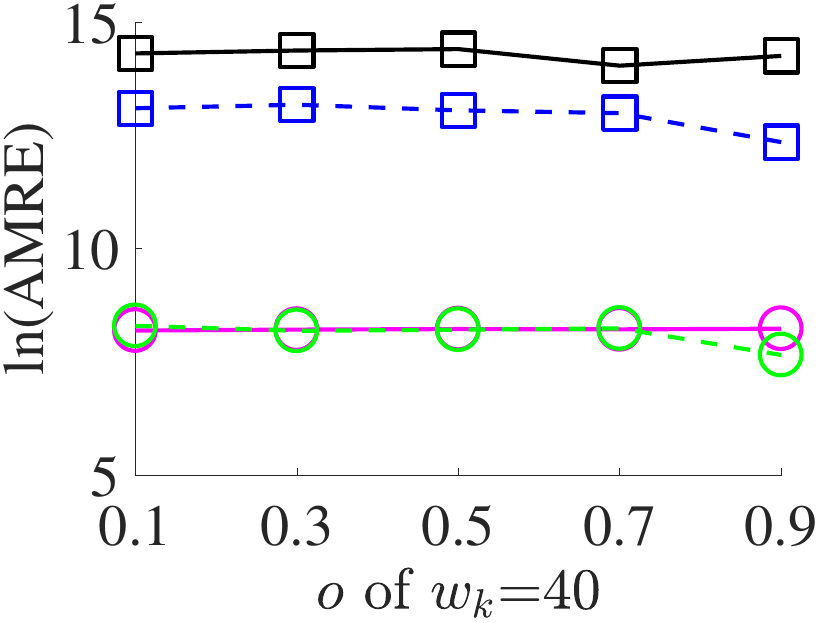}}
		\label{subfig:log_ratio_change_two_window_size}}\hfill 
	\caption{\small $AMRE$ with ratio for  window size varied.}
	\label{fig:alter_ratio_two_window_size}
\end{figure*}

\subsection{Measures}
We evaluate the performance of different mechanisms based on their running time and data utility. 
We measure data utility as \textit{Average Mean Relative Error} ($AMRE$) and \textit{Average Jensen-Shannon Divergence} ($AJSD$, $\bar{D}_{JS}$). Let $T$ represent the number of time slots and $d$ denote the dimension of data.

$AMRE$ is defined as the average value of Mean Relative Error ($MRE$), which is shown in Equation~\eqref{dataUtility}.

\begin{equation}\label{dataUtility}
	\begin{aligned}
		AMRE = \frac{1}{T}\sum_{\tau=1}^{T}MRE_{\tau} = \frac{1}{T}\sum_{\tau=1}^{T}\frac{1}{d}\|\vectorfont{r}_{\tau}-\vectorfont{c}_{\tau}\|_2^2.
	\end{aligned}
\end{equation}

$AJSD$ is defined as the average value of Jensen-Shannon Divergence ($JSD$, $D_{JS}$)~\cite{DBLP:journals/tit/Lin91}, which is based on Kullback-Leibler Divergence~\cite{kullback1951information}, as shown in Equation~\eqref{jsUtility}.
\begin{equation}\label{jsUtility}
	\begin{aligned}
		\bar{D}_{JS}(\vectorfont{r}\|\vectorfont{c}) &= \frac{1}{T}\sum_{\tau=1}^{T}D_{JS}(\vectorfont{r}\|\vectorfont{c})=\frac{1}{T}\sum_{\tau=1}^{T}\left(\frac{1}{2}D_{KL}(\vectorfont{r}\|\vectorfont{v})+\frac{1}{2}D_{KL}(\vectorfont{c}\|\vectorfont{v})\right) \\ &=\frac{1}{2T}\sum_{\tau=1}^{T}\sum_{j=1}^{d}\left(\vectorfont{r}_{\tau}(j)\log{\left(\frac{\vectorfont{r}_{\tau}(j)}{\vectorfont{v}_{\tau}(j)}\right)}+\vectorfont{c}_{\tau}(j)\log{\left(\frac{\vectorfont{c}_{\tau}(j)}{\vectorfont{v}_{\tau}(j)}\right)}\right),
	\end{aligned}
\end{equation}
where $\vectorfont{v}$ represents the average distribution of $\vectorfont{r}$ and $\vectorfont{c}$, i.e., $\vectorfont{v}(j)=\frac{1}{2}(\vectorfont{r}(j)+\vectorfont{c}(j))$.
For time slot $\tau$, $r_{\tau}(j)$ and $c_{\tau}(j)$ represent the $j$-th dimensional values in the obfuscated and original data, respectively.

\subsection{Overall Utility Analysis}\label{exp:utility}
Figure~\ref{fig:alter_e} shows the natural logarithm of $AMRE$ as the privacy budget $\epsilon$ varies.
Across all datasets, $AMRE$ decreases as $\epsilon$ increases, because a larger $\epsilon$ results in smaller noise variance, leading to a lower $AMRE$.
The decrease in $AMRE$ is more pronounced on real datasets compared to synthetic ones. 
It is because data density function changes rapidly in real datasets, while changing gradually in synthetic datasets.
When the density function changes rapidly, the dissimilarity at each time slot becomes large. In this case, \solutionMethodA{} publishes more new statistical results than \solutionMethodB{} because \solutionMethodA{} always reserves part of its privacy budget for the next time slot, even though the budget decreases over time within a window. Thus, \solutionMethodA{} leads to higher accuracy than \solutionMethodB{}.
When the density function changes gradually, the dissimilarity at each time slot remains small.
In this case, publishing one highly accurate statistical result at a time slot is more important than publishing multiple new statistical results.
Therefore, \solutionMethodB{} performs significantly better than \solutionMethodA{}.
\solutionCMPPLDPU{} performs worse than other methods across all datasets except for TLNS, since LDP methods achieve lower accuracy than CDP methods under the same privacy budget.
In real datasets, our \solutionMethodA{} consistently outperforms other methods.
The $AMRE$ of \solutionMethodA{} is on average $70.8\%$ ($17.5\%$ in terms of $\ln{(AMRE)}$) lower than that of \solutionCMPA{} on
\trajectoryDatasetName{} dataset and $69.6\%$ ($15.9\%$ in terms of $\ln{(AMRE)}$) lower on \checkInDatasetName{} dataset. 
Our \solutionMethodB{} performs slightly worse than \solutionCMPB{}, since our \solutionMethodB{} is more sensitive to noise in high-dimensional data.
For synthetic datasets, our \solutionMethodB{} consistently outperforms other methods.
Compared to \solutionCMPB{}, the $AMRE$ of \solutionMethodB{} is lower on average of $36.9\%$ ($6.0\%$ in terms of $\ln{(AMRE)}$) on \tlnsDatasetName{} dataset, $27.7\%$ ($4.2\%$ in terms of $\ln{(AMRE)}$) on \sinDatasetName{} dataset, and $28.9\%$ ($4.5\%$ in terms of $\ln{(AMRE)}$) on \logDatasetName{} dataset.
Moreover, our \solutionMethodA{} consistently outperforms \solutionCMPA{}.

Figure~\ref{fig:alter_w} shows the natural logarithm of $AMRE$ as the window size $w$ varies.
As $w$ increases, $AMRE$ rises gently, particularly on the synthetic datasets.
This occurs because a large window size results in a small privacy budget at each time slot, leading to increased error.
\solutionCMPPLDPU{} shows lower performance than other methods on all datasets except for TLNS, since LDP methods achieve lower accuracy than CDP methods under equivalent privacy budgets.
For real datasets, our \solutionMethodA{} achieves the lowest error compared to others methods.
The $AMRE$ of \solutionMethodA{} is on average $63.1\%$ ($15.6\%$ in terms of $\ln{(AMRE)}$) lower than that of \solutionCMPA{} on \trajectoryDatasetName{} dataset and $68.4\%$ ($16.5\%$ in terms of $\ln{(AMRE)}$) on \checkInDatasetName{} dataset. 
For synthetic datasets, our \solutionMethodB{} demonstrates the lowest error among all methods.
Compared to \solutionCMPB{}, the $AMRE$ of \solutionMethodB{} is lower by average of $35.1\%$ ($5.4\%$ in terms of $\ln{(AMRE)}$) for \tlnsDatasetName{}, $4.2\%$ ($0.4\%$ in terms of $\ln{(AMRE)}$) for \sinDatasetName{}, and $16.6\%$ ($2.2\%$ in terms of $\ln{(AMRE)}$) for \logDatasetName{}.
Moreover, our \solutionMethodA{} consistently outperforms \solutionCMPA{} across all datasets.

In summary, our \solutionMethodA{} demonstrates superior performance on real datasets, with $68\%$ smaller $AMRE$ on average than \solutionCMPA{}. 
For synthetic datasets, our \solutionMethodB{} outperforms \solutionCMPB{} with $24.9\%$ smaller $AMRE$ on average.

\subsection{Impact of User Requirement Type}
We define a set of users with  privacy  requirement as \textit{$(w_k,\epsilon_k)$-requirement type}.
In this subsection, we examine the impact of user type on the utility.

For the analysis, we consider $\epsilon_k\in\{0.6,1.0\}$ with a default of $0.6$, and the $w_k\in\{40,120\}$ with a default of $120$.
We first vary the users' quantity ratio of $\epsilon_k=1.0$ from $0.1$ to $0.9$ while keeping $w_k=120$. 
Then we vary the users' quantity ratio of $w_k=40$ from $0.1$ to $0.9$ while keeping $\epsilon_k=0.6$.

Figure~\ref{fig:alter_ratio_two_budget} illustrates the change in users' quantity ratio for $\epsilon_k=1.0$ from $0.1$ to $0.9$, with a fixed window size of $w_k=120$.
Figure~\ref{fig:alter_ratio_two_window_size} shows the effect on changing users' quantity for $w_k=40$ from $0.1$ to $0.9$, with a fixed privacy budget of $\epsilon_k=0.6$.
We observe that as the users' quantity ratio increases, the $AMRE$ remains relatively stable.
However, when the users' quantity ratio of $\epsilon_k=1.0$ or  $w_k=40$ exceeds $0.8$, we can see a significant decrease in $AMRE$ for \solutionMethodA{} and \solutionMethodB{}.
This occurs because when the ratios surpasses a certain threshold, the optimal budget from  OBS in Algorithm~\ref{alg:OPT_B_C} becomes dominated by a higher $\epsilon$, resulting in lower error.

\section{Conclusion}
We address the \problemDefineTotalName{} problem by proposing a mechanism called \solutionATotalName{} (\solutionA{}).
Based on \solutionA{}, we develop two methods: \solutionMethodATotalName{} (\solutionMethodA{}) and \solutionMethodBTotalName{} (\solutionMethodB{}).
We evaluate both methods against recent solutions, \solutionCMPATotalName{} (\solutionCMPA{}) and \solutionCMPBTotalName{} (\solutionCMPB{}), to demonstrate their efficiency and effectiveness. 
Our results show that \solutionMethodA{} reduces error by $68\%$ compared to \solutionCMPA{} on real datasets, while \solutionMethodB{} achieves $24.9\%$ lower error than \solutionCMPB{} on synthetic datasets.

\section{acknowledgment}
Peng Cheng's work is partially supported by NSFC under Grant No. 62102149 and the Fundamental Research Funds for the Central Universities. Lei Chen’s work is partially supported by National Key Research and Development Program of China Grant No. 2023YFF0725100, National Science Foundation of China (NSFC) under Grant No. U22B2060, Guangdong-Hong Kong Technology Innovation Joint Funding Scheme Project No. 2024A0505040012, the Hong Kong RGC GRF Project 16213620, RIF Project R6020-19, AOE Project AoE/E-603/18, 
Theme-based project TRS T41-603/20R, CRF Project C2004-21G, Guangdong Province Science and Technology Plan Project 2023A0505030011, Guangzhou municipality big data intelligence key lab, 2023A03J0012, Hong Kong ITC ITF grants MHX/078/21 and PRP/004/22FX, Zhujiang scholar program 2021JC02X170, Microsoft Research Asia Collaborative Research Grant, 
HKUST-Webank joint research lab and 2023 HKUST Shenzhen-Hong Kong Collaborative Innovation Institute Green Sustainability Special Fund, from Shui On Xintiandi and the InnoSpace GBA. Heng Tao Shen is supported by the Fundamental Research Funds for the Central Universities. Xuemin Lin is supported by NSFC U2241211. 
Corresponding author: Peng Cheng.

\bgroup\small
\bibliographystyle{ACM-Reference-Format}
\let\xxx=\bibitem\def\bibitem{\par\vspace{1mm}\xxx}
\bibliography{reference}
\egroup

\section{Appendix}
\subsection{Running time Analysis}
	In this subsection, we compare the running time of \solutionCMPA{}, \solutionCMPB{}, \solutionMethodA{} and \solutionMethodB{}.

	Figure~\ref{fig:alter_e_running_time} shows the running time as the privacy budget varies from $0.2$ to $1$. For synthetic datasets, the running time remains stable across different privacy budgets. This stability occurs because as datasets change gradually and skipped time slots increase, different privacy budgets have minimal impact on the number of new publications.
	In real datasets, the running time of \solutionCMPB{} increases slightly, likely due to larger privacy budgets requiring more comparisons between dissimilarity and error when datasets change rapidly. 
	\solutionMethodA{} requires the highest computation time among all methods, particularly with synthetic datasets, while \solutionCMPA{} requires the least time for real datasets and some synthetic datasets (i.e., TLNS and Sin). It is because non-personalized methods (\solutionCMPA{} and \solutionCMPB{}) have fewer steps in \solutionCMPA{} than in \solutionCMPB{} for publication judgments (which is denoted as algorithm complexity running time, $TC_{ac}$). As a result, \solutionCMPA{} requires less time than \solutionCMPB{} for most datasets.
	Personalized methods (\solutionMethodA{} and \solutionMethodB{}), however, require additional steps for optimal budget selection. 
	These methods also have a higher probability of dissimilarity exceeding error (since they achieve lower error rates than non-personalized methods), resulting in fewer skips or nullifications compared to non-personalized methods. 
	Fewer skips or nullifications leads to more comparisons in the total stream publication (which is defined as publication number running time, $TC_{pn}$).
	When time slots are sufficiently large, $TC_{pn}$ has a greater impact than $TC_{ac}$.
	This effect becomes particularly noticeable with data changing slowly (as seen in synthetic datasets).

	Figure~\ref{fig:alter_w_running_time} shows the running time as the window size changes from $40$ to $200$. All methods except \solutionMethodA{} maintain stable running times as the window size increases. For \solutionMethodA{}, its running time increases when the window sizes increase, because larger windows result in smaller per-user privacy budgets within each window. Reserving half of the privacy budget for future publications leads to larger dissimilarity and error.
	Since \solutionMethodA{}'s optimal budget selection step mitigates error's growth, the dissimilarity increases at a lower rate than the error, resulting in more frequent publications.
	Similar to Figure~\ref{fig:alter_e_running_time}, \solutionCMPA{} requires the least running time among all methods on real datasets and most synthetic datasets (i.e., TLNS and Sin), while \solutionMethodA{} requires the highest running time. It is because $TC_{ac}$ is lower in \solutionCMPA{} than in \solutionCMPB{}. 
	Additionally, personalized methods introduce an optimal budget selection step that increases the running time by $TC_{pn}$. Compared to \solutionMethodB{}, the dissimilarity of \solutionMethodA{} increases more rapidly than the error, resulting in fewer skips or nullifications and thus a larger $TC_{pn}$ in \solutionMethodA{}. When time slots are sufficiently large, $TC_{pn}$ has a greater impact than $TC_{ac}$, causing \solutionMethodA{} to have the longest running time.

\begin{figure*}[h]\centering
	\subfigure{
		\scalebox{0.33}[0.33]{\includegraphics{figures/experiment_result/bar_2.pdf}}}\hfill\\
	\addtocounter{subfigure}{-1}\vspace{-2ex}
	\subfigure[][{\small \trajectoryDatasetName{}}]{
		\scalebox{0.238}[0.238]{\includegraphics{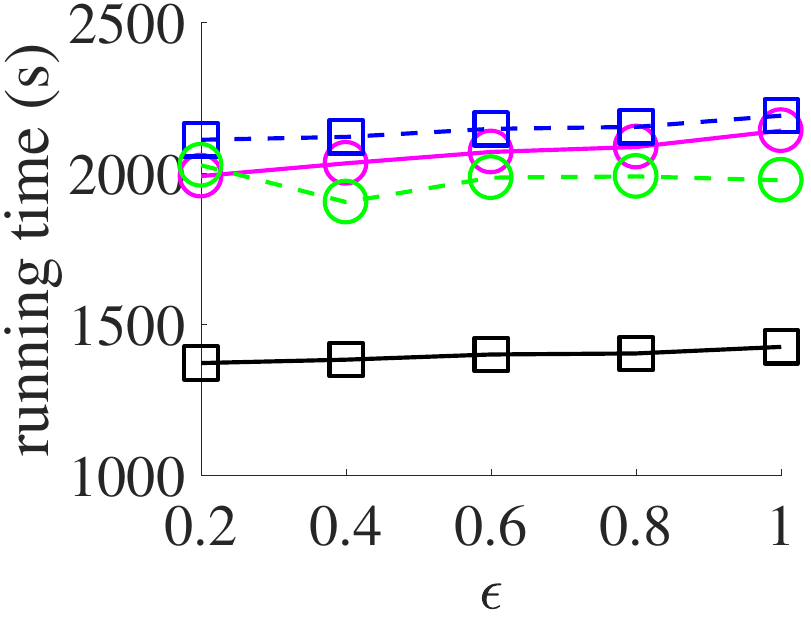}}
		\label{subfig:trajectory_budget_change_running_time}}\hfill
	\subfigure[][{\small \checkInDatasetName{}}]{
		\scalebox{0.238}[0.238]{\includegraphics{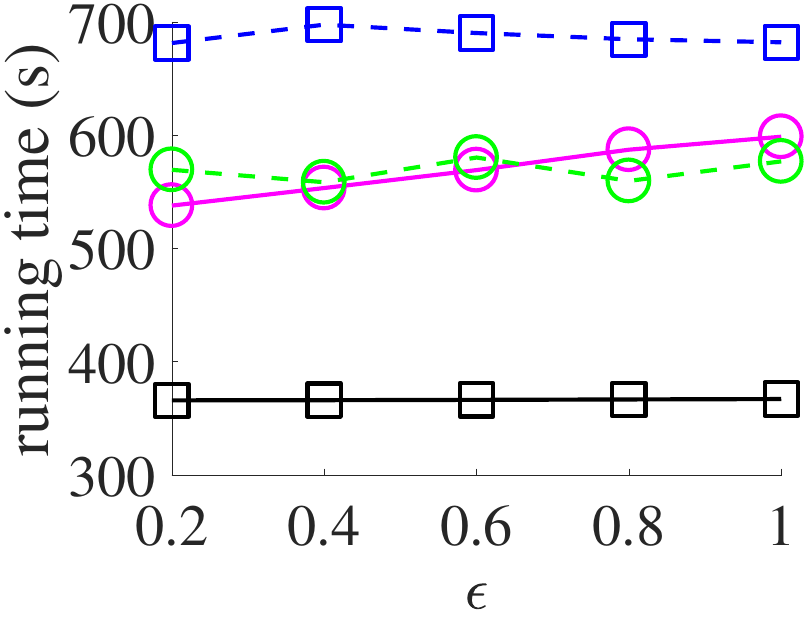}}
		\label{subfig:check_in_budget_change_running_time}}\hfill	
	\subfigure[][{\small \tlnsDatasetName{}}]{
		\scalebox{0.238}[0.238]{\includegraphics{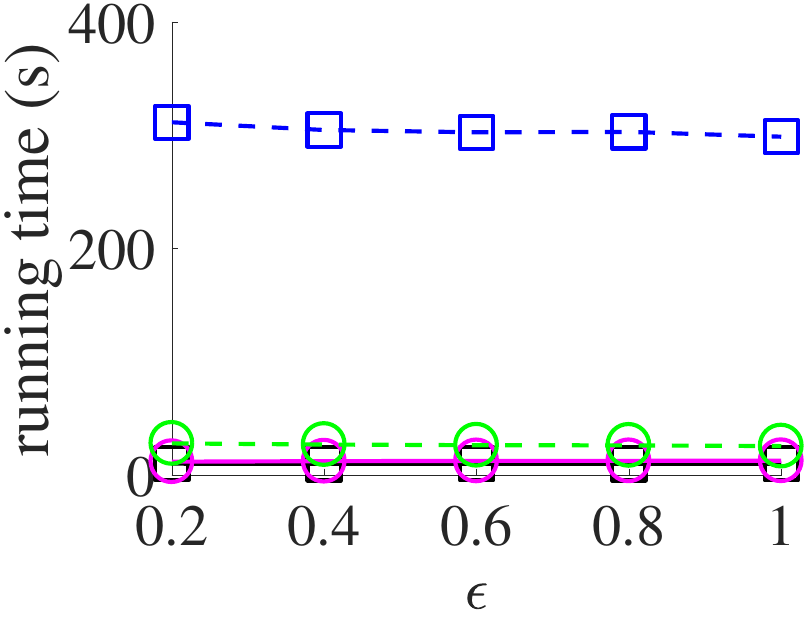}}
		\label{subfig:tlns_budget_change_running_time}}\hfill 
	\subfigure[][{\small \sinDatasetName{}}]{
		\scalebox{0.238}[0.238]{\includegraphics{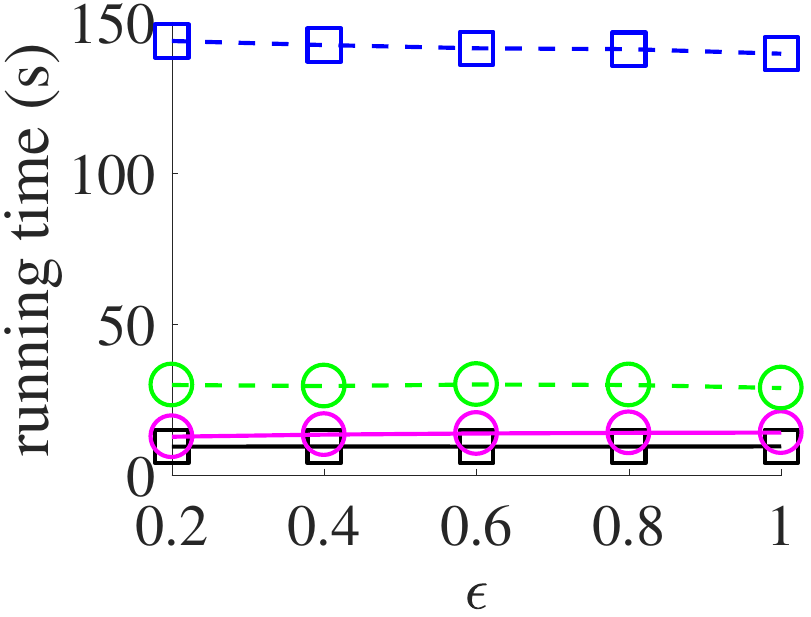}}
		\label{subfig:sin_budget_change_running_time}}\hfill 
	\subfigure[][{\small \logDatasetName{}}]{
		\scalebox{0.238}[0.238]{\includegraphics{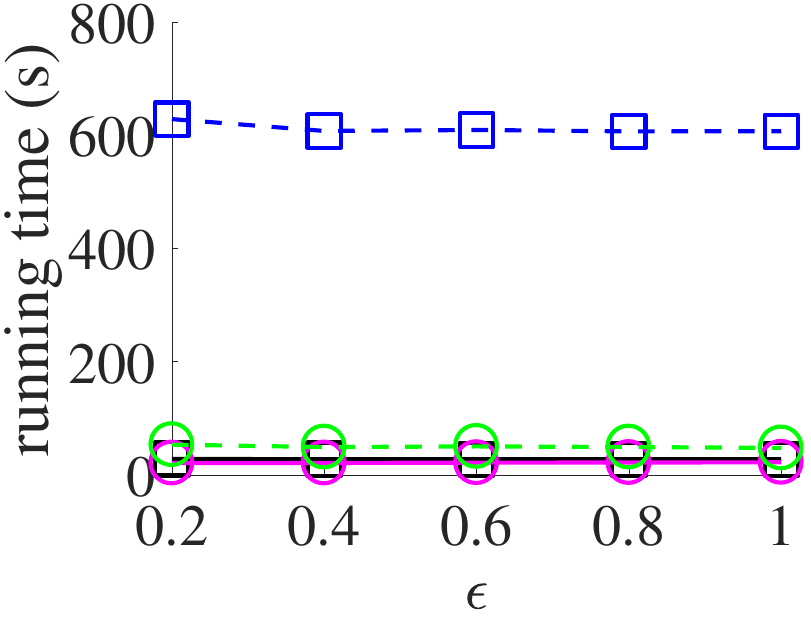}}
		\label{subfig:log_budget_change_running_time}}\hfill 
	\caption{\small The running time with $\epsilon$ varied.}
	\label{fig:alter_e_running_time}
\end{figure*}

\begin{figure*}[h]\centering
	\subfigure{
		\scalebox{0.33}[0.33]{\includegraphics{figures/experiment_result/bar_2.pdf}}}\hfill\\
	\addtocounter{subfigure}{-1}\vspace{-2ex}
	\subfigure[][{\small \trajectoryDatasetName{}}]{
		\scalebox{0.238}[0.238]{\includegraphics{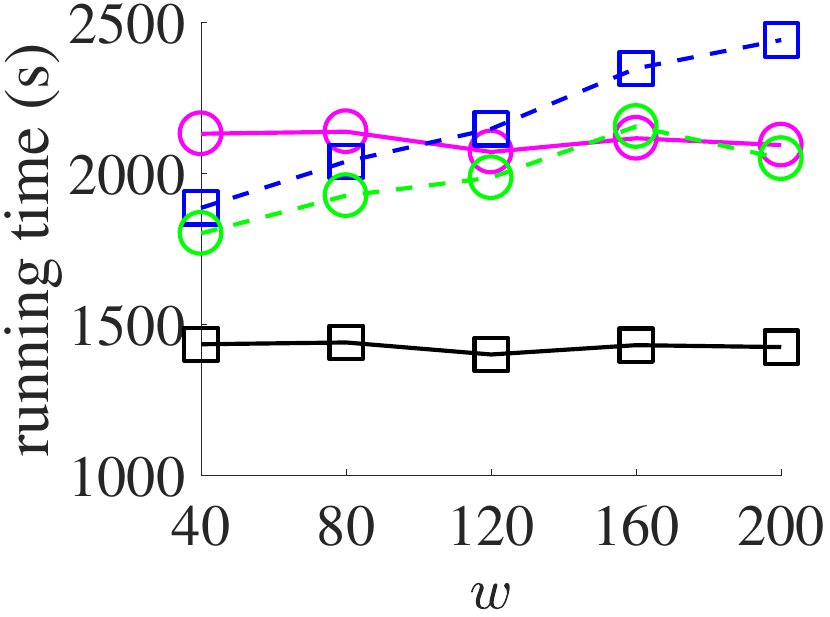}}
		\label{subfig:trajectory_window_size_change_running_time}}\hfill
	\subfigure[][{\small \checkInDatasetName{}}]{
		\scalebox{0.238}[0.238]{\includegraphics{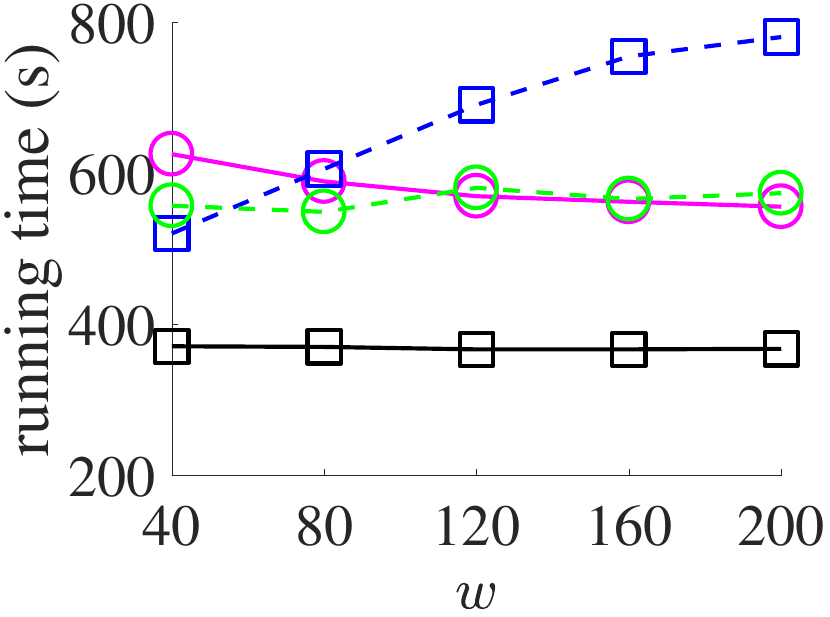}}
		\label{subfig:check_in_window_size_change_running_time}}\hfill	
	\subfigure[][{\small \tlnsDatasetName{}}]{
		\scalebox{0.238}[0.238]{\includegraphics{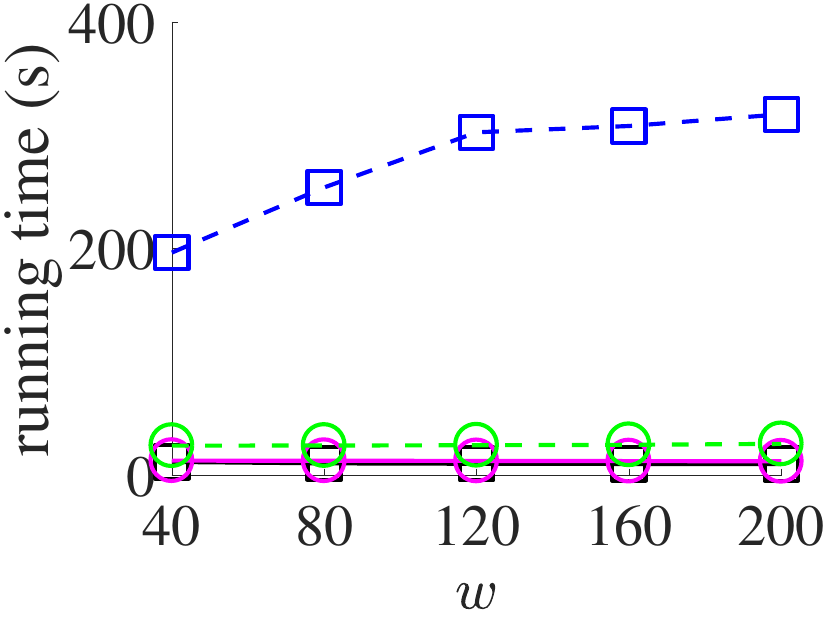}}
		\label{subfig:tlns_window_size_change_running_time}}\hfill 
	\subfigure[][{\small \sinDatasetName{}}]{
		\scalebox{0.238}[0.238]{\includegraphics{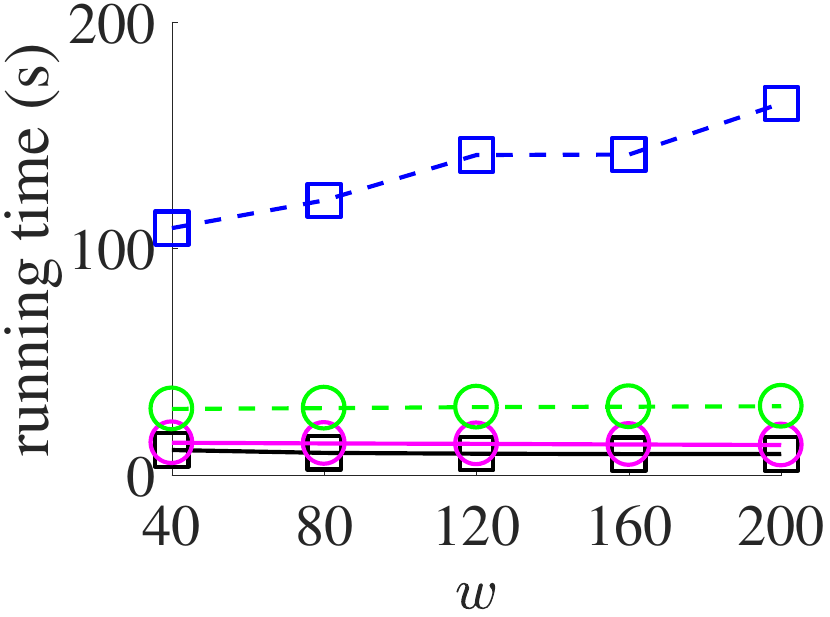}}
		\label{subfig:sin_window_size_change_running_time}}\hfill 
	\subfigure[][{\small \logDatasetName{}}]{
		\scalebox{0.238}[0.238]{\includegraphics{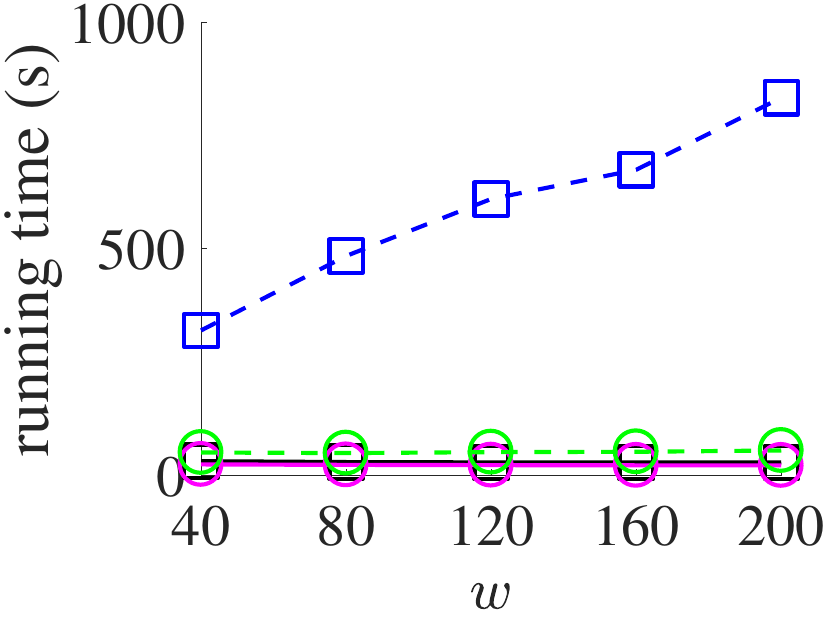}}
		\label{subfig:log_window_size_change_running_time}}\hfill 
	\caption{\small The running time with $w$ varied.}
	\label{fig:alter_w_running_time}
\end{figure*}

\subsection{Experimental Result under $AJSD$ Metric}\label{appendix:under_AJSD_metric}
	In this subsection, we compare the performance of \solutionCMPA{}, \solutionCMPB{}, \solutionCMPPLDPU{},  \solutionMethodA{} and \solutionMethodB{} using $AJSD$ metric.

\begin{figure*}[]\centering
	\subfigure{
		\scalebox{0.33}[0.33]{\includegraphics{figures/experiment_result_add/bar3.pdf}}}\hfill\\
	\addtocounter{subfigure}{-1}\vspace{-2ex}
	\subfigure[][{\small \trajectoryDatasetName{}}]{
		\scalebox{0.238}[0.238]{\includegraphics{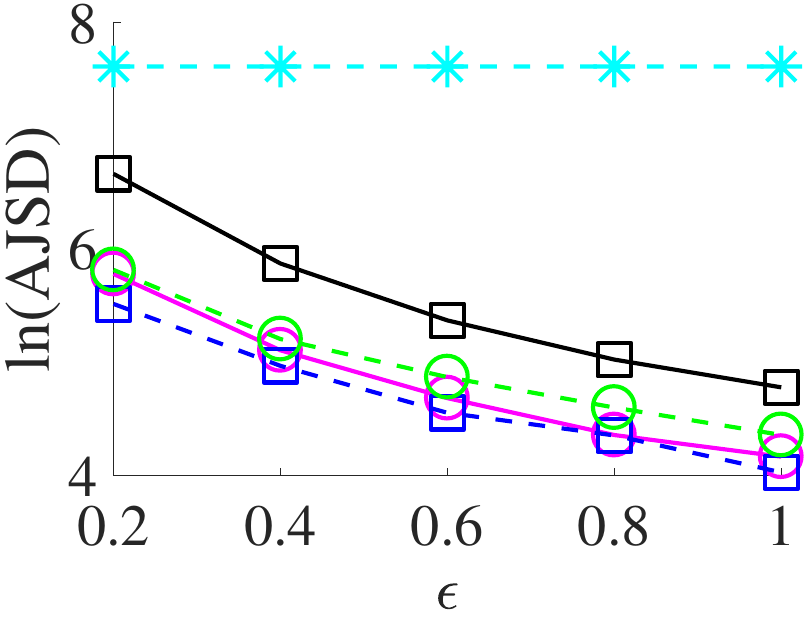}}
		\label{subfig:trajectory_budget_change_AJSD}}\hfill
	\subfigure[][{\small \checkInDatasetName{}}]{
		\scalebox{0.238}[0.238]{\includegraphics{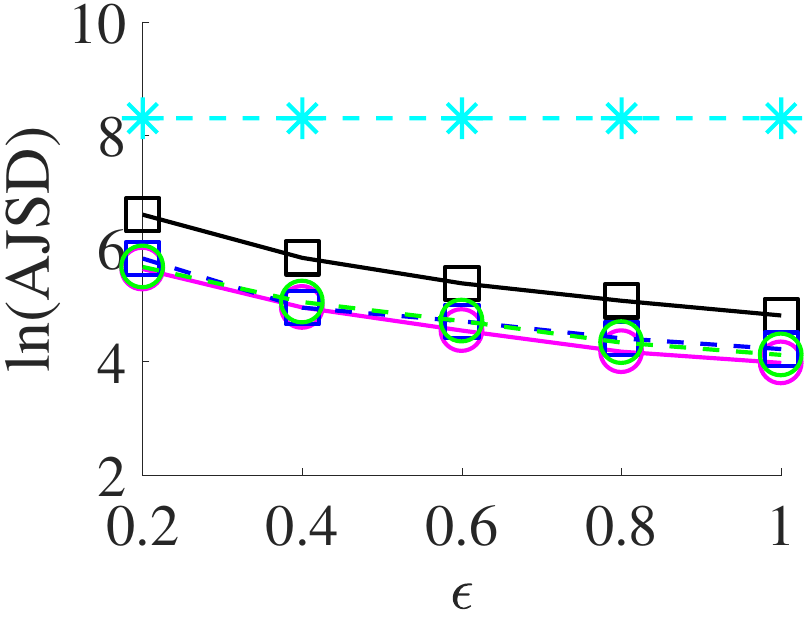}}
		\label{subfig:check_in_budget_change_AJSD}}\hfill	
	\subfigure[][{\small \tlnsDatasetName{}}]{
		\scalebox{0.238}[0.238]{\includegraphics{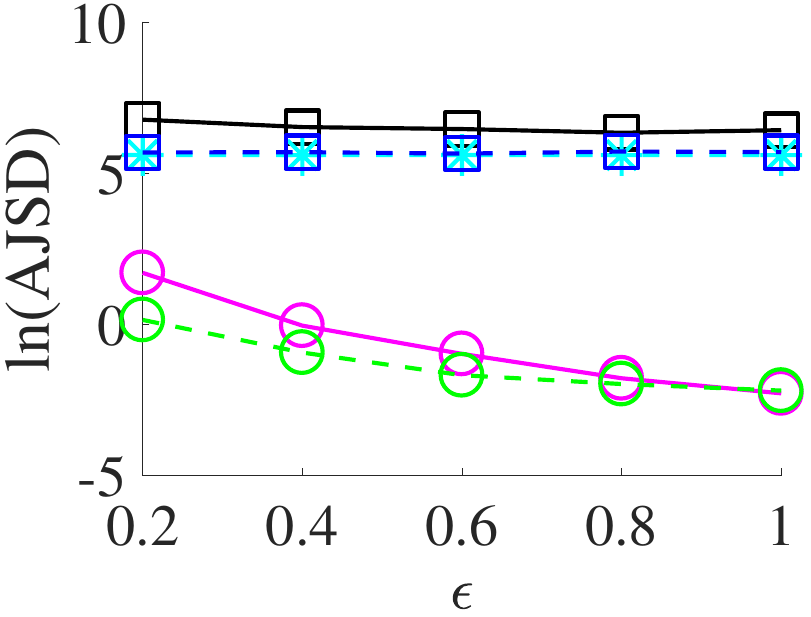}}
		\label{subfig:tlns_budget_change_AJSD}}\hfill 
	\subfigure[][{\small \sinDatasetName{}}]{
		\scalebox{0.238}[0.238]{\includegraphics{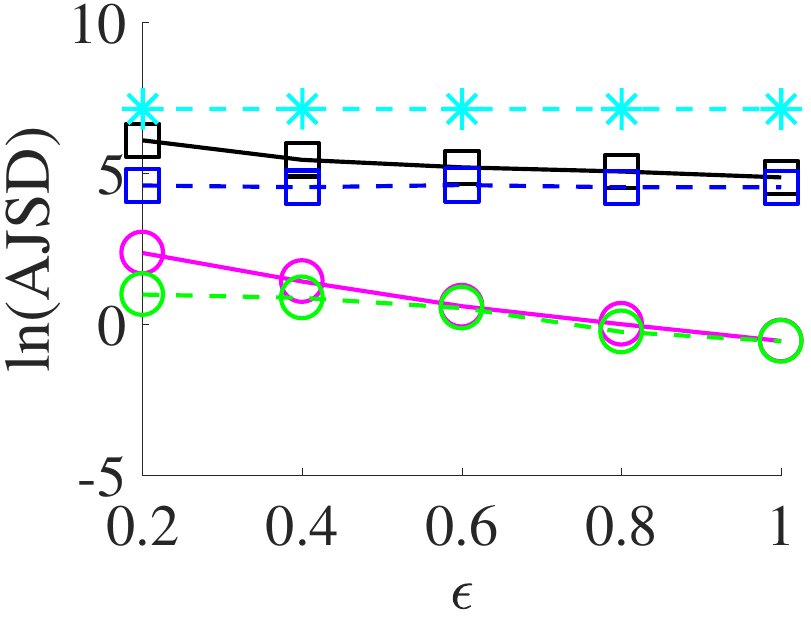}}
		\label{subfig:sin_budget_change_AJSD}}\hfill 
	\subfigure[][{\small \logDatasetName{}}]{
		\scalebox{0.238}[0.238]{\includegraphics{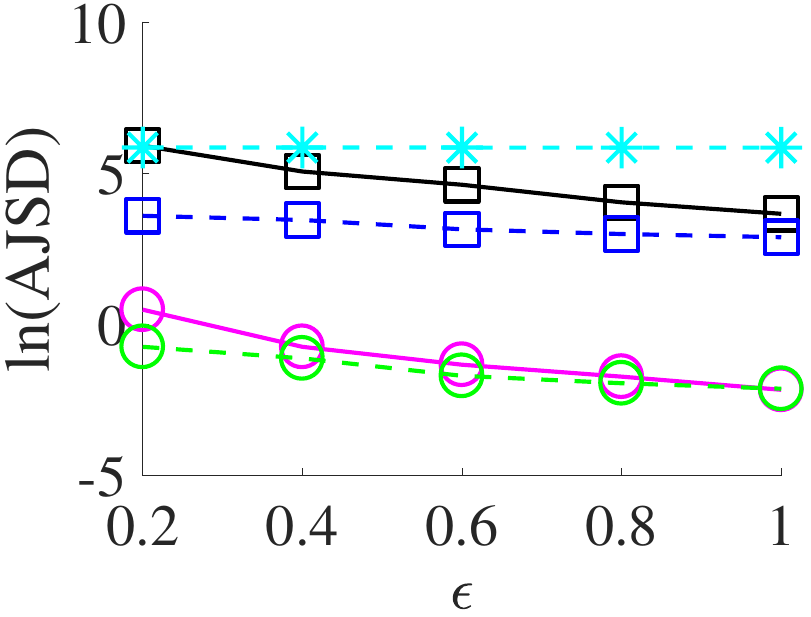}}
		\label{subfig:log_budget_change_AJSD}}\hfill 
	\caption{\small The $AJSD$ with $\epsilon$ varied.}
	\label{fig:alter_e_AJSD}
\end{figure*}

\begin{figure*}[h]\centering
	\subfigure{
		\scalebox{0.33}[0.33]{\includegraphics{figures/experiment_result_add/bar3.pdf}}}\hfill\\
	\addtocounter{subfigure}{-1}\vspace{-2ex}
	\subfigure[][{\small \trajectoryDatasetName{}}]{
		\scalebox{0.238}[0.238]{\includegraphics{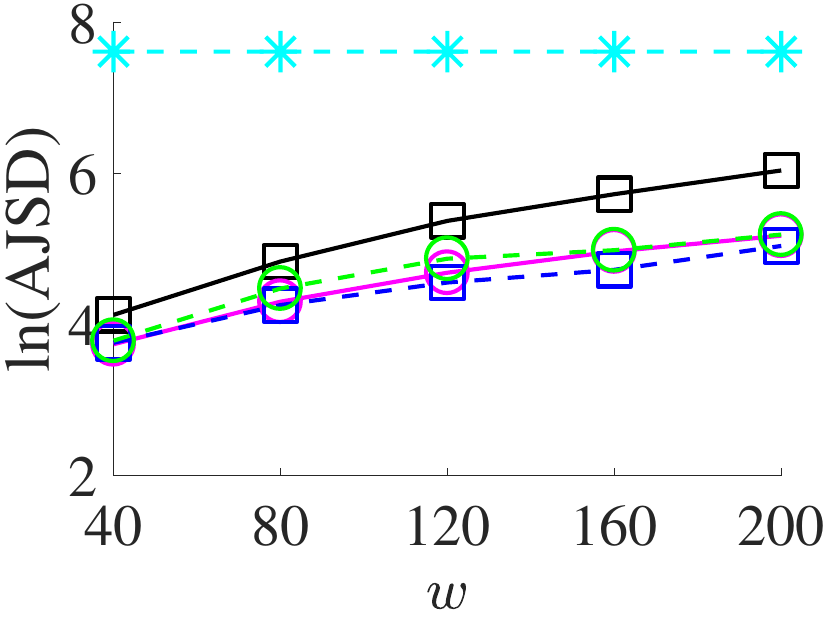}}
		\label{subfig:trajectory_window_size_change_AJSD}}\hfill
	\subfigure[][{\small \checkInDatasetName{}}]{
		\scalebox{0.238}[0.238]{\includegraphics{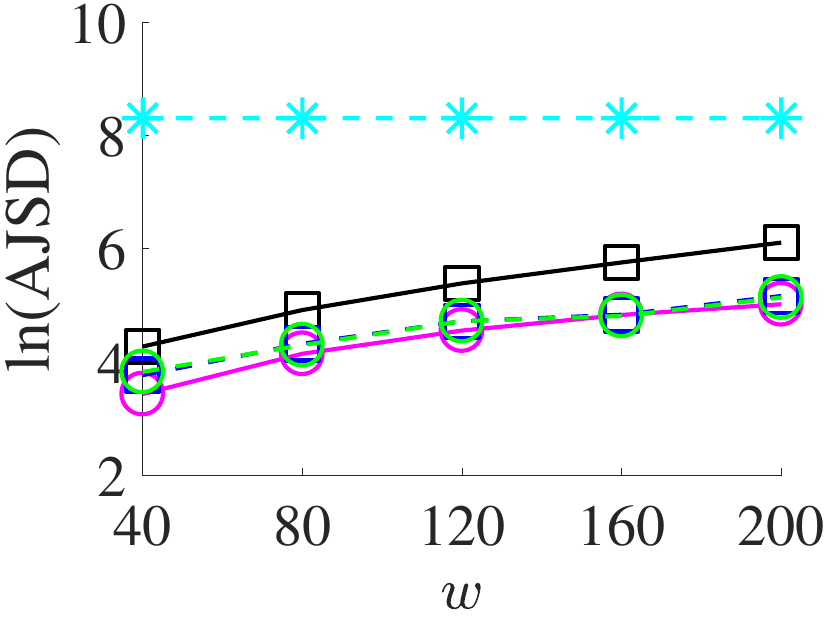}}
		\label{subfig:check_in_window_size_change_AJSD}}\hfill	
	\subfigure[][{\small \tlnsDatasetName{}}]{
		\scalebox{0.238}[0.238]{\includegraphics{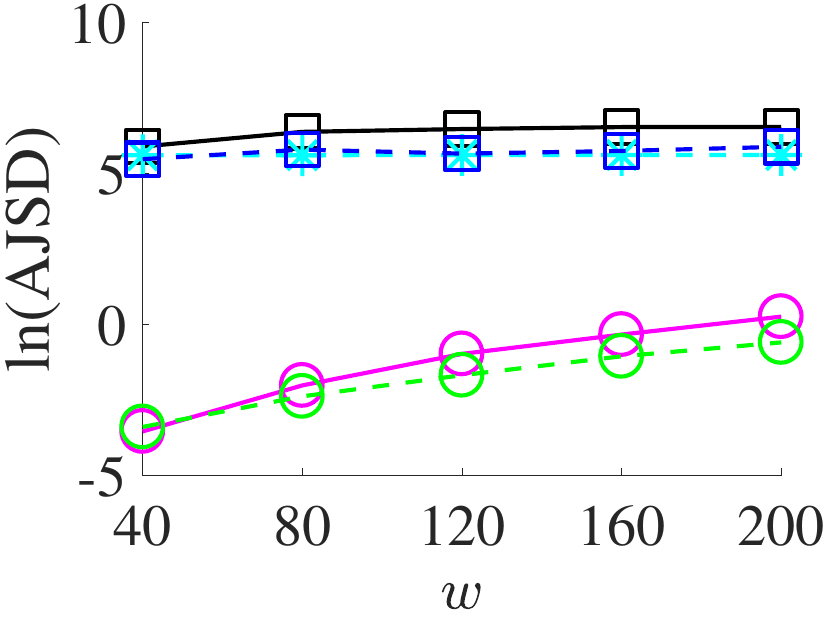}}
		\label{subfig:tlns_window_size_change_AJSD}}\hfill 
	\subfigure[][{\small \sinDatasetName{}}]{
		\scalebox{0.238}[0.238]{\includegraphics{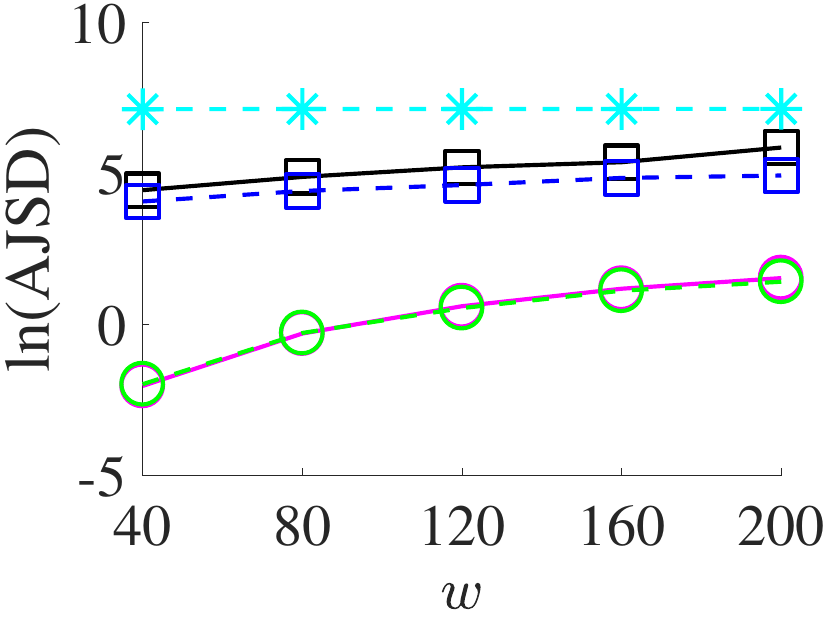}}
		\label{subfig:sin_window_size_change_AJSD}}\hfill 
	\subfigure[][{\small \logDatasetName{}}]{
		\scalebox{0.238}[0.238]{\includegraphics{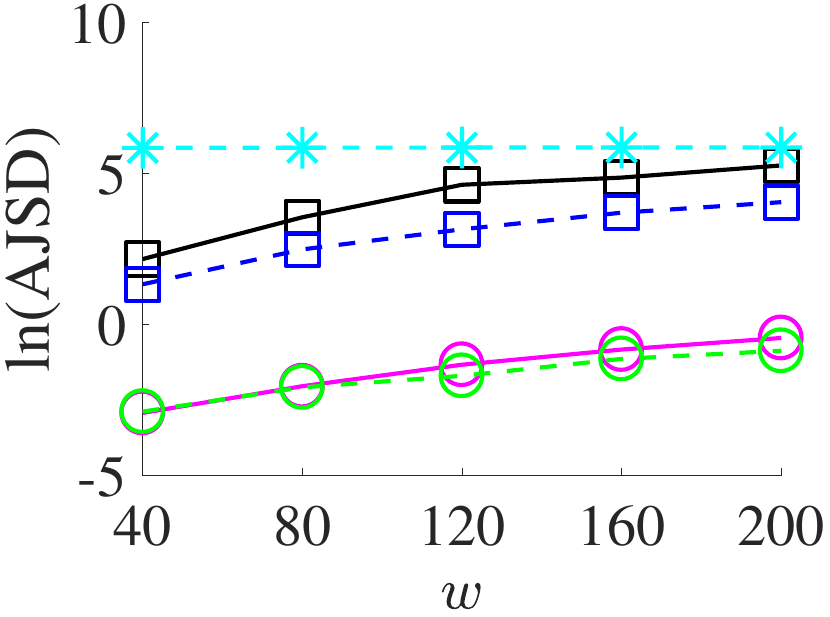}}
		\label{subfig:log_window_size_change_AJSD}}\hfill 
	\caption{\small The $AJSD$ with $w$ varied.}
	\label{fig:alter_w_AJSD}
\end{figure*}

	Figure~\ref{fig:alter_e_AJSD} shows the results of $AJSD$ as the privacy budget $\epsilon$ varies from $0.2$ to $1$.
	For all methods, $AJSD$ decreases as $\epsilon$ increases, which aligns with the $AMRE$ results in Section~\ref{exp:utility}.
	\solutionCMPPLDPU{} performs worse than other methods across all datasets except TLNS as LDP methods achieve lower accuracy than CDP methods under the same privacy budget.
	Both \solutionMethodA{} and \solutionMethodB{} consistently outperform \solutionCMPA{}.
	\solutionMethodA{} achieves the best accuracy on the Taxi dataset, while \solutionMethodB{} performs best with the three synthetic datasets. In the Foursquare dataset, \solutionCMPB{} outperforms other methods, due to the dataset's sparsity causing larger error in $AJSD$ calculation.

	Figure~\ref{fig:alter_w_AJSD} shows the results of $AJSD$ as the window size $w$ varies from $20$ to $200$.
	$AJSD$ also increases with larger window sizes for all methods.
	\solutionCMPPLDPU{} shows lower utility than other methods in all datasets except TLNS, since LDP methods achieve lower accuracy than CDP methods under equivalent privacy budgets.
	Consistent with the results in Figure~\ref{fig:alter_e_AJSD}, both \solutionMethodA{} and \solutionMethodB{} outperform \solutionCMPA{}.
	\solutionMethodA{} achieves the best performance in the Taxi dataset, while \solutionMethodB{} leads in the three synthetic datasets. 
	For the Foursquare dataset, \solutionCMPB{} achieves the lowest $AJSD$. However, this result may be unreliable due to the dataset's sparsity, which causes large calculation errors in the $AJSD$ measurements.

\end{document}